\documentclass[10pt]{article}
\usepackage {amsmath, amssymb,amsthm, graphicx, mathrsfs}
\usepackage{amsmath, amssymb, amsthm, authblk, bigints}
\theoremstyle {plain}
\newtheorem {Thrm}{Theorem}[section]
\newtheorem {Lem}{Lemma}[section]
\newtheorem {Corol}{Corollary}[section]
\newtheorem {Prop}{Proposition}[section]
\theoremstyle {definition}

\newtheorem{Rem}{Remark}

\newtheorem {Def}{Definition}[section]
\numberwithin {equation}{section}

\newcommand{\bsl}{\mathbf}
\newcommand{\ep}{\varepsilon}
\newcommand {\del}{\nabla}

\newcommand{\dive}{\textnormal{div}}

\newcommand{\dpar}{\partial}

\newcommand{\ueph}{{\bsl {u}_\ep^h}}
\newcommand{\veph}{{\bsl {v}_\ep^h}}
\newcommand{\uepht}{{\widetilde {\bsl {u}}_\ep^h}}
\newcommand{\zeph}{{\bsl {z}_\ep^h}}

\newcommand{\dx}{\mathrm {d}\bsl {x}}
\newcommand{\dy}{\mathrm {d}\bsl {y}}

\newcommand{\dl}{\mathrm {d}\lambda}

\newcommand{\wcon}{\rightharpoonup}
\newcommand{\stwo}{\overset {2}{\tends}}
\newcommand{\wtwo}{\overset {2}{\rightharpoonup}}
\newcommand{\stwor}{\overset {2}{\longrightarrow}}
\newcommand{\slq}{\underset{\mu_\ep^h}{\rightarrow}}
\newcommand{\wlq}{\underset{\mu_\ep^h}{\rightharpoonup}}
\newcommand{\dmu}{\mathrm {d}\mu}
\newcommand{\dlambda}{\mathrm {d}\lambda}

\newcommand{\lomegaeph}{L^2(\Omega,\dmu_\ep^h)}

\newcommand{\lqh}{L^2_{\rm per}(Q,\dmu^h)}
\newcommand{\lqx}{L^2_{\rm per}(Q,\dmu)}
\newcommand{\lqy}{L^2_{\rm per}(Q,\dlambda)}

\newcommand{\tends}{\rightarrow}
\newcommand{\rbb}{\mathbb {R}}

\newcommand{\nbb}{\mathbb {N}}

\newcommand{\whep}{\Omega^{h,\ep}}

\newcommand{\lime}{\lim_{\ep\tends 0}}
\newcommand{\limh}{\lim_{h\tends 0}}
\newcommand{\xy}{(\bsl {x},\bsl {y})}

\newcommand{\Tr}{\textnormal {tr}\ }

\newcommand{\per}{\textnormal {{per}}}
\newcommand{\pot}{\textnormal {{pot}}}
\newcommand{\sol}{\textnormal {{sol}}}
\newcommand{\home}{\textnormal {hom}}

\newcommand{\bzeta}{ \boldsymbol {\zeta}}

\newcommand{\vphi}{\varphi}
\newcommand{\bphi}{ \boldsymbol {\varphi}}

\newcommand{\bvphi}{\boldsymbol{\phi}}
\newcommand{\eph}{{h,\ep}}

\newcommand{\lopen}[2]{L^2(#1,#2)}

\begin{document}
%	\begin {titlepage}
%	\begin {center}
%		\textsc {\Large {Analytical Techniques for Partial Differential Equations on Thin Structures and their Applications for the Design of Metamaterials (Chapter 1)}}
%		\\[1cm] \Large {Ph.D.}
%		\\[1cm] \Large {2015}
%		\\[1cm] \Large {James Alexander Evans}
%	\end {center}		
%	\end {titlepage}

\title{\sc Homogenisation of thin periodic frameworks with high-contrast inclusions}
%Higher-order homogenisation of the system of Maxwell equations}
\author[1]{Kirill D. Cherednichenko\footnote{Corresponding author}}
\author[2]{James A. Evans}
\affil[1]{Department of Mathematical Sciences, University of Bath, Claverton Down, Bath, BA2 7AY, UK. Email: k.cherednichenko@bath.ac.uk}
\affil[2]{Cardiff School of Mathematics, Cardiff University, Senghennydd Road, CF24 4AG, UK}

\maketitle

\begin{abstract}
We analyse a problem of two-dimensional linearised elasticity for a two-component periodic composite, where one of the components consists of disjoint soft inclusions embedded in a rigid framework. We consider the case when the contrast between the elastic properties of the framework and the inclusions, as well as the ratio between the period of the composite and the framework thickness increase as the period of the composite becomes smaller. We show that in this regime the elastic displacement converges to the solution of a special two-scale homogenised problem, 
where the microscopic displacement of the framework is coupled both to the slowly-varying ``macroscopic'' part of the solution and to the displacement of the inclusions.  
%equilibrium equations for such elastic media  
We prove the convergence of the spectra of the corresponding elasticity operators to the spectrum of the homogenised operator with a band-gap structure. 
\end{abstract}
%\begin{keywords}
%\end{keywords}

\


		\section* {Introduction}\addcontentsline{toc}{section}{Introduction}

%The principal theory used throughout this chapter is the well established (in the context of homogenisation) theory of two-scale convergence. 
The multi-scale extension of the notion of the weak $L^2$-limit
%concept of a non-classical limit as the solution to a homogenisation problem was introduced in 
was proposed in \cite {bib37}, 
%in 1989 and has since been extended many times. In Nguetseng's original work, 
%where the 
%analysed functionals where the test functions considered were jointly continuous and periodic in the second argument with integration taken against the usual Lebesgue measure. His work introduced 
%the concept of a weak two-scale limit as the solution of some homogenisation problems under weak assumptions about the sequence of interest. This work was extended in 
\cite {bib38}, where a general theorem about two-scale compactness of $L^2$-bounded sequences was proved and a corrector-type result for the uniformly elliptic periodic homogenisation problem was established.
%homogenisation  
%that the weak two-scale limit was correct up to some strongly convergent remainder (in $L^2$). 
Multi-scale convergence has proved to be an effective tool in the analysis of the behaviour of periodic composite media 
%periodic homogenisation 
%problems with 
under minimal spatial regularity assumptions on the material properties of the composite, {\it e.g.} measurability and boundedness.
%%a complicated geometry of the periodic reference cell. 
%of the composite medium.
%have an advantage over other approaches in homogenisation theory, in particular, the ability to deal with problems with more complicated geometries. 
Further, in problems where solutions do not converge in the strong $L^2$-sense, for example in 
%periodic
the presence of degeneracies, see {\it e.g.} \cite{Smyshlyaev_degeneracies}, the related
%multi-scale convergence 
techniques have the additional benefit of capturing the multi-scale structure of the limit, by providing a suitable generalised notion of strong convergence.
%depending on  variables where the limit is periodic in one of the arguments. 
As opposed to the uniformly elliptic case,
%classical, unformly homogenisation 
where the limit function 
%as $\varepsilon\to0$ 
only depends on the macroscopic variable and is a solution to a single boundary-value problem, 
%to determine this limit,
%uniquely, 
the multi-scale limit for degenerate homogenisation problems satisfies  
%will (often) decompose into a macroscopic part and a microscopic part satisfying 
a coupled system of 
%boundary-value 
equations 
%which determine this non-classical limit uniquely 
%(see {\it e.g.} \cite {bib29})
%However, it is not always true that the homogenised system 
 %boundary-value problems 
% (homogenised system) will separate into equations 
 for the macroscopic and microscopic parts of the limit solution. This happens to be the case for periodic ``thin structures'', which are the subject of the present work.  
 %respectively. Indeed, the coupling in the limiting system can be more sophisticated meaning that the homogenised system can't be separated as is often found in problems on thin structures.

We define a {\sl thin structure} as an arrangement of rods 
%organised into a framework where the rods have some associated
of thickness $a>0$ 
%and where a finite number of these rods 
joined together at a number of junction points (``nodes''). 
Fig.\,\ref {examplenetwork} shows an example of a thin structure, where the two panels show rods 
%of thickness $\hat{h}$ (left) 
and the ``singular'' structure obtained by taking the mid-lines of the rods (right).
In the literature, equations of elasticity on thin structures are studied by treating $a=a(\varepsilon)$ as a parameter linked to the typical rod length 
$\varepsilon.$ 
%either for a fixed rod thickness $a$ or 
In the context of homogenisation, 
%it is usual to consider 
the rods are often assumed to be arranged periodically with 
%some associated cell of periodicity with small 
period $\ep,$ and the asymptotic behaviour of the structure is studied as $\varepsilon\to0.$ 
The use of two-scale convergence for the study of periodic singular structures
%, which can be viewed as a limit of thin  
has been proposed in \cite {bib30, bib33}, where
 %In 2000, Zhikov extended 
% the work of Nguetseng \cite {bib37} and Allaire \cite {bib38} 
 the two-scale approach of \cite {bib37, bib38} 
 %and  \cite {bib38} 
 %on two-scale analysis 
was extended to the setting of 
%. In \cite {bib31}, two-scale analysis for 
general Borel measures, 
%(as opposed to using the standard Lebesgue measure) 
and conditions on the measure sufficient for passing to the two-scale limit 
%necessary for a well-posed homogenisation theory 
were determined. 
%to presented. 
%Following on from this work, in Zhikov \cite {bib30} a full description of a multitude of thin periodic network problems are described wherein the existence (in most problems) of a non-classical limit is established. 

The use of multi-scale convergence techniques for the analysis of periodic thin structures was initiated in the work \cite{bib30}, 
which showed that if the thickness of the rods $a=a(\varepsilon)$ is a function of the period $\varepsilon$ of the network, such that 
$\lim_{\varepsilon\to0}a(\varepsilon)=0,$ then 
%as $\ep\tends 0$, 
%Zhikov also established in this work the dependence on the way 
the overall limit behaviour of the framework depends on the asymptotics of the ratio $a/\ep^2$ as $\varepsilon\to0.$ 
%behaves as $\ep\tends 0$. 
%Indeed, it was established 
In particular, in the case when 
%of the ``critical" scaling 
$\lim_{\varepsilon\to0}a/\ep^2=\theta>0,$ sequences of symmetric gradients of the solutions are, in general, not compact with respect to strong two-scale convergence. As a consequence, the equation describing the limit energy balance is no longer obtained by setting  the test function to be the solution of the homogenised equation for the corresponding ``singular structure", obtained by 
considering the mid-lines of the rods with the measure induced by the thin structure ({\it cf.} Fig.\,\ref{examplenetwork}). This problem was addressed by
%rectified in 2003 by Zhikov \& Pastukhova 
\cite {bib32}, where the correct form of the energy equality was determined and the limit system of equations was derived.
% which describe the homogenised problem were derived.
%The accomplishment of this work by Zhikov and Zhikov \& Pastukhova is complimented also 
This study was followed by the analysis of Sobolev spaces for a variable measure \cite{bib57, bib58}, Korn inequalities for periodic frames 
%(Zhikov \& Pastukhova 
\cite {bib75}, and gaps in the spectrum of the elasticity operator on a high-contrast periodic structure \cite{bib66} with non-vanishing volume fraction of the components as $\varepsilon\to0.$ 
In the paper \cite{bib66}, which can be viewed as the development of the results of \cite{bib31} to the high-contrast elasticity context, the band-gap nature of the spectrum of the limit operator is analysed and 
%conclude that 
the convergence %(in the sense of Hausdorff) 
of the spectra of the heterogeneous problems to the limit spectrum is proved.
% of the limit problem. 
%Further consideration of homogenisation on thin structures is given in the works of Pastukhova \cite {,bib59}, Zhikov \& Pastukhova \cite {bib123} and Bouchitt\'{e} \& Fragal\`{a} \cite {bib33}.
Notably, as was first observed by \cite{Zh_Past_Doklady}, the spectrum of the limit problem for a thin structure in the case of the above ``critical" scaling $\lim_{\varepsilon\to0}a/\ep^2=\theta>0$ shows a remarkable similarity to the limit spectrum for the high-contrast, fixed-volume-fraction case of \cite{bib66}. Some reasons for this similarity have been found in a recent work \cite{CherKis}, which uses operator-theoretic tools to show that the resolvents of both models are operator-norm close to a limit Kronig-Penney model of the so-called ``$\delta'$-type''. 

 %%%when the 
% \cite{bib30}  
%%%that the spectrum of the ``double-porosity model'', where the components of the composite have contrasting properties, 
%, which can an example of a high-contrast problem. The analysis of this problem concludes that 
%%%is close to a band spectrum 
%that has  
%of the resulting homogenised operator has 
%a band-gap structure 
%%%whose complement consists of an infinite set of disjoint intervals (``gaps'').
%, a property exhibited by a number of high-contrast problems. 
%%%This property is also possessed by the homogenised operator that we derive in the present work. 

In the present work 
%into a problem where consideration is given to the homogenisation of 
%for the
 %study of an elasticity problem on 
we consider a two-component periodic composite where the region occupied by the main material (``matrix'') is a framework with $a/\varepsilon^2\tends\theta>0$  as $\varepsilon\to0,$ and the complementary part of the space consisting of disjoint ``inclusions'' is filled by a less rigid material, so that the ratio between the stiffness of inclusions and the matrix is of the order $O(\varepsilon^2),$ $\varepsilon\to0.$ In other words, in addition to the assumption of high contrast, 
{\it cf.} \cite{bib31}, we assume that the stiff component is a thin structure so that its volume fraction is of the order $O(\varepsilon),$ $\varepsilon\to0.$  While our analysis uses some elements of  
%combine and 
both multi-scale approaches to thin structures of \cite {bib30, bib32}  and high-contrast structures of \cite{bib31, bib66}, the proofs of our results, namely homogenisation (Theorem \ref{homogenisation_theorem}) and spectral convergence (Theorem \ref{extension32}), require new tools that link the behaviour of solutions to the original sequence of problems with rapidly oscillating coefficients on the matrix and on the inclusions. The limit functions for the restrictions of the solutions to each of the two components are coupled together as described in Section \ref{sec231}, and lead to a homogenised system of equations of a new kind. 

The main technical difference between our problem and the existing work on homogenisation of periodic thin structures is the presence of a soft component whose volume fraction tends to unity, and which therefore interacts with the singular framework in a nontrivial way, in the homogenisation limit $\varepsilon\to0.$ One immediate key observation that follows is that the values of the elastic displacements on the singular framework have to be properly identified in the sense of traces with the boundary values of the two-scale limit on the soft component (which fills the whole of the periodicity cell in the limit as $\varepsilon\to0:$ the individual soft inclusions are only separated by the singular framework!)

%As a result, in terms of analysis, we require a new Theorem 2.1. This, in turn, requires Proposition 2.6 as an auxiliary result, which is, of course, part of the new ``generalised'' two-scale convergence framework, which is also new. Another new statement we prove is Theorem 2.5, which has no earlier analogues, as it describes the behaviour of $L^2$-bounded sequences and their ($\varepsilon$-scaled) $L^2$-bounded symmetric gradients with respect to the new  notion of two-scale convergence, where measure sequences with {\it two} parameters $\varepsilon$ and $h$ are involved. As we mention above, the two-scale limit function is defined on the whole cell $Q,$ which is the limit of a sequence of increasing domains supporting the parameter-dependent measures. 

%The proof of Theorem 3.1 thus uses both Theorem 2.1 and Theorem 2.4, as well as relies on the new functional framework. Finally, in order to handle the passage from the thin to singular structure, we also require Lemma 3.1, which is crucial for passing to the limit in the identity (3.10). 

%New difficulties, which are associated with the dependence of the measures that depend on both $\varepsilon,$ $h,$ have to be overcome also in the compactness statement for eigenfunctions, Theorem 4.1, and, in particular, Lemma 4.1 used in the proof of this statement.
%works of Zhikov \& Pastukhova  

%there is high-contrast between the rod framework and then regions outside the framework.

The structure of the paper is as follows. In Section \ref{sec21} we describe the geometric and measure-theoretic aspects of the problem we analyse and set up the corresponding high-contrast partial differential equation (PDE) in the weak form. We also recall the known results about the two-scale asymptotic behaviour of the solutions the weak formulations in the case when the soft inclusions are replaced by voids. We then outline the main result of the paper, concerning the behaviour of solutions to PDEs with a thin structure and high contrast. In Section \ref{sec22} we introduce a version of the notion of two-scale convergence (Section \ref{2scaledef}), appropriate to our measure-theoretic setup, and prove several preliminary statements concerning the compactness and asymptotic behaviour of various expressions associated with the solutions to the original problems, such as the scaled strain on the soft component (Sections \ref{sec221}, \ref{sec223}) and the restriction of the solution to the stiff component (Section \ref{sec222}). We also introduce the spaces of rigid displacements, potential and solenoidal matrices (Section \ref{prd}), which are convenient for the description of the structure of two-scale limits of solution sequences. In Section \ref{sec23} we describe the system of homogenised equations (\ref{hom21}) and prove our main result, Theorem \ref{homogenisation_theorem}, concerning convergence of solutions to the original equations (\ref{prob1}) to the homogenised system. Finally, in Section \ref{sec24} we study the behaviour of the spectra of the operators corresponding to the 
problems (\ref{prob1}), on the basis of the homogenisation results obtained. In particular, in Section \ref{sec241} we describe the spectrum of the operator associated with the homogenised problem (\ref{hom21}), and in Section \ref{sec242} we prove two-scale compactness of the sequence of eigenfunctions of the operators related to the heterogeneous problems (\ref{prob1}), which then implies the Hausdorff convergence of their spectra to the spectrum of the homogenised operator.

%Section {sec233}

%		\section {Analytical tools and two-Scale convergence}\label {sec21}

			\section {Problem formulation and main result}\label{sec21}

				%The following problem 
								  %Over the last decade or so, the theory of
				  %% Homogenisation thin rod frameworks has been considered in
				   %by Zhikov 
				   %%\cite {bib31,bib30, bib32}. 
				 We consider
				%s as its geometry
				 a periodic rod framework (``stiff" component of the composite) filled by a different material (``soft'' component).
	%% Denoting by $\ep$ the length scale of the 
				%associated unit 
		%%		period 
				%cell and 
			%%	and by $h$  in a thin structure, 
				We assume that the rod thickness $a>0$ is a function of the period $\varepsilon>0,$ 
				and consider the regime when
				%after scaling the framework by $\ep$, a rod framework is said to be:
					%{\begin {enumerate}
					%	\item sufficiently 
						%%$\lim\limits_{\ep\tends 0} h/\ep^2=0$ (``supercritical'')
						%\item sufficiently thick if 
						%%$\lim\limits_{\ep\tends 0} h/\ep^2=\infty$ (``subcritical'') 
						%\item critically thick if 
						$\lim_{\ep\tends 0} a/\ep^2=\theta>0.$ 
						%(``critical'').
					%\end {enumerate}}
							% shown in the left panel. 
       %In the present work we focus of the homogenisation of equations of linearised elasticity when the rod framework is critically scaled. 
			%The paper 
				%of Zhikov \& Pastukhova 
				%\cite {bib32} is concerned with the homogenisation of the governing equations of linearised elasticity in the case of a critical scaling. 
				%It contains an explicit formulation of a limit problem which has a non-classical structure. 
				%The work presented here looks to
			 The ratio of the elastic moduli of the soft and stiff component is assumed to be of the order $O(\varepsilon^2).$
			%	 much smaller than 
				 %those of the frame.
 %In the work we consider a critically scaled framework 
				%extend the results of \cite{bib32}
				%Zhikov \& Pastukhova to the case 
%				when the space between the critically scaled rods
				% framework (known as the soft component) 
%				is filled in with with some highly contrasting material.
					\begin {figure}[t]
						\centering
						\includegraphics[trim=0cm 0cm 5cm 0cm, width=9cm,height=4cm]{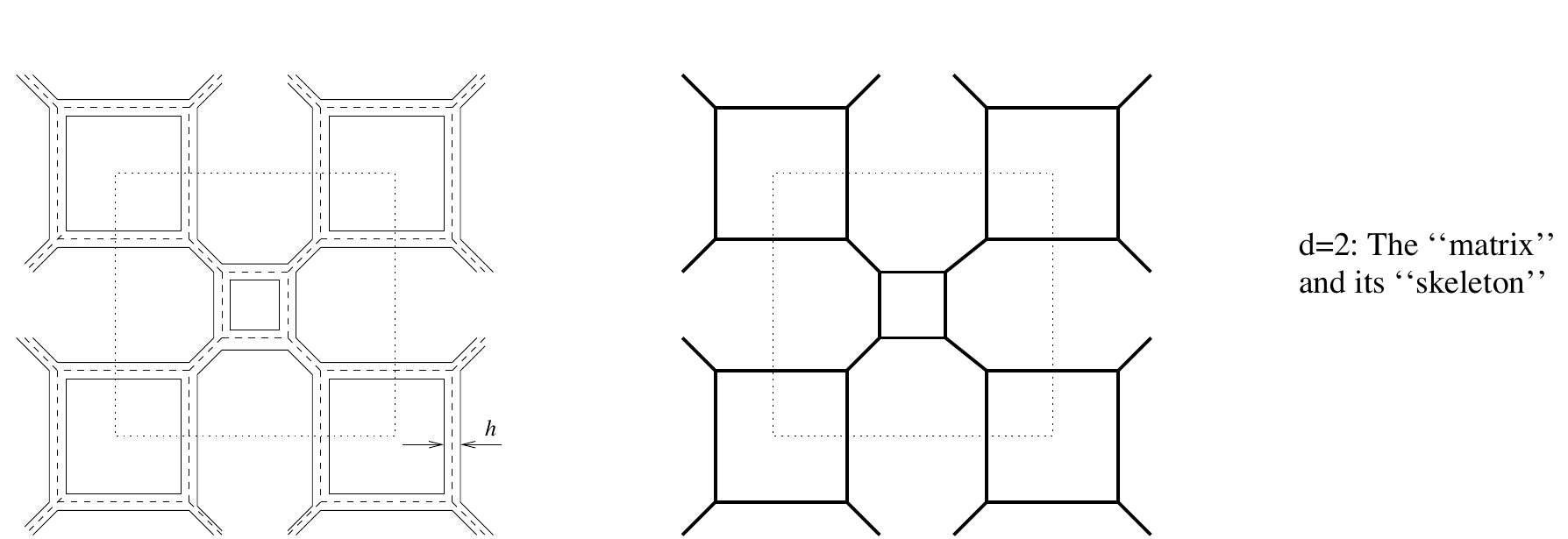}
						\caption {Example of a periodic network and unit cell.}
						\label {examplenetwork}
						\label {critthick}
					\end {figure}
%Let the following notation on the sets considered in this work be introduced.
				 Denote by $F_1^h$ the domain occupied by the scaled rods of thickness 
				 $h:=a/\varepsilon$ in the scaled structure of 
				 period $Q:=[0,1]^2$
				 %%$1$
				  %denote a periodic rod framework comprising of rods of thickness $h$
				  and by $F_1$ the corresponding singular structure, also of period $Q,$ obtained in the 
				  limit $h\to0.$ 
				 %associated tstructure f
				 %ormed 
				 %by taking the 
				 %%mid-lines of the rod framework $F_1^h$. 
				 %The unit cell of the problem will be denoted by 
				 %%%Consider the ``periodicity cell'' $Q:=[0,1]^2$ and denote 
				 %which is the union of two sets 
				 %$Q_1$ and $Q_0$ where 
				 %%%$Q_1:=F_1\cap Q$ and $Q_0:=Q\backslash Q_1$. 
				 %In other words, $Q_1$ is the stiff component of $Q$ and $Q_0$ is the soft component of $Q$. 
				 %Let $F_1^\eph$ denote the homothetic 
				 The original rod framework is then  the ``contraction'' $F_1^\eph :=\ep F_1^h$ of the framework $F_1^h.$  
				 %defined by the equality . 
				 The scaled soft component ${\mathbb R^2}\setminus F_1^h$ and the original soft component  $\ep({\mathbb R^2}\setminus F_1^h)$ are denoted by 
				 $F_0^h$ and $F_0^\eph,$ respectively. We denote by $\chi_1^h,$ $\chi_1^{h,\varepsilon}$ and  $\chi_0^h,$ $\chi_0^{h,\varepsilon}$ the characteristic functions
				 of the respective sets.
				
					%\begin {Rem} 
					%	The results established in this work can be applied to more general periodic frame networks than the one illustrated by Figure \ref {critthick2}. 
					%\end {Rem}
%Further notation will now be introduced. 
				 %This quantity is the associated elastic energy density.
						 %In the present work we study problems with coefficients of the form
				 %with high contrast between the stiff and soft inclusions, the elasticity tensor takes the form
				
				In what follows, we consider equations of two-dimensional elasticity in ${\mathbb R}^2.$ These are obtained from the full system of linearised elasticity in three dimensions when there is a direction, say $x_3,$ along which material properties are constant, assuming that the displacement does not depend on $x_3.$  At each 
				point $\bsl{x}\in{\mathbb R}^2,$ the fourth-order tensor of the elastic moduli of the 
				medium is set to be given by $$A^\ep(\bsl{x})=\varepsilon^2\chi_0^h(\bsl{x}/\ep)A_0+\chi_1^h(\bsl{x}/\ep)A_1,$$
				%%$$A^\ep(\bsl {x})=\begin {cases} 
			%%			\ep^2A_0(\bsl {x}/\ep), & 
			%%			%\text {in}
			%%			\bsl{x}\in F^{h, \varepsilon}_0,\\
			%%			\ \ \ A_1(\bsl {x}/\ep), & 
			%%			%\text {in}\
			%%			\bsl{x}\in F^{\ep,h}_1,
			%%		\end {cases}$$
				where $A_0$ and $A_1$ are constant
				%periodic, bounded and 
				positive definite fourth-order tensors:\footnote{The inner product of two symmetric matrices $\boldsymbol{\xi}=\{\xi_{ij}\}_{i,j=1}^2$ and $\boldsymbol{\eta}=\{\eta_{ij}\}_{i,j=1}^2$ is defined by 
				$\boldsymbol{\xi}:\boldsymbol{\eta} =\xi_{ij}\eta_{ij}$. 
				%In particular, $\boldsymbol{\xi}^2=\boldsymbol{\xi} \cdot \boldsymbol{\xi}$. 
				The product of the fourth-order elasticity tensor $A$ 
				with a symmetric matrix $\boldsymbol{\xi}$ is defined as $A\boldsymbol{\xi} =a_{ijkl}\xi_{kl}$ and 
				thus 
				%by the previous definition, it follows that 
				$A\boldsymbol{\xi}: \boldsymbol{\xi} 
				=a_{ijkl}\xi_{ij}\xi_{kl}$.} $A_j\boldsymbol{\xi} \cdot \boldsymbol{\xi}\geq c_j\boldsymbol{\xi}^2,$
				%\leq
				%\leq c_j^{-1}\boldsymbol{\xi}^2,$ 
				$c_j>0,$ $j=0,1,$ for all symmetric matrices $\boldsymbol{\xi}.$
				%:
				%, {\it i.e.}
					%%$c_j\boldsymbol{\xi}^2\leq A_j\boldsymbol{\xi} \cdot \boldsymbol{\xi}\leq c_j^{-1}\boldsymbol{\xi}^2,$ $c_j>0,$ $j=0,1.$
				
				For a bounded Lipschitz domain $\Omega\subset{\mathbb R}^2,$ 
				%be a bounded 
				%subset of the periodic frame network with
				 %Lipschitz set. 
				 %boundary. 
				 we denote by $\Omega^{h,\ep}_1:=\Omega \cap  F_1^\eph$ the stiff component 
				 %set of all stiff inclusions 
				 %%within $\Omega$ 
				 and by $\Omega^{h,\ep}_0:=\Omega \cap  F_0^\eph$ the soft component of the composite medium in 
				 %union of all soft inclusions in 
				 $\Omega.$  
				 Consider the measures $\lambda,$ $\lambda^h$ defined on $Q$ by 
				 \[
				 \lambda(B)=\frac{{\mathcal H}^1(F_1\cap B)}{{\mathcal H}^1(F_1\cap Q)}, \ \ \ \ \ \ \ \ \ 
				 \lambda^h(B)=\frac{{\mathcal H}^2(F_1^h\cap B)}{{\mathcal H}^2(F_1^h\cap Q)}\ \ \ \ \ \ \ 
				 \forall\ {\rm Borel}\ B\subset Q, 
				 \]
				 where ${\mathcal H}^d,$ $d=1,2,$ is the $d$-dimensional Hausdorff measure (see {\it e.g.} \cite{EvansGariepy}), and extended to ${\mathbb R}^2$ by $Q$-periodicity. Clearly, the weak convergence 
				 $\lambda^h\rightharpoonup\lambda$ holds as $h\to0,$ {\it i.e.} one has\footnote{We attach the superscript ``per'' to the notation for a function space when we refer to its subspace of 
				$Q$-periodic functions.}
					$$\limh\int_{Q}\bphi\,\dl^h=\int_{Q}\bphi\,\dl\ \ \ \ \forall\bphi\in\bigl[C^\infty_\per(Q)\bigr]^2.$$
				 Similarly, for the ``composite'' measures\footnote{The weightings in front of the measures $m$ and $\lambda$ in the definition of $\mu,$ which at the moment are set to $1/2$ and $1/2,$ respectively, can be split in an arbitrary way between the stiff and soft components, and it will only affect the coefficients in front of the first three integrals in the homogenised problem (\ref{hom21}).}
				 %$\mu$ denote the measure
				 %the periodic composite measure (as seen in Pastukhova \cite {bib59}) in $\rbb^2$ defined by the following relation
					$\mu:=(1/2)m+(1/2)\lambda$ and $\mu^h:=(1/2)m+(1/2)\lambda^h,$ where $m$ 
					is the plane Lebesgue measure, 
					%and $\lambda$ is the measure supported by 
				%the network $F_1$, proportional to the linear Lebesgue measure. Denote by $\lambda^h$ the measure supported by 
				%the rod network 
				%$F^h_1,$ proportional to the plane Lebesgue measure. Clearly $\lambda^h\wcon\lambda$ as $h\tends 0$, {\it i.e.} one has
				%Moreover, for the ``composite'' measures  
				one has
				$\mu^h\wcon\mu$ as $h\to0.$ 
				%where 
				%$\mu^h$ is the composite measure 
				%defined by 
				%%Choosing $\lambda$ and $\lambda^h$ so that $\lambda(Q)=\lambda^h(Q)=1$ yields 
			%%	$\mu(Q)=\mu^h(Q)=1.$
	%				$$\int_Qd\lambda=\int_Qd\lambda^h=1,$$
	%			yields
	%				$$\int_Qd\mu=\int_Qd\mu^h=1.$$
				Further, we consider the ``scaled'' measure $\lambda^h_\ep(B):=\ep^2\lambda^h (\ep^{-1}B)$ for all Borel $B\subset{\mathbb R}^2,$ and  $\mu^h_\ep:=(1/2)m+(1/2)\lambda^h_\ep,$ so that  $\mu_\ep^h\wcon m$ as $\varepsilon\to0.$
			Throughout the paper, for $\bsl{u}=\{u_j\}_{j=1}^3\in\bigl[H_0^1(\Omega)\bigr]^2,$ we denote by $\bsl{e}(\bsl{u})$ the (matrix-valued) symmetric gradient of $\bsl{u},$ with entries
			\begin{equation}
			\bsl{e}(\bsl{u})_{ij}:=\frac{1}{2}(\partial_ju_i+\partial_iu_j),\qquad i,j=1,2,3.	
			\label{sym_grad}		
			\end{equation}
			For $\ep, h>0$ and $\bsl {f}\in\bigl[L^2(\Omega)\bigr]^2$, we 
			look for $\bsl{u}^h_\varepsilon\in\bigl[H_0^1(\Omega)\bigr]^2$
			%consider 
				%a vector function 
				%%%%$\ueph\in \weph$ 
				such that
				%satisfying the identity
					%\begin {multline}
					\begin{equation*}
						\int_{\whep_1}A_1\bsl{e}(\ueph):\bsl{e}(\boldsymbol{\varphi} )\,\dmu_\ep^h +\ep^2\int_{\whep_0}A_0\bsl{e}(\ueph ):\bsl{e}(\boldsymbol{\varphi} )\,\dmu_\ep^h\ \ \ \ \ \ \ \ \ \ \ \ \ \ \ \ \ \ \ \ \ \ \ \ \ \ \ \ \ \ \ \ \ \ \ \  
					\end{equation*}
					\begin{equation}
						\label {prob1}	
						\ \ \ \ \ \ \ \ \ \ \ \ \ \ \ \ \ \ \ \ \ \ \ \ \ \ \ \ \ \ \ \ \ \ \ \  +\int_\Omega \ueph \cdot \boldsymbol{\varphi}\,\dmu_\ep^h 
						=\int_\Omega \bsl{f}\cdot \boldsymbol{\varphi}\,\dmu_\ep^h\ \ \ \forall\boldsymbol{\varphi}\in\bigl[H_0^1(\Omega)\bigr]^2.
						%W_{\ep, h}.
						%\bigl[C_0^\infty (\Omega)\bigr]^2.
					\end{equation}
					%\end {multline}
				%for all $\boldsymbol{\varphi} \in C_0^\infty (\Omega )^2$. 
				%The above problem has a unique solution.
				%(see Theorem \ref {unique}). 
				%%\begin {Thrm}\label {unique}
						Define a bilinear form $\mathfrak{B}^h_\ep(\cdot,\cdot)$ and a linear form $\mathfrak{L}^h_\ep(\cdot)$ 
						%on $\bigl[H_0^1(\Omega)\bigr]^2$ 
						%$\weph$  
						by 
						%the following relations:
							\begin{equation}
							\mathfrak{B}^h_\ep(\bsl {u},\bsl {v}):=\int_{\whep_1}A_1\bsl{e}(\bsl {u} ):\bsl{e}(\bsl {v} )\,\dmu_\ep^h +\ep^2\int_{\whep_0}A_0\bsl{e}(\bsl {u}):\bsl{e}(\bsl {v} )\,\dmu_\ep^h+\int_\Omega \bsl {u} \cdot \bsl {v}\,\dmu_\ep^h,\ \ \ \  \ \ \ \ \ 
							%$$
							%$$
							\mathfrak{L}^h_\ep(\bsl {v}):=\int_\Omega \bsl {f}\cdot \bsl {v}\,\dmu_\ep^h.
							\label{forms}
							\end{equation}
						Notice that $\mathfrak{B}^h_\ep$ is coercive and continuous, and 
						%the linear form 
						$\mathfrak{L}^h_\ep$ is continuous on $\bigl[H_0^1(\Omega)\bigr]^2.$ It is a consequence of the Lax-Milgram lemma 
						(see {\it e.g.} \cite[Chapter 6]{bib36}) that (\ref{prob1}) has a unique solution $\ueph.$ 
						%, a unique solution exists to problem (\ref {prob1}).
					%%\end {Thrm}
					%%\begin {proof}
						%%The proof of this result is a routine use of ellipticity estimates and the Cauchy-Schwarz inequality. The proof of this result follows the proof of Theorem \ref {unique1} analogously.
					%%\end {proof}
					%The homogenisation procedure that follows 
				%in this section 
				In what follows we aim to describe the structure of the limit problem 						%associated with
				for the weak two-scale limit of the function $\ueph$ as $h\to0,$ $\ep\tends 0,$ in such a way that $h/\varepsilon\to\theta>0.$ 
				%Note that the thickness of the rods $h(\ep)$ also goes to zero as $\ep\tends 0$. It has been established by Zhikov \cite {bib31,bib30} that the structure of the homogenised problem 						depends on the way in which $h$ goes to zero as $\ep\tends 0$, i.e., there is dependence on whether the rods are sufficiently thin, sufficiently thick or critically scaled.
					
					In the theory of homogenisation for periodic rod structures, when $A_0$ is formally replaced by zero in (\ref{prob1}), (\ref{forms}), the following results hold regardless of the asymptotic behaviour of the ratio $a/\ep^2,$ see \cite{bib30}, \cite{bib32}:
					\begin {enumerate}
						\item There exists a vector function 
						$\bsl {u}\xy\in\bigl[L^2\bigl(\Omega,L^2_{\rm per}(Q,\dlambda)\bigr)\bigr]^2$ 
						%periodic in $\bsl {y}$ 
						such that:
					\begin{equation*}
					{\rm a)\ \ \ \ }\frac {1}{\bigl|\Omega_1^\eph\bigr|}\int_{\Omega_1^\eph}\ueph(\bsl {x})\cdot\boldsymbol{\varphi}(\bsl{x}, \bsl{x}/\varepsilon)\dx\stackrel{\varepsilon\to0}{\longrightarrow}\int_\Omega\int_Q\bsl {u}(\bsl {x},\bsl {y})\cdot\boldsymbol{\varphi}(\bsl {x},\bsl {y})\dlambda(\bsl{y})\dx
					\ \ \ \ \ \ \ \ \ \ \ \ \ \ \ \ \ \ \ \ \ \ \ \ \ \ \ \ \ \ \ \ \ \ \ \ \ \ \ \ \ \ 
				\end{equation*}
				\begin{equation}
				\ \ \ \ \ \ \ \ \ \ \ \ \ \ \ \ \ \ \ \ \ \ \ \ \ \ \ \ \ \ \ \ \ \ \ \ \ \ \ \ \ \ \ \ \ \ \ \ \ \ \ \ \ \ \ \ \ \ \ 
				\forall\boldsymbol{\varphi}\in\bigl[L^2\bigl(\Omega,L^2_{\rm per}(Q,\dmu)\bigr)\bigr]^2;
						 \label{L2estimate1}
\end{equation}
%%for all $\bsl {u}\xy\in\bigl[C\bigl(\Omega,L^2_{\rm per}(Q,\dmu)\bigr)\bigr]^2;$
\begin{equation}
{\rm b)\ \ \ \ }	 \frac {1}{\bigl|\Omega_1^\eph\bigr|}\int_{\Omega_1^\eph}\bigl\vert\ueph(\bsl {x})\bigr\vert^2\dx\stackrel{\varepsilon\to0}{\longrightarrow}\int_\Omega\int_Q\bigl\vert\bsl {u}(\bsl {x},\bsl {y})\bigr\vert^2\dlambda(\bsl{y})\dx.\ \ \ \ \ \ \ \ \ \ \ \ \ \ \ \ \ \ \ \ \ \ \ \ \ \ \ \ \ \ \ \ \ \ \ \ \ \ \ \ \ \ \ \ \ \ \ \ \ \ \ \ \ \ \ \ \ \ \ \ \ \ \ \ \ \ \ \ \ \ \ \ \ \ 
						 \label{L2estimate2}
\end{equation}
							%%\begin{equation}
							%%\lime\frac {1}{\bigl|\Omega_1^\eph\bigr|}\int_{\Omega_1^\eph}\bigl|\ueph(\bsl {x})-\bsl {u}(\bsl {x},\bsl {x}/\ep)\bigr|^2\dx=0.
							%%\label{L2estimate}
							%%\end{equation}
						\item The vector $\bsl {u}(\bsl {x},\cdot)$ is a ``periodic rigid displacement'' (see Definition \ref {prd}). For many frameworks of interest this implies that 
							\begin{equation}
							\bsl {u}\xy=\bsl {u}_0(\bsl {x})+\boldsymbol{\chi}\xy,\ \ \ \ {\rm a.e.}\ \bsl{x}\in \Omega,\ \ \ \lambda{\text -}{\rm a.e.}\ \bsl{y}\in F_1\cap Q.
							\label{boundary}
							\end{equation}
						%The vector $\bsl {u}$ is called a and 
						where  $\bsl {u}_0\in\bigl[H_0^1(\Omega)\bigr]^2$ and $\boldsymbol{\chi}(\bsl{x}, \cdot)$ is a ``periodic transverse displacement'',
						{\it i.e.} for a.e. ${\mathbf x}\in\Omega$ the function 
				%, $\bsl {\chi}$ 
				$\boldsymbol{\chi}(\bsl{x}, \cdot)$ is orthogonal to the link of $F_1$ on which it is defined.
						%, see Definition \ref{prd} below.
						%(\ref{rigid1}). 
						\item The ``macroscopic'' equation 
						%is satisfied by $\bsl {u}_0\in\bigl[H_0^1(\Omega)\bigr]^2$:
							\begin{equation}
							-\frac{1}{2}\dive\bigl(A^\home_\lambda\bsl{e}(\bsl {u}_0)\bigr)+\int_Q\bsl{u}(\cdot, \bsl{y})\dlambda(\bsl{y})
							%\langle \bsl {u}\rangle_{\bsl{y}}
							=\bsl {f}
							\label{homogenised}
							\end{equation}
						holds, where $A^\home_\lambda$ is the ``$\lambda$-homogenised tensor'' defined by (\ref {ahom1}).
					\end {enumerate}

				%The above properties will be used to study
				%say something about  
				%the problem (\ref {prob1}) for the critical scaling and where there is contrast between the stiff and soft components. 
			%%%%%%
			%%	It is shown by a formal asymptotic argument that in the case under consideration (critical fraemword thickness, contrast of order $O(\varepsilon^2)$), on has 
				%on the whole domain $\Omega$ 					that the solution takes the form 
				%%	\begin {equation}\label {solution001}
					%%	\bsl {u}\xy=\bsl {u}_0(\bsl {x})+\bsl {U}\xy,\ \ \ \bsl {U}\in\bigl[\lqw\bigr]^2.
					%%\end {equation}
				%A feature of interest in this problem is the non-classical trace associated with the function $\bsl {U}(\bsl {x},\cdot )$. Indeed, it will be shown that
			%%In what follows, we show that $$\underset{\bsl {y}\in F_1}{\tr}\bsl {U}\xy=\boldsymbol{\chi}\xy,$$ which implies, in particular, that 
				%where $\boldsymbol{\chi}$ is the function defined on the rod network.
				%% The fact that $\bsl {U}$ satisfies equations on both the soft inclusion and the network with this unknown trace data at the interface leads to a 
				%%very intricate 
				%%coupling between 
		%%$\bsl {U}$ and $\boldsymbol{\chi}$ are coupled.
    %%%%%%%%%%
   Our main result, Theorem \ref{homogenisation_theorem}, states that in the case when $a/\varepsilon^2\to\theta>0$ as $\varepsilon\to0,$ the solutions $\bsl{u}^h_\varepsilon$ to the problems (\ref{prob1}), where $h=a/\varepsilon,$ converge in an appropriate two-scale sense (see Section \ref{sec22}) to a function $\bsl{u}(\bsl{x}, \bsl{y}),$ $\bsl{x}\in\Omega,$ $\bsl{y}\in Q,$ whose trace on $F_1\cap Q$
   %the boundary $\partial Q$ of the period cell $Q$ 
   has the form (\ref{boundary}) and satisfies an equation involving $F_1$-transversal  components of the $\bsl{y}$-gradient of the function $\bsl{u}=\bsl{u}(\bsl{x}, \bsl{y}).$ In addition, the function $\bsl{u}(\bsl{x}, \cdot)-\bsl{u}_0(\bsl{x})=:\bsl{U}(\bsl{x}, \cdot),$ $\bsl{x}\in\Omega,$ belongs to the space $\bigl[H_{\rm per}^1(Q)\bigr]^2$ a.e. $\bsl{x}\in\Omega$ and satisfies an elliptic equation that couples its values to the solution 
   $\bsl{u}_0$ of (\ref{homogenised}), where the average $\int_Q\bsl{u}(\cdot, \bsl{y})\dlambda(\bsl{y})$ is replaced by $\int_Q\bsl{u}(\cdot, \bsl{y})\dmu(\bsl{y}).$
   %\langle \bsl {u}\rangle_{\bsl{y}}:=
 % \int_Q u(\cdot, \bsl{y})\dy$ 
% in the interior of $Q$
% , which in particular implies the estimate   
   
    More precisely, for each link $I$ of the network $F_1,$ let $\boldsymbol{\tau}$ and $\boldsymbol{\nu}$ be unit tangent and normal vectors  that form a positively orientated system. Then all vectors $\bsl {v}\in{\mathbb R}^2$ are written as $\bsl{v}=v^{(\tau)}\boldsymbol{\tau}+v^{(\nu)}\boldsymbol{\nu},$ where $v^{(\tau)}=\bsl {v}\cdot\boldsymbol{\tau}$ and $v^{(\nu)}=\bsl {v}\cdot\boldsymbol{\nu}.$ In Section \ref{sec23} the vectors $\bsl {U}$ and $\boldsymbol{\chi}$ are shown to 
    satisfy a system of equations of the form
					%\begin {align}
					\begin{equation}
						\mathcal {A}_0\bsl {U}+\bsl {u}= \bsl {f},%\label {fast1}\\
						\ \ \ \ \ \ \ \ \ 
						\mathcal {L}_{\tau}^{\rm IV}\chi^{(\nu)}+\mathcal{T}_{\nu}U^{(\nu)}+u^{(\nu)}= f^{(\nu)},\qquad \mathcal {L}_{\tau}^{\rm I}\chi^{(\nu)}+\mathcal{T}_{\tau}U^{(\tau)}+u^{(\tau)}= f^{(\tau)},\label {fast2}
					\end{equation}
					%\end {align}
				where $\mathcal {A}_0$ is a second-order differential operator in $Q$ expressed in terms of the 
				tensor $A_0$ only,  
				$\mathcal {L}_{\tau}^{\rm IV},$ $\mathcal {L}_{\tau}^{\rm I}$ are fourth-order and first-order differential operators in the ``longitudinal'' direction 
				$\boldsymbol{\tau},$ and $\mathcal{T}_{\nu},$ $\mathcal{T}_{\tau}$ are first-order differential operators in the ``transverse'' 
				direction $\boldsymbol{\nu}$ corresponding to each link $I$. 

		\section {Two-scale structure of solution sequences}\label {sec22}

%%Consider the problem formulated in Section \ref {sec211}, that is, for a given function $\bsl {f}\in\bigl[C^\infty (\overline {\Omega})\bigr]^2$, find a vector-valued function $\ueph\in \weph$ satisfying the identity
					%\begin {multline}
	%%				\begin{equation}
		%%				\int_{\whep_1}A_1\bsl{e}(\ueph )\cdot \bsl{e}(\boldsymbol{\varphi} )\ \dmu_\ep^h +\ep^2\int_{\whep_0}A_0\bsl{e}(\ueph )\cdot \bsl{e}(\boldsymbol{\varphi} )\ \dmu_\ep^h 
						%+ \\ 
			%%			+\int_\Omega \ueph \cdot \boldsymbol{\varphi}\ \dmu_\ep^h =\int_\Omega \bsl 										{f}\cdot \boldsymbol{\varphi}\ \dmu_\ep^h\ \ \ \forall \boldsymbol{\varphi} \in \bigl[C_0^\infty (\Omega )\bigr]^2
				%%\end{equation}
				%	\end {multline}
				%for all $\boldsymbol{\varphi} \in (C_0^\infty (\Omega ))^2$. 
				%%It will now be illustrated that this problem has a unique solution for each fixed $\ep>0$.
									In this section we establish the structure of various two-scale limits on the soft and stiff components. This is achieved by taking the limits, as $\varepsilon\to0,$ of the integrals entering the identity (\ref{prob1}), with suitably chosen test functions $\boldsymbol{\varphi}.$

\subsection{Two-scale convergence: definition and properties}
\label{2scaledef}

We first recall the notion of weak and strong two-scale convergence and their basic properties, see \cite{bib30}.  Within this section $d=2$ or $d=3,$ and the measure sequence $\mu_\varepsilon^h$ (respectively, ``limit" measure $\mu$) can be replaced by the sequence $\lambda_\varepsilon^h$ (respectively, ``limit'' measure $\lambda$).  
	%It is assumed throughout that the sequence $\ueph$ is bounded in $\bigl[\lomegaeph\bigr]^2$
								\begin {Def}[Weak two-scale convergence]
						Suppose that  $h$ is a function of $\varepsilon$ such that $h\to0$ as $\varepsilon\to0,$ and $\{\ueph\}\subset\bigl[\lomegaeph\bigr]^d$ is a bounded sequence:
						%, {\it i.e.}
						\begin {equation}\label {limsup}
						\limsup_{\ep\tends 0}\int_\Omega |\ueph|^2\ \dmu_\ep^h <\infty.
					\end {equation} 
						We refer to $\bsl {u}(\bsl {x},\bsl {y})\in\bigl[L^2(\Omega\times Q,\dx\times\dmu)\bigr]^d$
						%=:\bigl[\lqw\bigr]^2$ 
						as the {\sl weak two-scale limit} of $\ueph$, denoted $\ueph\wtwo \bsl {u},$ if 
							\begin {equation}\label {wtwocon}
								\lime\int_\Omega\ueph(\bsl {x})\cdot\bsl {\Phi}(\bsl {x},\bsl {x}/\ep)\,\dmu^h_\ep=\int_\Omega\int_Q\bsl {u}\xy\cdot\bsl {\Phi}\xy\,\dmu(\bsl{y})\dx\ \ \ \ \ \ \forall\bsl{\Phi}\in\bigl[L^2(\Omega, C_{\rm per}(Q))\bigr]^d.
							\end {equation}
					\end {Def}
			%	The following proposition on two-scale compactness 
				%is an essential part of the theory of two-scale convergence. With this result, 
			%	ensures t a two-scale limit (up to possibly taking a subsequence) in a bounded sequence. 
				%is possible
%				, see {}
					\begin {Prop}[Two-scale compactness]
					\label {compact1}
						If a sequence $\ueph$ is bounded in $\bigl[\lomegaeph\bigr]^d$, then it is compact with respect to weak two-scale convergence.
					\end {Prop}
					\begin {Prop}
					\label {prop01}
						%Suppose that $\ueph$ be a bounded sequence in $\bigl[\lomegaeph\bigr]^2$ and 
						If $\ueph\wtwo \bsl {u}$ 
						%where $\bsl {u}\in \lqw^2$. 
						then $\Vert\bsl{u}\Vert_{[L^2(\Omega\times Q)]^d}\le\liminf_{\ep\tends 0}
						\Vert\bsl{u}^h_\ep\Vert_{[L^2(\Omega, \dmu^h_\ep)]^d}.$
						%\int_\Omega |\ueph|^2\ \dmu^h_\ep\geq \int_\Omega\int_Q |\bsl {u}|^2\ \dmu \dx.$
							%\begin {equation}\label {bound}
								%\liminf_{\ep\tends 0}\int_\Omega |\ueph|^2\ \dmu^h_\ep\geq \int_\Omega\int_Q |\bsl {u}|^2\ \dmu \dx.
							%\end {equation}
					\end {Prop}
				%Akin to the idea of strong convergence, there is an analogous definition for the concept of strong two-scale convergence.
				
					\begin {Def}
					%[Strong two-scale convergence]
						Let $\ueph$ be a bounded sequence in $\bigl[\lomegaeph\bigr]^d$. We say that a function
						$\bsl {u}=\bsl {u}\xy\in\bigl[L^2(\Omega\times Q, \dx\times\dmu)\bigr]^d$ is the {\sl strong two-scale limit} of $\ueph$, denoted $\ueph\stwo \bsl {u}$, if for any weakly two-scale convergent sequence 											$\bsl {v}^h_\ep\wtwo \bsl {v}$ one has							
						\begin {equation}\label {stwocon}
								\lime\int_\Omega\ueph\cdot\bsl {v}^h_\ep\,\dmu^h_\ep=\int_\Omega\int_Q\bsl {u}\xy\cdot\bsl {v}\xy\,\dmu(\bsl{y})\dx.
							\end {equation}
					\end {Def}
				%Strong convergence implies weak convergence, by taking $\bsl {v}^h_\ep(x)=\Phi(\bsl {x},\bsl {x}/\ep)$ in the above definition. Moreover, 
				Note that by setting $\bsl {v}^h_\ep=\ueph$ one has
					\begin {equation}\label {squareuep}
						\lime\int_\Omega|\ueph|^2\,\dmu_\ep^h=\int_\Omega\int_Q|\bsl {u}|^2\ \dmu\dx.
					\end {equation}
				The next proposition shows that the converse 
				%of the above observation 
				also holds.
					\begin {Prop}
					\label{strong_conv}
						%Let $\ueph$ be a bounded sequence in $\bigl[\lomegaeph\bigr]^2$ and $\bsl {u}\in\bigl[\lqw\bigr]^2$. 
						If $\ueph\wtwo \bsl {u}$ and the convergence (\ref {squareuep}) holds, then $\ueph\stwo \bsl {u}$.
					\end {Prop}
					
					\begin{Prop}
					\label{a_factor}
For any arbitrary $a\in L^\infty(Q),$ the weak (resp. strong) two-scale convergence of $\bsl{u}_\ep^h$ to $\bsl{u}(\bsl{x}, \bsl{y})$ implies the weak (resp. strong) two-scale convergence 
of $a(\cdot/\ep)\bsl{u}_\ep^h$ to $a(\bsl{y})\bsl{u}(\bsl{x}, \bsl{y}).$
					\end{Prop}
					
					%%\begin {proof}
			%%			Let $(\ep_k)_{k\in\nbb}$ be a subsequence of $\ep$ such that $\ep_k\tends 0$. Assume that $\ueph\wtwo\bsl {u}$, $\bsl {v}^h_\ep\wtwo \bsl {v}$. Then the following limit exists:
		%%					$$a:=\lim_{\ep_k\tends 0}\int_\Omega\ueph
							%(\bsl {x})
			%%				\cdot \bsl {v}^h_\ep\,
							%(\bsl {x})\ 
		%%					\dmu^h_\ep.$$
		%%				By Proposition \ref {prop01} the following inequality also holds for all $t\in\rbb$:
		%%					\begin {equation}\label {lowsemcont}
		%%						\lim_{\ep_k\tends 0}\int_\Omega|\bsl {v}^h_\ep+t\ueph|^2\,\dmu^h_\ep\geq \int_\Omega\int_Q|\bsl {v}+t\bsl {u}|^2\,\dmu\dx.
			%%				\end {equation}
		%%				Moreover, utilising relation (\ref {squareuep}), it also follows that
		%%					$$\lim_{\ep_k\tends 0}\int_\Omega|\bsl {v}^h_\ep+t\ueph|^2\,\dmu^h_\ep=\int_\Omega\int_Q|\bsl {v}|^2\,\dmu\dx+2ta+t^2\int_\Omega\int_Q|\bsl {u}|^2\,\dmu\dx.$$
	%%					Hence, combining the above with the inequality (\ref {lowsemcont}), it follows that 
		%%					$$2ta\geq 2t\int_\Omega\int_Q\bsl {u}\cdot\bsl {v}\ \dmu\dx.$$
	%%					Considering first $t>0$ and then $t<0$ yields 
						%the equality
	%%						$$a=\int_\Omega\int_Q\bsl {u}\cdot \bsl {v}\ \dmu\dx,$$
	%%					as required.
	%%				\end {proof}

		%\end {subappendices}

			\subsection {Two-scale compactness of solutions to (\ref{prob1})}\label{sec221}
				Consider the equation (\ref {prob1}) with $\boldsymbol{\varphi} =\ueph$:
					%\begin {multline*}
					\begin{equation}
						\int_{\whep_1}A_1\bsl{e}(\ueph ):\bsl{e}(\ueph )\ \dmu_\ep^h +\ep^2\int_{\whep_0}A_0\bsl{e}(\ueph ):\bsl{e}(\ueph )\ \dmu_\ep^h 
						%+ \\ 
						+\int_\Omega |\ueph |^2\ \dmu_\ep^h =\int_\Omega \bsl {f}\cdot \ueph\ \dmu_\ep^h.
					\end{equation}
					%\end {multline*}
				Using ellipticity estimates on the left-hand side and the inequality $2ab\le a^2+b^2,$ $a,b\in{\mathbb R},$ on the right-hand side yields
				%followed by inequality (\ref {inequal1}) it can be seen that
					$$c_0\ep^2\int_{\whep_0}\bigl|\bsl{e}(\ueph)\bigr|^2\,\dmu_\ep^h+c_1\int_{\whep_1}\bigl|\bsl{e}(\ueph)\bigr|^2\,\dmu_\ep^h +\frac {1}{2}\int_\Omega |\ueph |^2\,\dmu_\ep^h \leq \frac {1}{2}\int_\Omega |\bsl{f} |^2\,\dmu_\ep^h,
					$$
					where $c_0,$ $c_1$ are the ellipticity constants of $A_0,$ $A_1.$
					%\|\bsl {f}\|^2_{[L^2(\Omega, \dmu_\ep^h)]^2},
					%\ \ \ \ \ \ \ c_0,c_1>0.$$
	%				$$\iff \| \bsl{e}(\ueph)\|_{\lopen{\whep_1}{\dmu_\ep^h}}^2+\left ( \ep\| \bsl{e}(\ueph)\|_{\lopen{\whep_0}{\dmu_\ep^h}}\right )^2+\frac {1}{2}\| \ueph\|^2_{L^2(\Omega )}\leq \| \bsl {f}\|^2_{L^2(\Omega )}.$$
				Hence, the following {\it a priori} bounds hold.
					\begin {Prop}
						Let $\ueph$ be a sequence in $\bigl[\lomegaeph\bigr]^2$ of solutions to (\ref {prob1}). Then 
						there exists $C>0$ such that
						% following bounds hold:
							%\begin {align}
							\begin{equation*}
								\|\bsl {u}^h_\ep\|_{[L^2(\Omega,\dmu_\ep^h)]^2} \leq  C,\ \ \ \ 
								%\label {bound1}\\
								\bigl\|\bsl{e}(\bsl {u}^h_\ep)\bigr\|_{[L^2(\whep_1,\dmu_\ep^h)]^3} \leq  C,\ \ \ \ 
								%\label {bound2}\\
								\ep\bigl\|\bsl{e}(\bsl {u}^h_\ep)\bigr\|_{[L^2(\whep_0,\dmu_\ep^h)]^3}  \leq  C. 
								\ \ \ 
								%\label {bound3}
							\end{equation*}
							%\end {align}
					\end {Prop}
				Using two-scale compactness of $L^2$-bounded sequences (see Proposition \ref{compact1}), we assume that the sequences 
				\[
				\ueph,\ \ \ \ \chi_1^\eph\ueph\ \  {\rm (displacements),\ \ \ \ and}\ \ 
				%and the sequences 
				\ \ \ \ \chi_1^\eph \bsl{e}(\ueph),\ \ \ \ \ep \chi_0^\eph \bsl{e}(\ueph)\ \ \ {\rm (strains)}
				\]
				 weakly two-sale converge to functions 
				%, {\it i.e.} there exist 
				$
				\bsl {u}\xy\in\bigl[L^2(\Omega\times Q,\dx\times\dmu)\bigr]^2,$ 
				$\widehat{\bsl {u}}\xy\in\bigl[L^2(\Omega\times Q,\dx\times\dlambda)\bigr]^2$ (displacements), and
				$\bsl{p}\xy\in\bigl[L^2(\Omega\times Q,\dx\times\dlambda)\bigr]^3,$
				%%\ \ 
				%\bigl[\lqw\bigr]^2:=
				%%\ \ \  \ 
				$\widetilde {\bsl{p}}\xy\in\bigl[L^2(\Omega\times Q,\dx\times\dy)\bigr]^3$ (strains),
				%\in\bigl[\lqw\bigr]^3,
				%%\]
				respectively. Here, ${\rm d}{\mathbf y}$ denotes the differential of the Lebesgue measure on $Q,$ and each of the spaces $L^2(\Omega\times Q,\dx\times\dlambda)$ and $L^2(\Omega\times Q,\dx\times\dy)$ is treated as a subspace of $L^2(\Omega\times Q,\dx\times{\rm d}\mu).$
				%such that 
					%\begin {align}
					%	\ueph\wtwo & \ \bsl {u}\xy\in \lqw^2,\ & \text {in} & \ \ \lopen{\Omega}{\dmu_\ep^h}^2,\label {struc1} \\
					%	\chi_1^\eph\ueph\wtwo & \ \widetilde {\bsl {u}}\xy\in \lqw^2,\ \ \ & \text {in} & \ \ \lopen{\whep_1}{\dmu_\ep^h}^2,\label {struc2} \\
					%	\bsl{e}(\ueph)\wtwo & \  p\xy\in \lqw^3,\ \ \ & \text {in} & \ \ \lopen{\whep_1}{\dmu_\ep^h}^3,\label {struc3} \\
					%	\ep \bsl{e}(\ueph)\wtwo & \ \widetilde {p}\xy\in \lqw^3,\ \ \ & \text {in} & \ \ \lopen{\whep_0}{\dmu_\ep^h}^3.\label {struc4}
					%\end {align}
				%Note that $\chi_1^\eph$ is the characteristic function on the stiff inclusions.
				%With these bounds, and hence the existence of the two-scale limits to various terms appearing in the
				%%%Using the integral identity (\ref {prob1}), in Section ... , we establish the structures of the these two-scale limits.

	%The results presented in this section are common notions which crop up in the analysis of thin structures.
	
\subsection{Rigid displacements, potential and solenoidal matrices}
				
%%				 The concepts of a periodic rigid displacement and transverse displacement are of particular importance for homogenisation of thin structures.
					\begin {Def}
					%[Periodic rigid displacements]
					\label {prd}
						A vector function $\bsl {u}\in\bigl[\lqy\bigr]^2$ 
						%defined on the singular network $F_1$ 
						is said to be a {\sl periodic rigid displacement} (with respect to the measure $\lambda$) if there exists a sequence $\{\bsl {u}_n\}
						%_{n\in\nbb}
						\subset\bigl[C_\per^\infty(Q)\bigr]^2$ such that
							%$\bsl {u}_n\tends \bsl {u},$ $\bsl{e}(\bsl {u}_n)\tends 0$
							%\ \ \ \text {in}\ 
						$\bigl(\bsl {u}_n, \bsl{e}(\bsl {u}_n)\bigr)\tends (\bsl{u}, 0)$
							in $\bigl[\lqy\bigr]^5$ as $n\to\infty.$
							%\times\bigl[\lqy\bigr]^3.$
						We denote the set of periodic rigid displacements by $\mathcal {R},$ omitting the reference to the measure $\lambda.$
					\end {Def}
				We assume (see {\it e.g.} \cite {bib30} for relevant examples of periodic frameworks) that any $\bsl {u}\in \mathcal {R}$ has a unique representation 
				%of the form
					\begin {equation}\label{rigid1}
						\bsl {u}(\bsl {y})=\bsl {c}+\boldsymbol{\chi}(\bsl {y}),\ \ \ \bsl{y}\in Q,
					\end {equation}
				where $\bsl {c}\in{\mathbb R}^2$ %(with respect to $\bsl {y}$) 
				and $\boldsymbol{\chi}$ is a periodic transverse displacement, {\it i.e.} on each link of the singular network $F_1$
				%, $\bsl {\chi}$ 
				it is orthogonal to the link. 
				%The  
				Thus $\mathcal {R}$ is the direct sum of $\rbb^2$ and the set of periodic transverse displacements, which we denote by $\widehat{\mathcal{R}}.$
	%%%			Denoting by $\widehat{\mathcal{R}}$ the set of transverse displacements, 
				%is							and thus 
				%the set of all periodic rigid displacements can be decomposed into the form 
		%%%		we thus have $\mathcal {R}=\rbb^2\oplus\widehat{\mathcal{R}}$. 
					The next definition characterises periodic transverse displacements that occur in the study of rod networks with $a/\varepsilon^2\to\theta>0$ as $\varepsilon\to0.$ 
				%that are critically scaled.
					\begin {Def}\label {rigidstuff}
						%Let $\boldsymbol{\nu},\ \boldsymbol{\tau}$ be normal and tangent to each link $I$ of the periodic network $F_1$ such that they form a positively orientated frame, and 
						Denote by $I_1,\dots ,I_n$ the links of the network $F_1$ sharing an arbitrary node $\mathcal{O},$ and denote by $(\boldsymbol{\chi}\cdot\boldsymbol{\nu})'$ the derivative in the tangential direction:
						%, {\it i.e.}
							$(\boldsymbol{\chi}\cdot\boldsymbol{\nu})'
							%=\frac {\mathrm {d}(\boldsymbol{\chi}\cdot\boldsymbol{\nu})}{\mathrm {d}\boldsymbol{\tau}} 
							:=(\boldsymbol{\tau}\cdot\del)(\boldsymbol{\chi}\cdot\boldsymbol{\nu}).$ The set $\widehat{\mathcal{R}}^0\subset \widehat{\mathcal{R}}$ is defined to consist 
						%is used to denote the set 
						of periodic transverse displacements $\boldsymbol{\chi}$ satisfying the following conditions:
							%\begin {enumerate}
								%\item 
								
								(C1) The function $\boldsymbol{\chi}\cdot\boldsymbol{\nu}_j|_{I_j}$, $j=1,2,\dots, n$, 
								%is square integrable and 
								has square integrable second derivatives on $I_j,$ {\it i.e.} one has
								%denoted 
								$\boldsymbol{\chi}\cdot\boldsymbol{\nu}\in H^2(I_j)$.
								%\item 
								
								(C2) The first derivative along the link is continuous across each node:
									$(\boldsymbol{\chi}\cdot\boldsymbol{\nu}_1)'\big|_{\mathcal {O}}=(\boldsymbol{\chi}\cdot\boldsymbol{\nu}_2)'\big|_{\mathcal {O}}=\dots =(\boldsymbol{\chi}\cdot\boldsymbol{\nu}_n)'\big|_{\mathcal {O}}.$
								%\item
								
								 (C3) Each node is fastened:
									$\boldsymbol{\chi}|_{\mathcal {O}}=\bsl {0}.$
							%\end {enumerate}
						
						The norm in $\widehat{\mathcal{R}}^0$ is defined to be the sum of the $H^2$-norms of $\boldsymbol{\chi}\cdot\boldsymbol{\nu}$ over all the links.
					\end {Def}
					%%%\begin {Rem}
					%%%	In what follows, we denote by $(\boldsymbol{\chi}\cdot\boldsymbol{\nu})'$ the derivative in the tangential direction:
						%, {\it i.e.}
						%%%	$(\boldsymbol{\chi}\cdot\boldsymbol{\nu})'
							%=\frac {\mathrm {d}(\boldsymbol{\chi}\cdot\boldsymbol{\nu})}{\mathrm {d}\boldsymbol{\tau}} 
						%%%	=(\boldsymbol{\tau}\cdot\del)(\boldsymbol{\chi}\cdot\boldsymbol{\nu}).$
				%%%	\end {Rem}
					\begin {Def}
					%[Potential \& Solenoidal Matrices]
						For a given Borel measure $\varkappa$ on $Q,$ we define the space $V_\pot^\varkappa$ of $\varkappa$-{\sl potential matrices} as the closure of the set $\bigl\{\bsl{e}(\bsl {u})\,|\,\bsl {u}\in\bigl[C^\infty_\per(Q)\bigr]^2\bigr\}$ in the space $\bigl[L^2_\per(Q, {\rm d}\varkappa)\bigr]^3.$ A symmetric matrix 
						$\bsl{v}\in\bigl[L^2_\per(Q, {\rm d}\varkappa)\bigr]^3$ is said to be $\varkappa$-{\sl solenoidal} if 
							$$\int_Q\bsl{v}\cdot \bsl{e}(\bsl {u})\,{\rm d}\varkappa=0\ \ \ \ \ \ \ \ \forall\,\bsl {u}\in\bigl[C^\infty_\per(Q)\bigr]^2.$$
				
				Denoting by $V_\sol^\varkappa$ the set of $\varkappa$-solenoidal matrices, we can write (see {\it e.g.} \cite{bib30})
						% is denoted  and moreover, the following orthogonal decomposition holds:
							$
							\bigl[L^2_{\rm per}(Q, {\rm d}\varkappa)\bigr]^3=V_\pot^\varkappa\oplus V_\sol^\varkappa.
							%\ \ \ \ \ \ \bigl[\lqy\bigr]^3=V_\pot^\lambda\oplus V_\sol^\lambda.
							$ 
							It follows that the orthogonal decomposition 
					$
					%\bigl[\lqw\bigr]^3
					%\bigl[L^2(\Omega\times Q,\dx\times\dmu)\bigr]^2=\lpot\oplus\lsol,
					%\ \ \ \ \ 
					\bigl[L^2(\Omega\times Q,\dx\times{\rm d}\varkappa)\bigr]^3=L^2(\Omega, V_\pot^\varkappa)\oplus L^2(\Omega, V_\sol^\varkappa)
					$ 
					holds,
				%It has been shown (for example in Zhikov \cite {bib30}) that 
				where the {\sl two-scale} $L^2$-spaces of $\varkappa$-potential and $\varkappa$-solenoidal vector fields are the closures
				%$\lpot$ $\lsol$ are the closures in $\bigl[\lqw\bigr]^3$ 
				of the linear spans of matrices $w\bsl{e}(\bsl {u}),$ $w\in C_0^\infty(\Omega),$ $\bsl {u}\in\bigl[C^\infty_\per(Q)\bigr]^2$
				%. Moreover, $\lsol$ is the closure in $\lqw$ 								of the linear span of matrices 
				and $w\bsl{v},$ $w\in C_0^\infty(\Omega),$ $\bsl{v}\in V_\sol^\varkappa,$ 
				%$\bigl({\rm resp.} V_\sol^\lambda\bigr),$ 
				with respect to the norm of $\bigl[L^2(\Omega\times Q,\dx\times{\rm d}\varkappa)\bigr]^3.$  
				When $\varkappa$ is the Lebesgue measure on $Q,$ we simply write $V_{\rm pot},$ $V_{\rm sol},$ $\bigl[L^2(\Omega\times Q)\bigr]^3.$
				%on $\Omega\times Q.$ 
\end {Def}

			\subsection {Convergence on the stiff component}\label {sec222}
				%The majority of the following results are formulated in the established notation but can otherwise be found in Zhikov \& Pastukhova \cite {bib32}. 
				
				We first study the relationship between the limit functions $\bsl {u}\xy$ and $\widehat{\bsl{u}}\xy,$ see Section \ref{sec221}.
				%structure of the weak two-scale limit of the functions $\chi_1^\eph\ueph$. 
				%In order to say something about the limit of this sequence, the following theory will be needed:
					\begin {Def}
Denote $\boldsymbol{\psi}^h_\varepsilon:=\boldsymbol{\psi}^h(\cdot/\varepsilon),$ where $\boldsymbol{\psi}^h\in\bigl[\lqh\bigr]^2$ extended to ${\mathbb R}^2$ by $Q$-periodicity. 

1. We say that the sequence $\boldsymbol{\psi}^h_\varepsilon$ 
												%It is said that $\boldsymbol{\psi}^h(\bsl {x}/\ep)$ 
						{\sl weakly converges} to $\boldsymbol{\psi}\in\bigl[L^2_{\rm per}(Q, {\rm d}\mu)\bigr]^2,$ 
						%$\bigl[\lqeph\bigr]^2$ 
						and write $\boldsymbol{\psi}^h_\ep\wlq\boldsymbol{\psi},$ if 
							$$\int_Q\boldsymbol{\psi}^h_\ep\cdot \boldsymbol{\xi}(\cdot/\ep)
							\,\dmu_\ep^h\longrightarrow\int_Q\boldsymbol{\psi}
							%(\bsl {y})
							\cdot \boldsymbol{\xi}
							%(\bsl {y})
							\,\dmu
							%$$
						%for all $
						\ \ \ \ \ \forall\boldsymbol{\xi}\in\bigl[C_\per^\infty(Q)\bigr]^2,
						$$
						where the test function $\boldsymbol{\xi}$ is extended to ${\mathbb R}^2$ by $Q$-periodicity.
%In this case we write $\boldsymbol{\psi}^h(\bsl {x}/\ep)\wlq\boldsymbol{\psi}.
%(\bsl {y}).
%$						
					%\end {Def}
					%\begin {Def}
					
						2. We say that $\boldsymbol{\psi}^h_\ep$
						%(\bsl {x}/\ep)\in
						%$ be a function in $
		%				\bigl[\lqeph\bigr]^2$ 
						%. It is said that $\boldsymbol{\psi}^h(\bsl {x}/\ep)$ 
						{\sl strongly converge} to a function $\boldsymbol{\psi}\in\bigl[L^2_{\rm per}(Q, {\rm d}\mu)\bigr]^2,$
						%in $\lqeph^2$, 
						and write $\boldsymbol{\psi}^h_\ep\slq\boldsymbol{\psi},$
						%(\bsl {y})$ 
						if 
							$$\int_Q\boldsymbol{\psi}^h_\ep\cdot \boldsymbol{\xi}^h(\cdot/\ep)\,\dmu_\ep^h\longrightarrow\int_Q\boldsymbol{\psi}
							%(\bsl {y})
							\cdot \boldsymbol{\xi}
							%(\bsl {y})
							\,\dmu \ \ \ 
							%{\rm whenever}
							 \ \ {\rm if\ and\ only\ if}\ \ \ \boldsymbol{\xi}^h_\ep\wlq\boldsymbol{\xi}.$$
					\end {Def}
				%Hence, the following result follows.
				
					\begin {Prop}
					\label{basicprop}
						%Let $\ueph$ be a bounded sequence in $\bigl[\lomegaeph\bigr]^2,$ with weak two-scale limit $\bsl {u}\xy$. 
					 If  $\ueph(\bsl{x})\wtwo \bsl {u}(\bsl{x}, \bsl{y})$ 
					 %(see Appendix) 
					 and $\boldsymbol{\psi}^h_\ep\slq\boldsymbol{\psi},$  then
							$$\int_\Omega\ueph\cdot \boldsymbol{\psi}^h_\ep\vphi\,\dmu_\ep^h\longrightarrow\int_\Omega\int_Q\bsl {u}\xy\cdot \boldsymbol{\psi}(\bsl {y})\vphi(\bsl {x})\,\dmu(\bsl{y})\dx\ \ \ \ \ \ \forall\vphi\in C_0^\infty(\Omega).$$
					\end {Prop}
					\begin {proof}
						Since $\boldsymbol{\psi}^h_\ep\slq\boldsymbol{\psi},$ it follows that for all 
						$\bzeta\in\bigl[C_\per(Q)\bigr]^2$ the relation
							%%$$\limh\int_Q\bigl |\boldsymbol{\psi}^h-\bzeta\bigr |^2\,\dmu_\ep^h=\int_Q\left |\boldsymbol{\psi}-\bzeta\right |^2\,\dmu\ \ \ \forall\bzeta\in\bigl[C^\infty_\per(Q)\bigr]^2,$$
						%%and hence, by mean value property, one has 
							\begin {equation}\label {proof1}
								\lime\int_\Omega\bigl |\boldsymbol{\psi}^h_\ep-\bzeta(\cdot/\ep)\bigr|^2\,\dmu_\ep^h
								=|\Omega|\int_Q\left |\boldsymbol{\psi}-\bzeta\right |^2\,\dmu
							\end {equation}
							holds.
						%Consider the following expression:
							%\begin {multline*}
						%	\begin{equation}
						%		\left | \int_\Omega\ueph(\bsl {x})\cdot\boldsymbol{\psi}^h\fxep\vphi(\bsl {x})\ \dmu_\ep^h-\int_\Omega\ueph(\bsl {x})\cdot\bzeta\fxep\vphi(\bsl {x})\ \dmu_\ep^h\right |
								%=\\
						%		= \left | \int_\Omega\ueph(\bsl {x})\cdot\Big (\boldsymbol{\psi}^h\fxep-\bzeta\fxep\Big )\vphi(\bsl {x})\ \dmu_\ep^h\right |.
						%	\end{equation}
							%\end {multline*}
						Notice further that, by the H\"{o}lder inequality, one has
							%\begin {multline*}
							\begin{equation*}
								\left | \int_\Omega\ueph\cdot\big(\boldsymbol{\psi}^h_\ep-\bzeta(\cdot/\ep)\big)
								\vphi\,\dmu_\ep^h\right | 
								%\\ 
								\leq \max_\Omega |\vphi | \| \ueph\|_{[\lomegaeph]^2}\biggl ( \int_\Omega\bigl|\boldsymbol{\psi}^h-\bzeta(\cdot/\ep)\bigr|^2\dmu_\ep^h\biggr )^{1/2}.
\end{equation*}								
%							\end {multline*}
						The weak two-scale convergence of $\ueph$ and the relation (\ref {proof1}) imply that
							$$\limsup_{\ep\tends 0}\biggl| \int_\Omega\ueph\cdot\boldsymbol{\psi}^h_\ep\vphi\,\dmu_\ep^h-\int_\Omega\int_Q\bsl {u}\xy\cdot\bzeta(\bsl {y})\vphi(\bsl {x})\,\dmu(\bsl{y})\dx\biggr |\ \ \ \ \ \ \ \ \ \ \ \ \ \ \ \ \ \ \ \ \ \ \ \ \ \ \ \ \ \ \ \ \ \ \ \ \ \ \ \ \ \ $$
							$$=\limsup_{\ep\tends 0}\bigg | \int_\Omega\ueph\cdot\boldsymbol{\psi}^h_\ep\vphi\,\dmu_\ep^h-\int_\Omega\ueph\cdot\bzeta(\cdot/\ep)\vphi\,\dmu_\ep^h\bigg|
							%$$
							%$$
							\leq C\biggl ( \int_Q |\boldsymbol{\psi}
							%(\bsl {y})
							-\bzeta
							%(\bsl {y})
							|^2\,\dmu\biggr)^{1/2}\ \ \ \ \ \ \forall \bzeta\in\bigl[C_\per(Q)\bigr]^2.$$
						The claim now follows by choosing an approximation sequence 
						%to $\bzeta$, 
						$\bzeta=\bzeta_k$ 
						%say, 
						such that $\bzeta_k\tends \boldsymbol{\psi}$ in $\bigl[\lqx\bigr]^2$.
					\end {proof}			
				%Then the following theorem can be formulated with regards the limit on the stiff inclusions of $\Omega$.
					\begin {Thrm}
					\label{trace_theorem}
						%Let $\chi_1^\eph(\bsl {x})=\chi^h_1(\bsl {x}/\ep)$ be the characteristic function on the set $\whep_1$ and let $\ueph(\bsl {x})\wtwo\bsl {u}\xy$. Moreover, let $\widetilde {\bsl {u}}\xy$ be such that $\ueph\chi_1^\eph\wtwo\widetilde {\bsl 										{u}}\xy\in L^2(\Omega\times Q)^2$. Then
								%At each point $\bsl {y}\in F_1,$ t
								The function $\widehat{\bsl {u}}$ is the trace of $\bsl {u}$ on $F_1,$ in the sense 
								that $\bsl {u}\xy=\widehat{\bsl {u}}\xy$ a.e. $\bsl{x}\in\Omega,$ $\lambda$-a.e. $\bsl{y}\in F_1.$ 
							%%\begin {equation*}
							%\label {struc5}
							%%	\widetilde {\bsl {u}}\xy=\underset{\bsl {y}\in F_1}{\tr}\bsl {u}\xy
							%%\end {equation*}
						%%holds, where $\underset{\bsl {y}\in F_1}{\tr}\bsl {u}\xy$ denotes the trace of the function $\bsl {u}$ on the network $F_1$.
					\end {Thrm}
					\begin {proof}
						For all functions $\widehat{\boldsymbol{\psi}}\in\bigl[L^2_{\rm per}(Q,\dlambda)\bigr]^2$ and $h>0,$ we define
					%	let $\boldsymbol{\psi}^h(\bsl {x}/\ep)\slq \boldsymbol{\psi}(\bsl {y})$ where 
					%$\boldsymbol{\psi}$ is the function defined by
					\begin {equation}\label {psi1}
						\boldsymbol{\psi}(\bsl {y}):=
							\begin {cases}
								\widehat{\boldsymbol{\psi}}(\bsl {y}),\ &\ \bsl {y}\in F_1\cap Q,\\
								0,\ &\ \bsl {y}\in Q\backslash F_1,
							\end {cases}
							\ \ \ \ \ \ \ \ \ \ \ \ \ 
							[\boldsymbol{\psi}]^h(\bsl {y}):=\begin {cases}
							\sum_{I_j}\widehat{\boldsymbol{\psi}}^h_j(\bsl {y}),\ \ & \ \bsl{y}\in F_1^h\cap Q,\\
							0,\ &\ \bsl {y}\in Q\backslash F_1^h,
							\end {cases}
							%%%\begin {cases} 
							%%%....&\\
							%%%....&
%								\widetilde {\boldsymbol{\psi}}(\bsl {y}),\ & \ \ \bsl {y}\in F_1\cap Q,\\
%								0,\ &\ \ \bsl {y}\in Q\backslash F_1,
							%%%\end {cases}
					\end {equation}
					where the summation is carried out over all links $I_j$ of $F_1\cap Q,$ and for each link $I_j$ we set 
					$\widehat{\boldsymbol{\psi}}^h_j(\bsl{y})=\widehat{\boldsymbol{\psi}}(\bsl{y}^*)$ whenever $\bsl{y}$ is in 
					the $h$-neighbourhood of $I_j$ and $|\bsl{y}-\bsl{y}^*|={\rm dist}(\bsl{y}, I_j),$ $\bsl{y}^*\in I_j,$ and
					$\widehat{\boldsymbol{\psi}}^h_j(\bsl{y})=0$ otherwise.
						Notice that for all $\varphi\in C_0^\infty(\Omega)$ one has
							%\begin {multline}
							\begin{equation}
							\label {int1}
								\int_\Omega\ueph\cdot[\boldsymbol{\psi}]^h_\ep\varphi
								\,\dmu^h_\ep=\int_\Omega\ueph\chi_1^h(\cdot/\ep)\cdot[\boldsymbol{\psi}]^h_\ep\varphi\,\dmu^h_\ep
								%%%+ 
								%\\ +
								%%%\int_\Omega\ueph\chi_0^h(\cdot/\ep)
%								\big(1-\chi_1^h(\cdot/\ep)\big)
								%%%\cdot[\boldsymbol{\psi}]^h_\ep\varphi\,\dmu^h_\ep.
							\end{equation}
%Consider the following integral 
%%where $\boldsymbol{\psi}^h$ is the natural extension of the function given in formula (\ref {psi1}).									%\end {multline}
						%%Clearly, in the limit as $\ep\tends 0$, the left-hand side converges by definition and is given as
Due to the fact that $[\boldsymbol{\psi}]_\ep^h\slq\boldsymbol{\psi},$ the following convergence holds (see Proposition \ref{basicprop}):							
							\begin{equation}
							\int_\Omega\ueph\cdot[\boldsymbol{\psi}]^h_\ep\varphi\,\dmu^h_\ep\stackrel{\varepsilon\to0}{\longrightarrow}
							\int_\Omega\int_Q \bsl {u}\xy\cdot\boldsymbol{\psi}(\bsl {y})\varphi(\bsl{x})\,\dmu(\bsl{y})\dx=
							\frac{1}{2}\int_\Omega\int_Q \bsl {u}\xy\cdot\widehat{\boldsymbol{\psi}}(\bsl {y})\varphi(\bsl{x})\,\dlambda(\bsl{y})\dx.
							%,\ \ \ \ \ \ \ \varepsilon\to0.
							%=\frac {1}{2}\int_\Omega\int_{F_1\cap Q}\Big ( \underset{\bsl {y}\in F_1}{\tr}\bsl {u}\xy\Big )\cdot\widetilde {\boldsymbol{\psi}}(\bsl {y})\ \dl\dx.
							\label{oneside}
							\end{equation}
						Similarly, for the first integral on the right-hand side of (\ref{int1}), we obtain
							\begin{equation}
							\int_\Omega\ueph\chi_1^h(\cdot/\ep)\cdot[\boldsymbol{\psi}]^h_\ep\varphi\,\dmu^h_\ep
							\stackrel{\varepsilon\to0}{\longrightarrow}\frac {1}{2}\int_\Omega\int_{Q}\widehat{\bsl {u}}\xy\cdot\widehat{\boldsymbol{\psi}}(\bsl {y})\varphi(\bsl{x})\,\dl(\bsl{y})\dx.
							\label{otherside}
							\end{equation}
						%Hence, it only remains to be shown that 
						%%%Finally, the second integral on the right-hand side of (\ref{int1}) goes to zero as $\ep\tends 0,$ by virtue of the convergence 
							%%%$\ueph\chi_0^h(\cdot/\ep)
							%(1-\chi_1^h(\cdot/\ep)\bigr)
							%%%\wtwo\bsl {u}\xy\chi_0(\bsl {y}).$
						%where $\chi_0$ is the characteristic function of $Q\backslash F_1$. 
						It follows that the limits in (\ref{oneside}) and (\ref{otherside}) coincide, as required.
						%Hence the integral converges to zero as desired.
					\end {proof}
				The next theorem, proved in \cite{bib30}, describes the structure of the two-scale limit $\widehat{\bsl{u}}.$
				% of $\ueph$ 
				%on the stiff component. 
				Recall that on the stiff component $F_1^{h,\varepsilon}$ the symmetric gradient is bounded and hence
					$\ep\chi_1^{\varepsilon,h}\bsl{e}(\ueph)\tends 0$ in
					% \ \ \ \text {in}\ \ 
					$\bigl[\lopen{\whep_1}{\dlambda_\ep^h}\bigr]^3.$
				%Hence the following:
					\begin {Thrm}\label{Zh_ref}[Theorems 12.2, 12.3 and Lemma 9.6 in \cite{bib30}]
					\label{u0_chi_theorem}
						
						1. %Notice that 
							It follows from $
							%\begin {cases}
								\chi_1^{h, \varepsilon}\ueph\wtwo\widehat{\bsl{u}}\xy$ in $\bigl[\lopen{\whep_1}{\dmu_\ep^h}\bigr]^2$ and  
								%\ \ \ \ 
								%\ \ \  & \text {in}\ \ \lopen{\Omega}{\dmu_\ep^h}\ \ \ \ \ 
								%.\\
								$\ep\chi_1^{h, \varepsilon}\bsl{e}(\ueph)\tends 0%\ \ \  & \text {in}\ \ \lopen{\whep_1}{\dmu_\ep^h}.
							%\end {cases}
							$ in $\bigl[\lopen{\whep_1}{\dlambda_\ep^h}\bigr]^3,$
						  that $\forall\bsl{x}\in\Omega,$ $\lambda$-a.e. $\bsl{y}\in F_1$ one has 
						%$\bsl {u}\xy\in\lopen{\Omega}{\mathcal {R}}$. Moreover, 
						$\widehat{\bsl{u}}\xy=\bsl {u}_0(\bsl {x})+\boldsymbol{\chi}\xy$ where $\bsl {u}_0\in\bigl[H^1_0(\Omega)\bigr]^2$ and $\boldsymbol{\chi}\in\lopen{\Omega}{\widehat{\mathcal{R}}}.$
					%%\end {Thrm}
								%	Bearing this result in mind, the structure of the symmetric gradient on the stiff components is now given by the following result.
					%%\begin {Thrm}\label {stiffcon109}
						
						2. If the convergence $\chi_1^{h, \varepsilon}\ueph\wtwo\bsl {u}_0(\bsl {x})+\boldsymbol{\chi}\xy$ holds in $\bigl[\lopen{\whep_1}{\dlambda_\ep^h}\bigr]^2$ and the sequence  $\bigl\{\chi_1^{h,\varepsilon}\bsl{e}(\ueph)\bigr\}$ is bounded in
					% \ \ \ \text {in}\ \ 
					$\bigl[\lopen{\whep_1}{\dlambda_\ep^h}\bigr]^3,$ then, up to passing to a subsequence, one has $\bsl{e}(\ueph)\wtwo \bsl{e}(\bsl {u}_0(\bsl {x}))+\bsl{v}\xy$ in
							%\ \ \ \text {in}\ \ 
							$\bigl[\lopen{\whep_1}{\dlambda_\ep^h}\bigr]^3,$
						where $\bsl{v}\xy\in \lopen{\Omega}{V_\pot^\lambda}.$
					%3.It follows from Theorem \ref{u0_chi_theorem} and \cite[Lemma 9.6]{bib30} that 
							%%\begin {equation*}
							%\label {Zhikov1}
								%%\bsl{e}(\ueph)\wtwo \bsl{e}(\bsl {u}_0)+v\xy,\ \ \ \ \ \ \ \ \ v\xy\in \lopen {\Omega}{V_\pot}, 
							%\end {equation}
							%\begin {equation}\label {Zhikov2}
						%%%\end {equation*}
						%where
						%%%In addition, the convergence  
						%\begin {equation}

					Under the additional assumption that 
					\[
					\lim_{\varepsilon\to0}\int_{\Omega}A_1\bsl{e}(\ueph):\bsl{e}_{\bsl{y}}
					(\boldsymbol{\varphi})(\cdot/\varepsilon) w\,{\rm d}\lambda_\varepsilon^h=0\ \ \ \ \  \forall \boldsymbol{\varphi}\in\bigl[C^\infty_{\rm per}(Q)\bigr]^2,\ \ \ w\in C^\infty_0(\Omega), 
					\]
					the two-scale convergence
					$\chi_1^{h,\varepsilon}A_1\bsl{e}(\ueph)\wtwo A_1\bigl\{\bsl{e}\bigl(\bsl {u}_0(\bsl{x})\bigr)+\bsl{v}(\bsl{x},\bsl{y})\bigr\}$ holds in $\bigl[\lopen{\whep_1}{\dlambda_\ep^h}\bigr]^3,$ where the limit function is an element of $\lopen {\Omega}{V_\sol^\lambda}.$ 
					\end{Thrm}
					\begin{Rem} 
					Define the ``$\lambda$-homogenised'' tensor $A^{\rm hom}_\lambda$ by the minimisation problem
						\begin{equation}
						A^{\rm hom}_\lambda \boldsymbol{\xi}\cdot \boldsymbol{\xi}=\min_{v\in V_\pot^\lambda}\int_QA_1(\boldsymbol{\xi}+\bsl{v}):(\boldsymbol{\xi}+\bsl{v})\,\dlambda\ \ \ \ \forall\boldsymbol{\xi}\in{\rm Sym}_2,
						\label{ahom1}
						\end{equation}
						where ${\rm Sym}_2$ is the space of symmetric $(2\times2)$-matrices.					Theorem \ref{Zh_ref} implies that 
					$\chi_1^{h,\varepsilon}A_1\bsl{e}(\ueph)\rightharpoonup A^{\rm hom}_\lambda\bsl{e}(\bsl{u}_0)$ in the sense of the usual weak convergence in $\bigl[L^2(\Omega)\bigr]^3.$ 
					%Throughout the present paper we assume that $A^{\rm hom}_\lambda$ is positive-definite and that periodic rigid displacements take the form (\ref {rigid1}). 

					\end{Rem}
					%holds. 
							%\ \ \ \ \ 	A^\home \bsl{e}(\boldsymbol{\xi}):=\min_{v\in V_\pot}\int_QA_1\bigl(\bsl{e}(\boldsymbol{\xi})+v\bigr)\,\dmu.
							%\label{Ahomlim}
							%\end {equation}

					%\end {Thrm}
					
					%Combining Proposition \ref{techlem2}, Theorem \ref{korn1}
					The description of the structure of the two-scale limit of $\chi_1^{\ep,h}\bsl{u}_\varepsilon^h$
					% trace of $\bsl{u}$ on $F_1$ 
					is a consequence of several statements proved in \cite{bib32}.
					%, see Appendix A.
					 Combining this with Theorem \ref{trace_theorem}, we obtain the following result ({\it cf.} \cite[Theorem 3.1]{bib32}).
				\begin{Thrm}
				In the formula $\widehat{\bsl{u}}(\bsl{x}, \bsl{y})=\bsl{u}_0(\bsl{x})+\boldsymbol{\chi}(\bsl{x}, \bsl{y}),$ the transverse displacement $\boldsymbol{\chi}$
				%=\widehat{\bsl{u}}(\bsl{x}, \bsl{y}$ 
				is an element of the space $L^2(\Omega,\widehat{\mathcal R}^0).$ 
					\end{Thrm}

			\subsection {Convergence on the soft component}\label {sec223}
					\begin {Thrm}
					\label{soft_conv_theorem}
						For all sequences $\{\ueph\}\subset\bigl[H^1(\Omega)\bigr]^2$
						%%$\{\ueph\}\subset\bigl[C^\infty(\overline{\Omega})]$
						%\bigl[\lopen{\Omega}{\dmu^h_\ep}\bigr]^2$ 
						%%and 
						%%$\bigl\{\ep\chi_0^{\ep,h}\bsl{e}(\ueph)\bigr\}\subset\bigl[\lopen{\whep_0}{\dmu_\ep^h}\bigr]^3$ be bounded. 
						%sequences in $\bigl[\lopen{\whep_0}{\dmu_\ep^h}\bigr]^3.$  If 
						such that
						%i) 
						$\ueph\wtwo\bsl {u}\xy$ in $\bigl[\lopen{\whep_0}{\dmu_\ep^h}\bigr]^2$ and
						%ii) 
						$\ep\chi_0^{\ep,h}\bsl{e}(\ueph)\wtwo\widetilde {\bsl{p}}\xy$ in $\bigl[\lopen{\whep_0}{\dmu_\ep^h}\bigr]^3,$ 
						%and
											%%	iii) $\ep\chi_1^{\ep,h}\bsl{e}(\ueph)\to\bsl{0}$ in $\bigl[\lopen{\whep_1}{\dmu_\ep^h}\bigr]^3,$ 
						%\noindent{
						one has  $\bsl{u}\in\bigl[L^2(\Omega, H^1(Q))\bigr]^2$ and 
						%	\begin {equation}
													%\label {struc8}
							%$\ep \bsl{e}(\ueph)\wtwo\widetilde 
							$\widetilde{\bsl{p}}\xy={\bsl{e}}_{\bsl{y}}\bigl(\bsl {u}\xy\bigr)$ a.e. $\bsl{x}\in\Omega,$ $\bsl{y}\in Q.$
							%}
							%\in \lopen{\Omega}{V_\pot}.$
							%\end {equation}
					\end {Thrm}
					\begin {proof}
						%Let $\bsl {a}$ and $\bsl{b}$ be such 
						For each $\delta>0,$ consider a $C^\infty$-domain $Q_\delta$ such that $F_0^{2\delta}\cap Q\subset Q_\delta\subset F_0^\delta\cap Q$ 
						and the set
						%Suppose that $\bsl{b}^h\in\bigl[C^\infty(F_0\cap Q^h)\bigr]^3,$ $ ... $ 
						%$\bsl{b}\in[L^2(Q)\bigr]^3$ such that ${\rm div}\,\bsl{b}\in[L^2(Q)\bigr]^2.$
					$$
					%\bsl {a} \in
					{\mathcal X}_\delta:=\{
					%\dive\,\bsl{b}|\,
					\bsl{b}\in\bigl[C^\infty(Q_\delta)\bigr]^3:\ \bsl{b}\,\bsl{n}\vert_{\partial Q_\delta}=0\bigr\},
					$$
%					Note that \bsl{b}\in$
%					$$					
%					=
%					\biggl\{\bsl{a}\in C^\infty(Q_\delta):\,\int_{Q_\delta}\bsl{a}=0\biggr\},
%					$$
						where $\bsl{n}$ is the unit normal to $\partial Q_\delta.$ 
						For all $\bsl{b}\in{\mathcal X}_\delta,$ $\bsl{a}={\rm div}\,\bsl{b}$ in $Q_\delta,$ consider the functions 
						%$\widetilde{\bsl{b}},$ $\widetilde{\bsl{a}}$  by zero on $Q\setminus Q_\delta.$
		%%%				\] 
					\begin {equation}
					%\label {psi1}
					\widetilde{\bsl{a}}(\bsl {y}):=					
					\begin {cases} 
					\bsl{a}(\bsl {y}),\ &\ \bsl {y}\in Q_\delta,\\
					0,\ &\ \bsl{y}\in Q\backslash Q_\delta,
			\end {cases}
							\ \ \ \ \ \ \ \ \ \ \ \ \ 
						\widetilde{\bsl{b}}(\bsl{y}):=
						\begin {cases}
						\bsl{b}(\bsl {y}),\ &\ \bsl {y}\in Q_\delta,\\
					    0,\ &\ \bsl{y}\in Q\backslash Q_\delta,
							\end {cases}
							\end{equation}
			extended to ${\mathbb R}^2$ by $Q$-periodicity.
								%%				and 
						%all $\varepsilon, h:$
						Then for sufficiently small $\varepsilon>0$ (recall that $h\to0$ as $\varepsilon\to0$) the following identity holds:
							\begin{equation}
							\ep\int_{\whep_0}\widetilde{\bsl{b}}(\cdot/\varepsilon):\bsl{e}(\boldsymbol{\psi})\,\dmu_\ep^h=-\int_{\whep_0}\widetilde{\bsl {a}}(\cdot/\varepsilon)
							\cdot\boldsymbol{\psi}\,\dmu_\ep^h\ \ \ \ \forall\boldsymbol{\psi}\in\bigl[H^1_0(\Omega)]^2.
							\label{ab_identity}
							\end{equation}
		Setting $\boldsymbol{\psi}=\varphi\ueph,$ $\varphi\in C_0^\infty(\Omega),$ in (\ref{ab_identity}) 
		%$\vphi\in C^\infty_0(\Omega)$, 
		yields\footnote{Throughout, we use the notation $\otimes$ for the symmetrised 
				%version of the
			 tensor product.}
							%\begin {multline*}
							\begin{equation*}
								\ep\int_{\whep_0}\widetilde{\bsl{b}}(\cdot/\varepsilon)\vphi: \bsl{e}(\ueph)\,\dmu_\ep^h+\ep\int_{\whep_0}\widetilde{\bsl{b}}(\cdot/\varepsilon): 
							%\frac {1}{2}\bigl(\ueph\otimes \del\vphi+\del\vphi\otimes\ueph\bigr)
								(\ueph\otimes \del\vphi)\,\dmu_\ep^h
								= 
								%\\ =
								-\int_{\whep_0}\widetilde{\bsl {a}}(\cdot/\varepsilon)\cdot\vphi\ueph\,\dmu_\ep^h.
							\end{equation*}
							%\end {multline*}
						%Using the assuof the theorem when 
						Passing to the limit in the last identity as $\ep\tends 0$ and using the fact that $\widetilde{\bsl{a}},$ $\widetilde{\bsl{b}}$ vanish in $Q\setminus Q_\delta,$
						%and using (\ref{basicprop}), 
						we obtain
						%, we obtain
						%the following equality is obtained:
							$$\int_\Omega\int_{Q_\delta}\widetilde {\bsl{p}}\xy\vphi(\bsl {x})\cdot\bsl{b}(\bsl {y})\,\dy\dx=-\int_\Omega\int_{Q_\delta}\bsl {u}\xy\vphi(\bsl {x})\cdot \bsl {a}(\bsl {y})\,%\dmu(\bsl{y})
							\dy\dx.$$
						As $\vphi\in C_0^\infty(\Omega)$ is arbitrary, it follows that
							\begin{equation}
							\int_{Q_\delta}\widetilde {\bsl{p}}\xy\cdot\bsl{b}(\bsl {y})\,\dy=-\int_{Q_\delta}\bsl {u}\xy\cdot \bsl {a}(\bsl {y})\,\dy
							%\dmu(\bsl{y})
							\ \ \ \ \ {\rm a.e.}\  \bsl{x}\in\Omega. 
							%\ \ \ \ \forall\bsl{a}\in{\mathcal X}.
							\label{y_identity}
							\end{equation}
						Taking divergence-free fields $\bsl {b}\in{\mathcal X}_\delta$ in (\ref{y_identity}) we infer (see {\it e.g.} \cite{DL})
						%%%${\rm curl}\,\widetilde {\bsl{p}}(\bsl{x}, \cdot)=0$ {\rm a.e.} $\bsl{x}\in\Omega.$
						%\in\lopen{\Omega}{V_\pot}.$ 
						%Consider a potential $\bsl {v}$ as seen in (\ref {pot1}) for the matrix $\widetilde {p}$ such that 
						%Further, consider $\bsl {v}$ such that 
						%%It follows  that 
					the existence of  $\bsl{v}\in\bigl[L^2(\Omega, H^1(Q_\delta))\bigr]^2$ such that 
						$\widetilde {\bsl{p}}(\bsl{x}, \bsl{y})={\bsl{e}}_{\bsl{y}}\bigl(\bsl {v}(\bsl{x}, \bsl{y})\bigr),$ $\bsl{y}\in Q_\delta,$ which implies 
						%t follows that
							$$\int_{Q_\delta}\bsl {v}\xy\,\cdot\,\bsl {a}(\bsl {y})\,\dy=\int_{Q_\delta}\bsl {u}\xy\,\cdot\,\bsl {a}(\bsl {y})\,\dy
							%\dmu(\bsl{y})
							\ \ \ \ \ {\rm a.e.}\  \bsl{x}\in\Omega,\ \ \ \ \ \ \ \ \ \ \ \ \ \ \ \ \ \ \ \ \ \ \ \ \ \ \ \ \ \ \ \ \ \ \ \ \ \ \ \ \ \ \ \ \ \ \ \ \ \ \ \ \ \ \ \ \ \ 
							$$
							$$ 
							\ \ \ \ \ \ \ \ \ \ \ \ \ \ \ \ \ \ \ \ \ \ \ \ \ \ \ \ \ \ \ \ \ \ \ \ \ \ \ \ \ \ \ \ \ \ \ \ \ \ \ \ \ \ \ \ \ \ \ \ 
							\ \ \ \ \forall\,\bsl{a}\in\bigl\{{\rm div}\,\bsl{b}|\,\bsl{b}\in{\mathcal X}_\delta\bigr\}=\biggl\{\bsl{a}\in\bigl[C^\infty(Q_\delta)\bigr]^2:\,\int_{Q_\delta}\bsl{a}(\bsl {y})\,\dy=0\biggr\}.
							$$
						Using the density in $\bigl[L^2(Q_\delta)\bigr]^2$ of vector functions $\bsl{a}$ having the above representation	
						%${\mathcal X}$ 
						implies that 
						%of functions $\bsl{a}$ that admit representation
							$\bsl {v}\xy$ and $\bsl {u}\xy$ differ by a constant for $\bsl{y}\in Q_\delta,$
							%+{\rm const},$
							%\langle \bsl {u}\rangle,$ 
							hence 
							$\widetilde{\bsl{p}}={\bsl{e}}_{\bsl{y}}(\bsl {v})={\bsl{e}}_{\bsl{y}}(\bsl {u}),$ a.e. $\bsl{y}\in Q_\delta.$ By virtue of the arbitrary choice of 
							the parameter $\delta,$ we conclude that $\widetilde{\bsl{p}}={\bsl{e}}_{\bsl{y}}(\bsl {u})$ for a.e. $\bsl{y}\in Q.$
						%as required.
					\end {proof}

		\section {
		%Statement and Proof of the 
		Homogenisation theorem}\label {sec23}
			%%%The proof of the homogenisation theorem is similar to the proof of the corresponding homogenisation theorem in %Zhikov \& Pastukhova 
			%%%\cite {bib32}. However, modifications in the structure of the extension functions are required to prove the result in question.
			% will be needed to accommodate the soft inclusions not found in the aforementioned work above
			%%%Accordingly, the limit equation  
			%not found in similar homogenisation problems; this characteristic 
			%%%includes two coupled microscopic equations which uniquely determine the function $\bsl {U}$ on the soft inclusions and its trace $\boldsymbol{\chi}$ on the limit network, see Section \ref{sec21}. 
			In what follows, we consider the case of the framework $F_1$ shown in Fig. \ref{critthick2} (``model framework''). However, the analysis presented is readily extended to any framework such that the representation (\ref{rigid1}) holds for periodic rigid displacements, with obvious modifications in the statements.
			
			 \begin {figure}[h!]
						\centering
						\includegraphics[trim=2cm 16cm 2cm 1cm width=7cm,height=0.25\textheight]{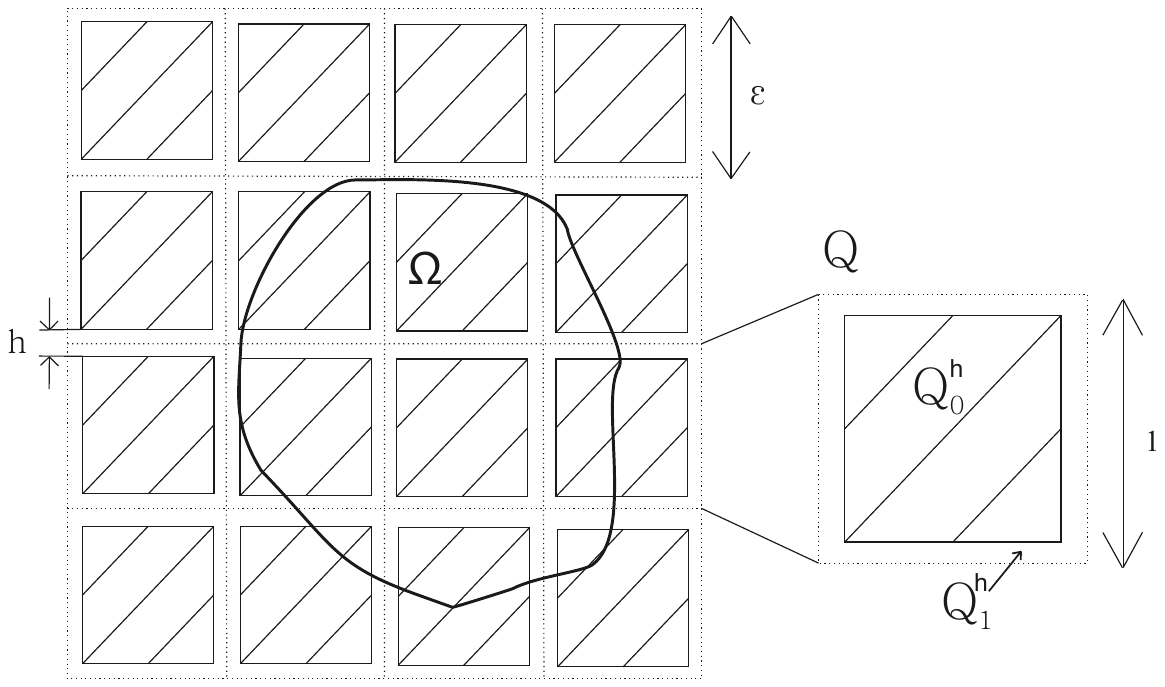}
						\caption {Periodic network with high contrast, where $Q_j^h:=F_j^h\cap Q,$ $j=0,1.$}
						\label {critthick2}
					\end {figure}

			\subsection {Homogenised system of equations}\label {sec231}
				%An important space of consideration is the so called 
				%%%We introduce the ``energy space'' 
				%denoted here as $V$. The space 
			%%%	$V$ containing the 
				%functions that satisfy the required properties of the 
		%%%		limit functions.
					\begin {Def} 
					\label{V_definition}
					We denote by $V$ the {\sl energy space} consisting of vector functions $\bsl{u}$ such that
						$$\bsl {u}\xy=\bsl {u}_0(\bsl {x})+\bsl {U}\xy,\ \ \ \bsl {u}_0\in\bigl[H_0^1(\Omega)\bigr]^2,\ \ \ 
						\bsl {U}\in\bigl[L^2\bigl(\Omega, H^1_{\rm per}(Q)\bigr)]^2,\ \ \ \ 
						%\ \ \ \text {where}\ 
						$$
						$$
						%\underset{\bsl {y}\in F_1}{\tr}
						\bsl {U}\xy=\boldsymbol{\chi}\xy,\ \ \ {\rm a.e.}\ \bsl{x}\in\Omega,\ \ \ \lambda{\text-}{\rm a.e.\ \ }\bsl{y}\in F_1\cap Q,\ \ \ \ \ \ \ \boldsymbol{\chi}\in L^2\bigl(\Omega,\widehat{\mathcal{R}}^0\bigr).$$
						We also denote by $\mathfrak{H}$ the closure of $V$ in $\bigl[L^2(\Omega\times Q, \dx\times\dmu)\bigr]^2.$
						%from Definition 
						%\ref{V_definition}, and 
					\end {Def}
				%The homogenised problem will now be introduced (to be proven in Section \ref {sec233}). 
			For a given $\bsl{f}\in\bigl[L^2(\Omega\times Q, \dx\times\dmu)\bigr]^2$, we refer to $\bsl {u}\in V$ as the {\sl solution of the homogenised problem} if 
				%the following integral equation is satisfied 
					\begin {multline}
					\label {hom21}
						\frac{1}{2}\int_\Omega A^\home_\lambda\bsl{e}(\bsl {u}_0):\bsl{e}(\boldsymbol{\varphi}_0)\,\dx+\frac {\theta^2}{6}\int_\Omega\int_QK_1\boldsymbol{\chi}''\cdot \bsl {\Phi}''\,\dl\dx + \frac {1}{2}\int_\Omega\int_QA_0{\bsl{e}}_{\bsl{y}}(\bsl {U}):{\bsl{e}}_{\bsl{y}}(\bsl {\Phi})\,\dy\dx
						\\
						+\int_\Omega\int_Q\big ( \bsl {u}_0+\bsl {U}\big )\cdot \bphi\,\dmu\dx = 
					\int_\Omega\int_Q \bsl {f}\cdot \bphi\,\dmu\dx\ \ \ \ \ \forall\bphi\xy=\bphi_0(\bsl {x})+\bsl{\Phi}\xy\in V,
					\end {multline}
					%for all functions $\bphi\xy=\bphi_0(\bsl {x})+\bsl{\Phi}\xy\in V.$ 
%%					Here, $A^\home$ is the homogenised tensor defined by
					%and is defined by the relation
				where ${\bsl{e}}_{\bsl{y}}(\bsl{u})$ ({\it cf.} (\ref{sym_grad})) denotes the symmetric gradient of $\bsl{u}=(\bsl{x}, \bsl{y})$ with respect to the variable $\bsl{y}$, the tensor $A^{\rm hom}_\lambda$ is given by (\ref{ahom1}), and $K_1$ is a function on the network $F_1\cap Q,$ defined on each link of the network by 
					%\begin {equation}\label {ahom1}
					%	A^\home_\mu\boldsymbol{\xi}\cdot \boldsymbol{\xi}=\min_{v\in V_\pot^\mu}\int_QA_1(\boldsymbol{\xi}+v)\cdot (\boldsymbol{\xi}+v)\,\dmu\ \ \ \forall\boldsymbol{\xi}\in{\rm Sym}_2,\ \ \ \ 
						$$
						K_1:=(A_1^{-1}\boldsymbol{\eta}\cdot\boldsymbol{\eta})^{-1}.\ \ \ \ \ \ 
						%$ $
						\boldsymbol{\eta}:=\boldsymbol{\tau}\otimes\boldsymbol{\tau},
						$$ 
						Here $\boldsymbol{\tau}$ is the tangent to the current link, and the prime denotes the tangential derivative, as in Definition \ref{rigidstuff}, {\it e.g.} $\boldsymbol{\chi}':=(\boldsymbol{\tau}\cdot\nabla(\boldsymbol{\chi}\cdot\boldsymbol\nu))\boldsymbol{\nu}.$
					%\end {equation}
%%				and the constant $K_1=\langle A_1^{-1}\eta\cdot\eta\rangle^{-1}$, $\eta =\boldsymbol{\tau}\otimes\boldsymbol{\tau}$ and $\boldsymbol{\tau}$ is the direction along the link.

				%The rest of this section is devoted to writing 
				The identity (\ref {hom21}) is equivalent to  a system of %partial differential equations 
				partial differential equations,
				% with accompanying conditions that define the solution uniquely. 
				%This 
				which is obtained by considering various classes of test functions in (\ref{hom21}). First, taking
				%a restriction where test 
				functions of the form $\bphi\xy=\bphi_0(\bsl {x})$ yields ({\it cf.} (\ref{homogenised}), where $\bsl{f}\in[L^2(\Omega)]^2$):
				 %are considered. 
				%the following macroscopic equation is obtained:
					\begin {equation}\label {macro1}
						-\frac{1}{2}\dive\bigl(A^\home_\lambda\bsl{e}(\bsl {u}_0)\bigr)+\bsl {u}_0+\langle\bsl {U}\rangle=\langle\bsl{f}\rangle,
						%\int_Q\bsl {f}(\cdot,\bsl{y})\dy,\ \ \ \bsl {u}_0\in \bigl[H_0^1(\Omega)\bigr]^2,
					\end {equation}
					where the angle brackets denote microscopic averaging, {\it i.e.} $\langle\bsl{U}\rangle:=\int_Q\bsl{U}(\cdot, \bsl{y})\dmu(\bsl{y}).$ 
					
					Further, we set $\bphi\xy=\varphi(\bsl{x})\boldsymbol{\Psi}(\bsl{y}),$ with $\varphi\in C_0^\infty(\Omega),$ $\boldsymbol{\Psi}\in\widetilde{V},$ where the space $\widetilde{V}$ consists of functions in $\bigl[H^1_{\rm per}(Q)\bigr]^2$ whose trace on $F_1\cap Q$ coincides with a rigid-body motion $\lambda$-a.e.
					 %are considered where two restrictions of $\Bphi$ will be examined. 
					 Assume for simplicity that the tensor $A_0$ is isotropic, {\it i.e.} 
				%%A tensor $A_0$ is said to be isotropic if 
				 %it is understood	that for a 
				 for all $\boldsymbol{\xi}\in{\rm Sym}_2$ one has
					\begin{equation}
					A_0\boldsymbol{\xi}=2M_0\boldsymbol{\xi}+L_0(\Tr\boldsymbol{\xi})I,\qquad M_0, L_0>0.
					\label{tensor_isotropic}
					\end{equation}
					 Taking first functions $\boldsymbol{\Psi}\in\bigl[C_0^\infty(F_0\cap Q)\bigr]^2,$ 
					 %({\it i.e.} the interior of the cell $Q$ 
					 %in the case of the 
					 %model framework) 
					 we obtain 
					 \begin {equation}
						-\frac{1}{2}M_0\bsl {\Delta}_{\bsl{y}}\bsl {U}-\frac{1}{2}(L_0+M_0)\del_{\bsl{y}}\dive_{\bsl{y}}\bsl {U}+\bsl {u}= P_{\mathfrak{H}}\bsl {f},\label{other1}
					\end {equation}
					\begin {equation}
						 \bsl {U}(\bsl {x},\cdot)\in\bigl[H^1_\per(Q)\bigr]^2,\ \ \bsl{x}\in\Omega,\ \ \ \ \ \ \ \bsl {U}\xy=\boldsymbol{\chi}\xy,\ \ \ \bsl{x}\in\Omega,\ \ \lambda{\text -}{\rm a.e.}\ \bsl{y}\in F_1,\ \ \ \ \ \ \boldsymbol{\chi}\in L^2\bigl(\Omega,\widehat{\mathcal R}^0\bigr),\label{other_cond}
					\end {equation}
					where $P_{\mathfrak{H}}$ is the orthogonal projection operator from $\bigl[L^2(\Omega\times Q, \dx\times\dmu)\bigr]^2$ onto $V.$ The equation (\ref{other1}) is a consequence of the fact that on smooth functions the following operator equation holds:
					%for smooth functions $\bsl {u}$ one has
					\begin{equation}
\dive A_0{\bsl{e}}=M_0\bsl {\Delta}+(L_0+M_0)\del\,\dive,
\label{isotropic}
\end{equation}
 which is the result of a direct calculation from (\ref{tensor_isotropic}) and the definition (\ref{sym_grad}) of the symmetric gradient operator ${\bsl{e}}.$ 
 %The formula 
%(\ref{isotropic}) is found in most books on elasticity theory. 
					%network $F_1$ followed by test functions supported on the soft component. 
				%It can be shown (see Appendix \ref {sb4}) that more general frame networks can be considered, however, considerations are initially restricted to the case of a single link $I$ aligned with the 						$Oy_2$ axis. Moreover, i
				
				 Finally, taking arbitrary $\boldsymbol{\Psi}\in\widetilde{V}$ yields additional equations coupling the framework $F_1\cap Q$ and the inclusion 
				 component $F_0\cap Q.$ For example, on the links parallel to the $y_2$-axis we obtain
					%we obtain the following system of PDEs:
	%%				the desired system of PDE's obtained take the form 
					\begin{align}
	%				\begin{aligned}
						\frac{\theta^2K_1}{6}\dpar_2^4\chi+\frac{1}{2}(L_0+2M_0)\dpar_1U_1+\bigl((u_0)_1+\chi\bigr)=(P_{\mathfrak{H}}\bsl{f})_1,  \label{wentzel1a}\\[0.5em]
						\frac{1}{2}(L_0+M_0)\dpar_2\chi+\frac{1}{2}M_0\dpar_1U_2+(u_0)_2=(P_{\mathfrak{H}}\bsl{f})_2, \label{wentzel1b}
						%\end{aligned}
					\end{align}
					where $(U_1, U_2)=\bsl{U},$ and $(\chi, 0)=\boldsymbol{\chi}.$
					%where $\bsl {\Delta}$ denotes the vector Laplacian. By isotropic tensor, it is understood that for a symmetric matrix $\boldsymbol{\xi}$,
					%\begin {equation}\label {isotropic}
					%	A_0\boldsymbol{\xi}=k_1\boldsymbol{\xi}+k_2(\Tr \boldsymbol{\xi})I\footnote {The constants $k_1,\ k_2$ are known as the Lam\'{e} constants and are usually denoted $\lambda$ and $\mu$. They have been relabeled $k_1$ and $k_2$ to save confusion with the notation for the measures 
				%		used in this work and to align with the notation used in \cite {bib32}}.
				%	\end {equation}
				%Note that $\bsl {x}$ plays the role of a parameter in the above two equations. This result can be generalised to the case when a square periodic frame network $F_1$ is considered as illustrated by Figure \ref {critthick2}  (with $A_0$ still isotropic). Consider such 
				Indeed, it is a straightforward consequence of (\ref{sym_grad}) and (\ref{tensor_isotropic}) that the boundary terms in the integration by parts in (\ref{hom21}) involve the expression
				\begin{equation}
				{\bsl{e}}_{\bsl{y}}(\bsl {U})_{ij}n_j\Psi_i=\bigl\{M_0(\partial_1U_1+\partial_1U_1)+L_0\partial_1U_1\bigr\}\Psi_1+\bigl\{M_0(\partial_2U_1+\partial_1U_2)+L_0\partial_2U_1\bigr\}\Psi_2,
				\label{expr_imp}
				\end{equation}
				where $n_j=1,$ $j=1,$ and $n_j=0,$ $j=2,3,$ are the components of the normal to the link. Setting first $\Psi_2=0,$ then $\Psi_1=0$ in (\ref{expr_imp}), we obtain the second term in (\ref{wentzel1a}), and the first and second term in (\ref{wentzel1b}), respectively. The equation (\ref{wentzel1a}), %(\ref{wentzel1b}), 
				viewed as boundary conditions on the function $\bsl{U}$ satisfying (\ref{other1})--(\ref{other_cond}), belongs to the class of so-called Ventcel' conditions \cite{Ventcel'}. Such conditions feature a coupling between the flux across the boundary (the term with $\partial_1U_1$) and diffusion along the boundary (the term with $\partial_2^4\chi$).
				%The expression (\ref{expr_imp}) is
				
				For a general periodic framework $F_1,$ on each link there is a positively orientated pair of vectors $\boldsymbol{\tau},\ \boldsymbol{\nu}$ with $\boldsymbol{\tau}$ pointing along the link and $\boldsymbol{\nu}$ orthogonal to the link. The corresponding version of the equations 
				(\ref {wentzel1a})--(\ref{wentzel1b}) on each link of $F_1$ is as follows:
					\begin {equation*}
					%\label {wentzel2}
					\begin{aligned}
						\frac {\theta^2K_1}{6}\dpar_\tau^4\chi+\frac{1}{2}(L_0+2M_0)\dpar_\nu U^{(\nu)}+\bigl(u_0^{(\nu)}+\chi\bigr)=(P_{\mathfrak{H}}\bsl{f})^{(\nu)},\\[0.5em]
					\frac{1}{2}(L_0+M_0)\dpar_\tau\chi+\frac{1}{2}M_0\dpar_\nu U^{(\tau)}+u_0^{(\tau)}=(P_{\mathfrak{H}}\bsl{f})^{(\tau)},	
					\end{aligned}
					\end {equation*}
				where  $\dpar_\tau,$ $\dpar_\nu$ denote differentiation along the link and 
				%$\dpar_\nu$ denotes differentiation 
				in the direction normal to the link, $\chi:=\boldsymbol{\chi}\cdot\boldsymbol{\nu},$  and the superscripts $(\tau),$ $(\nu)$ are understood in the sense of the notation introduced at the end of Section \ref{sec21}.
				%This equation must be satisfied on each of the four links forming the square frame. In Appendix \ref {sb4} it is shown that if							any general periodic framework $F_1$ is considered, then equation (\ref {wentzel2}) is in fact the equation which must be satisfied on each link without further amendment.

			\subsection {Extension theorem}\label {sec232}
Before proving the main result, we recall the description of a class of functions that extend periodic rigid displacements in 
$\widehat{\mathcal R}^0$ on the framework $F_1$ to the rod network $F_1^h,$ introduced in \cite{bib32}.
%%%The results of this section are essential to proving the homogenisation theorem in Section \ref {sec233}. We are mainly interested in functions $\bsl {G}\in H^1_\per(Q)$ that have trace $\bsl {g}\in \widehat{\mathcal R}^0$ on $F_1.$
%%%%%%%%%%%%%%%%%%%%%%%%%%%%%%%
%%%%%%%%%%%%%%%%%%%%%%%%%%%%%%%%
				%$\bsl {g}(\bsl {y}).$ 
				%is a periodic rigid displacement in $$. 
				%However, additional information about the space $\widehat{\mathcal R}^0$ is required in order to formulate the homogenised equation (\ref {hom21}). The following results have analogous 						counterparts in Zhikov \& Pastukhova \cite {bib32} which have been modified to the case of high-contrast problems.
					\begin {Def} Let $D$ denote the set of functions $\bsl {g}\in\widehat{\mathcal R}^0$ such that:
						\begin {enumerate}
							\item The function $\bsl {g}$ is infinitely smooth outside a neighbourhood of the nodes of the network $F_1;$
							\item In a neighbourhood $B_\delta({\mathcal O}):=\bigl\{\bsl{y}: \vert\bsl{y}-{\mathcal O}\vert<\delta\bigr\},$ $\delta>0,$ of each node ${\mathcal O}$ the function $\bsl {g}$ takes the form
								$\bsl {g}(\bsl {y})=C\bigl(\boldsymbol{\omega}(\bsl {y})-\boldsymbol{\omega}({\mathcal O})\bigr),$ $\bsl{y}\in F_1,$ where $C$ is a constant, $\boldsymbol\omega(\bsl {y}):=(-y_2,y_1)$.
						\end {enumerate}
					\end {Def}
				%The above equality actually means that $\bsl {g}(\bsl {y})\cdot\boldsymbol{\nu}_i=C(\bomega(\bsl {y})-\bomega(\bsl {y}_0))\cdot\boldsymbol{\nu}_i$ for each normal $\boldsymbol{\nu}_i$ orthogonal to the links meeting at the node $\bsl {y}_0$ with the constant $C$ fixed. The following result shows in 					fact that functions in $\widehat{\mathcal R}^0$ are approximated by functions in $D$.
					The following two statements are proved in \cite{bib32}.
					\begin {Prop}
						The set $D$ is dense in the space $\widehat{\mathcal R}^0$ with the respect to the norm of $\bigl[L^2_{\rm per}(Q,\dlambda)\bigr]^2.$
					\end {Prop}
					%\begin {proof}
						%The proof proceeds by showing that functions in $D$ can be used to approximate functions in $\widehat{\mathcal R}^0$. 
						%%Since smooth functions are dense in $H^2(I)$, where $I$ is a closed interval, it suffices to consider a neighbourhood of a node. Assume that there are $M$ intervals of length $1/4$ with end-point $\bsl {y}_0=\bsl {0}$. Let $\bsl {g}\in\widehat{\mathcal R}^0$. Then
							%%$\bsl {g}(\bsl {y})\cdot\boldsymbol{\nu}_i=g_i(t),$ $t=\boldsymbol{\tau}_i\cdot\bsl {y}\in[0,1/4],$
							%,\ \ \ 0\leq t\leq \frac {1}{4},$$
						%%where $\boldsymbol{\nu}_i$ and $\boldsymbol{\tau}_i$ are the unit normal and unit tangent for the $i$th link. By definition, $g_i\in H^2(0,1/4)$ and
						%%	$g_j(0)=0,$ $g_j'(0)=C,$ $j=1,2,...,M.$
					 %%It follows that the functions $g_i(t)-Ct$ can be approximated in $H^2(0,1/4)$ by functions with support in a neighbourhood of $t=0$. Hence, the functions $g_i(t)$ are approximated by $\bar {g}_i(t)$ that 								are linear in a neighbourhood of $t=0$. Finally, define a vector $\bar {\bsl {g}}$ by the equality $\bar {\bsl {g}}(\bsl {y})\cdot\boldsymbol{\nu}_i=\bar {g}_i(t)$. It follows from the above that in a neighbourhood of the node $\bsl {y}_0=\bsl {0}$, one has
						%%	$\bar {\bsl {g}}(\bsl {y})\cdot\boldsymbol{\nu}_i=Ct=C\boldsymbol{\tau}_i\cdot\bsl {y}=C\bomega(\bsl {y})\cdot\boldsymbol{\nu}_i.$
						%Hence the result.
				%%	\end {proof}
			%%	The next statement, see \cite{bib32}, concerns extension functions used in proving the homogenisation theorem at the end of this section. 
				%The proof \cite {bib32} but nonetheless are given for completeness of the text.
					\begin {Prop}\label {lemma12}
						For each $\bsl {g}\in D$, there exists a smooth extension $\bsl {g}^h=\bsl {g}^h(\bsl {y})$ to the network $F_1^h$ with the following properties: 
						\begin {enumerate}
							\item For each node ${\mathcal O}_k$ of $F_1\cap Q,$ the symmetric gradient ${\bsl{e}}_{\bsl{y}}(\bsl {g}^h)$ is zero in 
							%a neighbourhood 
							$B_{\delta_k}({\mathcal O}_k)$ for some $\delta_k>0.$ 
							\item For each $h>0$ and for each link $I$ of $F_1\cap Q$ we set%$\sigma={\rm Sp}(\bsl {y})$ is defined by relation 
								\begin{equation}
								\boldsymbol{\sigma}^h(\bsl{y}):=
								%%\begin {cases}
								%%%\gamma(\cdot/\ep),\ \ \ \
								%\ \ \ \text {where}\ 
								%%%\ \gamma(\bsl {y}):=
								\bigl(h^{-1}\boldsymbol{\nu}\cdot({\mathcal O}-\bsl {y})\bigr)(\boldsymbol{\tau}\otimes\boldsymbol{\tau}),\ \ \ \ \ \bsl{y}\in I^h
								\setminus\bigl(\cup_{k}B_{\delta_k}({\mathcal O}_k)\bigr),
								%%%\\
								%{\rm\ on\ the\ }h{\text-}{\rm rod},\\
								%%%\end{cases}
							\label{sigma_h}
							\end{equation}
							where $\boldsymbol{\tau},$ $\boldsymbol{\nu}$ are the unit tangent and normal to the link $I,$ $I^h$ is the $h$-neighbourhood of $I,$ ${\mathcal O}$ is either of the two end-points of $I,$
							and the union is taken over all nodes of $F_1\cap Q.$  Consider also the ``network-to-rod extension'' 
								$\bigl[(\bsl {g}\cdot\boldsymbol{\nu})''K_1\bigr]^h$ of the function 
								%%{\rm\ is\ the\ network{\text-}to{\text-}rod\ extension,\ see\ (\ref{psi1}),\ of\ the\ function\ }
								$(\bsl {g}\cdot\boldsymbol{\nu})''K_1,$ as in the second formula in (\ref{psi1}).
							
							Then the asymptotic formula
							% expansion in $h$ 
								\begin {equation}\label {dontcare2}
									A_1{\bsl{e}}_{\bsl{y}}(\bsl {g}^h)=h\bigl[(\bsl {g}\cdot\boldsymbol{\nu})''K_1\bigr]^h\boldsymbol{\sigma}^h+O(h^2),\ \ \ h\to0,
								\end {equation}
								holds on $(F_1^h\cap Q)\setminus\bigl(\cup_{k}B_{\delta_k}({\mathcal O}_k)\bigr).$
							%$$
							%%%\ \ \ \ \ \ K_1:=\bigl(A_1^{-1}\boldsymbol{\eta}\,\cdot\,\boldsymbol{\eta}\bigr)^{-1},\ \ \ \  \boldsymbol{\eta}:=\boldsymbol{\tau}\otimes\boldsymbol{\tau},
							%%%$$ 
							
							%{\rm\ to\ }F^h_1.$$
							%%%\item  The estimate $\Vert\bsl {g}^h-\bsl {g}\Vert_{[\lopen{Q}{\dmu^h}]^2}\le Ch$ holds with an $h$-independent constant $C>0.$ 
							%\ \ \ \text {in}\ 
	%%%						$\bigl[\lopen{Q}{\dmu^h}\bigr]^2.$
						\end {enumerate}
					\end {Prop}

			\subsection {Convergence of solutions}\label {sec233}
				%Consider the following theorem.
					\begin {Thrm}
					\label{homogenisation_theorem}
						For all $\varepsilon,$ $h,$ let $\ueph$ solve the integral identity (\ref {prob1}) with right-hand side $\bsl {f}=\bsl {f}^h_\varepsilon,$ 
						%for all $\varepsilon, h,$ 
						and suppose that 
						$h/\varepsilon\to\theta>0$ as $\ep\to0.$ %Then:
						%\begin {enumerate}
							%\item 
														If $\bsl {f}^h_\varepsilon\wtwo\bsl {f}$ then $\ueph\wtwo\bsl {u},$ and
							 %there is weak two-scale convergence of the the solution and the limit function 
							 $\bsl{u}$ satisfies
							 %the homogenised equation 
							 (\ref {hom21}).
							%\item 
							If $\bsl {f}^h_\varepsilon\stwo\bsl {f}$ then $\ueph\stwo\bsl {u}$
							%there is strong two-scale convergence of the the solution. A
							%dditionally, 
							and, in addition, there is convergence of the corresponding elastic energies.
						%\end {enumerate}
					\end {Thrm}
					\begin {proof}
				%	\begin {enumerate}
						%\item 
						%1) 
						%%%Consider the integral identity 
							%\begin {multline}
					%%%		\begin{equation}
					%%%		\label {intident01}
					%%%			\int_{\whep_1}A_1\bsl{e}(\ueph )\cdot \bsl{e}(\boldsymbol{\varphi} )\,\dmu_\ep^h +\ep^2\int_{\whep_0}A_0\bsl{e}(\ueph )\cdot \bsl{e}(\boldsymbol{\varphi} )\,\dmu_\ep^h + 
								%\\ +
						%%%		\int_\Omega \ueph \cdot \boldsymbol{\varphi}\,\dmu_\ep^h =\int_\Omega \bsl {f}^h_\varepsilon\cdot \boldsymbol{\varphi}\,\dmu_\ep^h\ \ \ \ \forall\varphi\in V.
					%%%		\end{equation}
							%\end {multline}
						%which is satisfied for all $\bphi \in V$.
						 						Setting $\bphi=\bphi_0(\bsl {x})$ in the identity (\ref{prob1}) and using 
						%(\ref{Ahomlim}) along with 
						Theorems \ref{u0_chi_theorem}, \ref{soft_conv_theorem}, we obtain
						%						taking the various two-scale limits yields
							%\begin {multline}
							\begin{equation}
							\label {intident02}
								\frac{1}{2}\int_\Omega A^\home_\lambda\bsl{e}(\bsl {u}_0):\bsl{e}(\bphi_0)\,\dx+ 
							%%%+\int_\Omega\int_Q A_0{\bsl{e}}_{\bsl{y}}(\bsl {u})\cdot {\bsl{e}}_{\bsl{y}}(\bphi_0)\,\dmu\dx + 
								%\\+
								\int_\Omega\int_Q \bsl {u}\cdot \bphi_0\,\dmu\dx =\int_\Omega\int_Q \bsl {f}\cdot \bphi_0\,\dmu\dx.
							\end{equation}
							%\end {multline}
						%%%As $\bphi_0$ does not depend on $\bsl {y},$ the second integral on the left drops out, hence 
						%by integration by parts, the following differential equation is obtained:
						%%%	$-\dive A^\home \bsl{e}(\bsl {u}_0)+\langle \bsl {u}\rangle= \bsl {f}.$
						
						%In order to obtain the rest of the homogenised equation, 
				%		Consider first the problem 
				%		\[
				%		{\rm div}\,\bsl{e}(\bsl{U}^h)={\rm div}\,\bsl{e}(\bsl{U})
				%		\]
								%the function $$
						%%%the function $\bsl {g}^h$ is the extension of $\bsl{g}$ constructed in Section \ref{sec232}, 
						Suppose that $\bsl{G}\in\bigl[C^\infty_{\rm per}(Q)\bigr]^2,$ $\bsl{g}\in D$
						%$\bsl{g}\in\widehat{\mathcal R}^0$ 
						are such that 
						$\bsl{G}(\bsl{y})=\bsl{g}(\bsl{y})$ for all
						%$\lambda$-a.e. 
						$\bsl{y}\in F_1\cap Q.$ 
						We approximate the function $\bsl{G}$ by a sequence 
						$\bsl{G}^h\in\bigl[C^\infty_{\rm per}(Q)\bigr]^2$ 
						such that 
						%$\bsl{e}(\widetilde{\bsl{G}}^h)=\bsl{e}(\tilde{\bsl{g}})$ 
						$\bsl{G}^h=\bsl{g}^h$ on $F_1^h\cap Q,$  
						where $\bsl{g}^h$ is the extension described in 
					        Proposition \ref{lemma12}. This is achieved, {\it e.g.}, by setting 
						$\bsl{G}^h=\bsl{G}\chi_h+\bar{\bsl{g}}^h(1-\chi_h),$ where 
						\[
						\bar{\bsl{g}}^h(\bsl{y}):=\left\{\begin{array}{ll}\bsl{0},\ \ \ \ \ \ \ \ \ \ \bsl{y}\in F_0^{2h}\cap Q,
						\\[0.5em]
						\bsl{g}^{h}(\bsl{y}),\ \ \ \ \bsl{y}\in\ F_1^{2h}\cap Q,\end{array}\right.			
						%\ \ \ \ \ 
						%\upsilon_h=\left\{\begin{array}{lll}1\ \ \ \ {\rm in}\ F_0^{2h}\cap Q,\\
						%0\ \ \ \ {\rm in}\ F_1^{h}\cap Q,\\
						%\bfn{g}^{2h}\ \ \ \ {\rm in}\ F_1^{2h}\cap Q,\end{array}
						\]
					       and $\chi_h$ is the convolution of 
						%This can be done, 
						%%%{\it e.g.} by setting $\widehat{\bsl{G}}^h=\bsl{G}$ on $F_0^{2h}\cap Q$ and smoothly joining\footnote{This can be done by considering, for each $h,$ the convolutions of 
						%%%\begin{equation}
						%%%\bsl{F}_1^h(\bsl{y}):=\begin {cases} 
					          %%%                                \bsl{G}(\bsl {y}),\ &\ \bsl {y}\in F_0^{2h}\cap Q,\\
					           %%%                               0,\ &\ \bsl{y}\in F_1^{2h}\cap Q,
			                            %%%                               \end {cases}
							%%%\ \ \ \ \ \ \ \ \ {\rm and}\ \ \ \ \ \ 
					             %%%                      	\bsl{F}_2^h(\bsl{y}):=
				        		      %%%                            \begin {cases}
						       %%%                          0,\ &\ \bsl {y}\in F_0^h\cap Q,\\
					                %%%                         \bsl{g}^h(\bsl {y}),\ &\ \bsl{y}\in F_1^h\cap Q,
						   	%%%\end {cases}
							%%%\end{equation}
							the characteristic function of the set $F_0^{3h/2}$
						 with a function $\upsilon(\cdot/h)$ such that $\upsilon\in C^\infty_0(\mathbb R^2),$ ${\rm supp}(\upsilon)\subset\{\bsl{z}\in{\mathbb R}^2: \vert\bsl{z}\vert\le1/4\}.$ 
						 \begin{Lem}
						 For the sequence $\bsl{G}^h$ constructed above, one has 
						 $\Vert\bsl{G}^h-\bsl{G}\Vert_{[H^1(Q)]^2}\to0$ as $h\to0.$
						 \end{Lem}
						 \begin{proof}
				Note first that since $\bsl{G},$ $\bsl{G}^h,$ $\bsl{g}^h$ %$\bsl{g}$ 
				are smooth and therefore their $L^2$-norms on $F_1^h$ and 
				$F_1^{2h}\setminus F_1^h$ are of order $O(h)$ as $h\to0,$ and in view of 
				the fact that $\bsl{G}^h=\bsl{G}$ on $F_0^{2h},$ one has 
				$\Vert\bsl{G}^h-\bsl{G}\Vert_{[L^2(Q)]^2}\to0$ as $h\to0.$
				
				Further, since $\bsl{G}^h-\bsl{G}=(\bar{\bsl{g}}^h-\bsl{G})(1-\chi_h)$ and  by the same argument as above one has 
				\[
				\bigl\Vert\bsl{e}(\bar{\bsl{g}}^h-\bsl{G})(1-\chi_h)\bigr\Vert_{[L^2(Q)]^3}\to0\ \ \ \ {\rm as}\ h\to0,
				\] 
				in order to estimate $L^2$-norm of $\bsl{e}(\bsl{G}^h-\bsl{G})$ it is sufficient to consider
				 $\bigl\Vert(\bar{\bsl{g}}^h-\bsl{G})\otimes\nabla\chi_h\bigr\Vert_{[L^2((F_1^{2h}\setminus F_1^h)\cap Q)]^3}.$ To this end, notice that $\nabla\chi_h=O(h^{-1}),$ and since $\bsl{G}=\bsl{g}^{h}$ on $F_1\cap Q$ one has 
				 $$
				 \bsl{G}(\bsl{y})=\bsl{g}^{h}(\bsl{y})+O(h),\ \ \ h\to0,\ \ \ \ \ \ \ \ \bsl{y}\in(F_1^{2h}\setminus F_1^h)\cap Q,$$ 
				 uniformly in $\bsl{y}.$
				%\in(F_1^{2h}\setminus F_1^h)\cap Q.$ 
				It follows that \[
				\bigl\Vert(\bar{\bsl{g}}^h-\bsl{G})\otimes\nabla\chi_h\bigr\Vert_{[L^2((F_1^{2h}\setminus F_1^h)\cap Q)]^3}\le Ch,\ \ \ \ \ \ \ C>0,
				\]
				%for any point $\bsl{y}'\in(F_1^{2h}\setminus F_1^h)\cap Q$ the estimate 
				%\[
			%\bigl\Vert(\bar{\bsl{g}}^h-\bsl{G})\otimes\nabla\chi_h\bigr\Vert_{[L^2((F_1^{2h}\setminus F_1^h)\cap Q)]^3}\le C\Bigl(h^{-1}\bigl\Vert\bsl{g}^{h}(\bsl{y}')-\bsl{g}^{h}(\cdot)\bigr\Vert_{[L^2((F_1^{2h}\setminus F_1^h)\cap Q)]^2}+h\bigl\vert\nabla\bsl{G}(\bsl{y}')\bigr\vert\Bigr),\ \ \ C>0,
			%	\]
%holds, 
from which the claim follows.
									 \end{proof}

				Taking in (\ref{prob1}) test functions 
				%of the form
							%\begin {equation}\label {testfun}
								$\bphi=\bphi^{\ep,h}=w\,\bsl {G}^h(\cdot/\varepsilon),$
							%\end {equation}
						where $w\in C_0^\infty(\Omega),$ 
%						satisfies the following conditions; if $\bsl {x}\in \whep_1$, $\bsl {g}^h$  is an extension function of the kind seen in Proposition \ref {lemma12}, for $\bsl {x}\in \whep_0$, $\bsl {g}^h$ is in $\bigl[L^2(\Omega,	\dmu_\ep^h)\bigr]^2$ and its limit $\bsl {g}$ belongs to $\bigl[H^1_\per(Q,\dmu)\bigr]^2$ and has a trace on the limit network $F_1$. 
						%The  limit function $w(\bsl {x})\bsl {g}(\bsl {y})$ will be denoted $\bsl {\Phi}\xy$. 
						%The identity (\ref{intident01}) becomes(\cdot/\ep)
						yields
							\begin {multline}
							\label {mess301}
								\ep^{-1}\int_{\whep_1}A_1\bsl{e}(\ueph ):{\bsl{e}}_{\bsl{y}}(\bsl {g}^h)(\cdot/\ep)w\,\dmu_\ep^h+ \int_{\whep_1}A_1\bsl{e}(\ueph ):\bigl(\bsl {g}^h(\cdot/\ep)\otimes \del w\bigr)\,\dmu_\ep^h
								\\
								+\ep\int_{\whep_0}A_0\bsl{e}(\ueph ):{\bsl{e}}_{\bsl{y}}(\bsl {G}^h)(\cdot/\ep)w\,\dmu_\ep^h
								\\
								+\ep^2\int_{\whep_0}A_0\bsl{e}(\ueph ):
								\bigl(\bsl{G}^h(\cdot/\ep)\otimes \del w\bigr)\,\dmu_\ep^h 
								=\int_\Omega (\bsl {f}^h_\varepsilon-\ueph)\cdot{\bsl {G}^h}(\cdot/\ep)w\,\dmu_\ep^h,\ \ \ \ \ \ \ \ 
							\end {multline}
						We denote the four terms on the left-hand side of (\ref {mess301}) by $I_j(\ep),$ $j=1,2,3,4.$
						%\dots,I_4(\ep).$  
						It follows from the $L^2$-boundedness of the sequence $\ep \bsl{e}(\ueph)$ and the fact that $A_1\bigl(\bsl{e}(\bsl {u}_0(\bsl{x}))+\bsl{v}\xy\bigr)$ is pointwise orthogonal to the matrix $\bsl {g}(\bsl{y})\otimes \del w(\bsl{x}),$ for $\bsl{y}\in F_1\cap Q,$ $\bsl{x}\in\Omega,$ (see \cite[Lemma 5.3]{bib31}) that the terms $I_4(\ep)$ and $I_2(\ep)$ converge to zero as $\ep\tends 0.$ 
						%Consider first the integral $I_3(\ep)$:
						%	\begin {equation}
						%		I_3(\ep):=\ep\int_{\whep_0}A_0\bsl{e}(\ueph )\cdot {\bsl{e}}_{\bsl{y}}(\bsl {g}^h )w\ \dmu_\ep^h.
						%	\end {equation}
						The convergence results on the soft component discussed in Section \ref {sec223} imply that
							$$\lime I_3(\ep)=\lim_{\varepsilon\to0}\frac {1}{2}\varepsilon\int_\Omega A_0{\bsl{e}}(\bsl {u}^h_\varepsilon):{\bsl{e}}_{\bsl{y}}(\bsl {G})(\cdot/\varepsilon)w\,\dmu_\ep^h+\lim_{\varepsilon\to0}\frac {1}{2}\varepsilon\int_\Omega A_0{\bsl{e}}(\bsl {u}^h_\varepsilon):{\bsl{e}}_{\bsl{y}}(\bsl{G}^h-\bsl {G})(\cdot/\varepsilon)w\,\dmu_\ep^h
							$$
							$$
							=\frac {1}{2}\int_\Omega\int_Q A_0{\bsl{e}}_{\bsl{y}}(\bsl {U} ): {\bsl{e}}_{\bsl{y}}(\bsl {G} )w\,\dy\dx.\ \ \ \ \ \ \ \ \ \ \ \ \ \ \ \ \ \ \ \ \ \ \ \ \ \ \ \ \ \ \ \ \ \ \ \ \ \ 
							$$
						%For the integral $I_1(\ep)$, 
						%%It follows from Proposition \ref{lemma12}, Lemma \ref{tech1} and Lemma (\ref{sig1})
						%the results on extension functions discussed in Section \ref {sec232} and various results seen in Appendix A 
The following statement is a consequence of \cite[Lemma 3.5]{bib32}.
\begin {Prop} 
						The  two-scale convergence
							\begin {equation}\label {sig1}
								\frac {h}{\ep} \chi_1^{h,\varepsilon}\bsl{e}(\ueph):\boldsymbol{\sigma}^h_\varepsilon\wtwo%\begin{cases}
								\frac {\theta^2}{3}\chi_1(\bsl{y})(\boldsymbol{\chi}\cdot\boldsymbol{\nu})''
								%_{\tau\tau}
								%\ \ {\rm on}\ F_1\
							\end {equation}
							holds, where $\boldsymbol{\sigma}^h$ is the function defined by (\ref{sigma_h}).
					%		in Proposition \ref{lemma12}.	
							\end {Prop}
							Taking into account the asymptotics (\ref{dontcare2}), it follows from the above proposition that the convergence
							\begin {equation}\label {techlem001}
								\lime I_1(\ep)=\frac {\theta^2}{6}\int_\Omega\int_QK_1\bsl {\boldsymbol{\chi}}''\cdot \bsl {g}''w\,\dl\dx
							\end {equation}
							holds. %%%Finally, we establish an approximation relationship for $\bsl{G},$ $\bsl{G}^h_\ep,$ which allows us to pass to the limit in the right-hand side of (\ref{mess301}).
					Finally, passing to the limit in (\ref{mess301}) as $\varepsilon\to0$ we obtain	
							%\begin {multline}
							\begin{equation}
							\label {hom698}
								\frac {\theta^2}{6}\int_\Omega\int_QK_1\bsl {\chi}''\cdot \bsl {g}''w\,\dl\dx+\frac {1}{2}\int_\Omega\int_Q A_0{\bsl{e}}_{\bsl{y}}(\bsl {u} ):{\bsl{e}}_{\bsl{y}}(\bsl {G} )w\,\dy\dx = 
								%\\ =
								\int_\Omega\int_Q (\bsl {f}- \bsl {u}) \cdot \bsl {G}w\,\dmu\dx,
							\end{equation}
							%\end {multline}
						Adding together the identities (\ref {intident02}) and (\ref {hom698}) and denoting 
						$\bphi\xy=\bphi_0(\bsl {x})+\bsl {\Phi}\xy,$
						% it follows that the weak two-scale convergence (\ref {ulimit1}) holds and the 			
						the homogenised formulation (\ref {hom21}) follows.

						%\item 
						In order to prove the strong convergence of solutions when $\bsl {f}^h_\varepsilon\stwo \bsl {f},$ consider another version of problem (\ref {prob1}) 
						%with it's solution denoted $\veph$ 
						with right-hand sides $\bsl {g}^h_\varepsilon\wtwo \bsl {g}:$
						%, i.e., consider the problem
							\begin {multline}\label {prob934}
								\veph\in\bigl[H_0^1(\Omega)\bigr]^2,\ \ \ \int_{\whep_1}A_1\bsl{e}(\veph ):\bsl{e}(\boldsymbol{\varphi} )\,\dmu_\ep^h +\ep^2\int_{\whep_0}A_0\bsl{e}(\veph): \bsl{e}(\boldsymbol{\varphi} )\,\dmu_\ep^h
								\\ +\int_\Omega \veph \cdot \boldsymbol{\varphi}\,\dmu_\ep^h =\int_\Omega \bsl {g}^h_\varepsilon\cdot\boldsymbol{\varphi}\,\dmu_\ep^h\ \ \ \ \forall\,\bphi\in\bigl[H_0^1(\Omega)\bigr]^2.
								%\bigl[C_0^\infty(\Omega)\bigr]^2.
							\end {multline}
						Setting $\bphi=\ueph$ in the above, $\bphi=\veph$ in the original problem (\ref {prob1}) with 
						$\bsl{f}=\bsl{f}_\varepsilon^h,$ and then subtracting one from the other yields
							\begin {equation}\label {GreenBetti}
								\lime\int_\Omega\ueph\cdot\bsl {g}^h_\varepsilon\,\dmu_\ep^h=\lime\int_\Omega\veph\cdot\bsl {f}^h_\varepsilon\,\dmu_\ep^h=\int_\Omega\int_Q\bsl {v}\cdot\bsl {f}\,\dmu\dx=\int_\Omega\int_Q\bsl {u}\cdot\bsl {g}\,\dmu\dx,
							\end {equation}
						%Assuming that $\bsl {g}^h_\varepsilon\wtwo \bsl {g}$, then it follows that
						%	$$\lime\int_\Omega\ueph\cdot\bsl {g}^h_\varepsilon\ \dmu_\ep^h=\int_\Omega\int_Q\bsl {v}\cdot\bsl {f}\ \dmu\dx=\int_\Omega\int_Q\bsl {u}\cdot\bsl {g}\ \dmu\dx.$$
						%This holds since 
						where %we used $\bsl {f}^h_\varepsilon\stwo \bsl {f}$, $\veph\wtwo \bsl {v}$ 
						where $\bsl {v}$ solves the homogenised equation with the right-hand side $\bsl{g}.$
						%%and equation (\ref {GreenBetti}) holds. Since $\bsl {g}^h_\varepsilon$ is arbitrary, the strong two-scale convergence $\ueph\stwo\bsl {u}$ is 						established.

						%To show the convergence of energies, m
						Finally, in order to show the convergence of energies, we set $\boldsymbol{\varphi}=\bsl{u}_\varepsilon^h$ in (\ref{prob1}) with $\bsl{f}=\bsl{f}_\varepsilon^h$ and use the definition of strong two-scale convergence as well as the identity (\ref{hom21}), as follows:
						%a standard two-scale convergence property (see {\it e.g.} \cite{bib38}):
							%\begin {multline*}
						\begin{equation*}
								\lime\left \{ \int_{\whep_1} A_1\bsl{e}(\ueph):\bsl{e}(\ueph)\,\dmu_\ep^h+\ep^2\int_{\whep_0}A_0\bsl{e}(\ueph):\bsl{e}(\ueph)\,\dmu_\ep^h\right \}=\int_\Omega\int_Q|\bsl {f}|^2\,\dmu\dx
								-\int_\Omega\int_Q|\bsl {u}|^2\,\dmu\dx
								\end{equation*}
								%\\ +
%											\end{equation*}	
							%\end {multline*}
%						Comparing the above with the homogenised equation (\ref {hom21}), it follows that
							%\begin {multline*}
							%	\lime\left \{ \int_{\whep_1} A_1\bsl{e}(\ueph)\cdot \bsl{e}(\ueph)\ \dmu_\ep^h+\ep^2\int_{\whep_0}A_0\bsl{e}(\ueph)\cdot \bsl{e}(\ueph)\ \dmu_\ep^h\right \} = 
								%\\ 
								\begin{equation*}
								=\frac{1}{2} \int_\Omega A^\home_\lambda\bsl{e}(\bsl {u}_0):\bsl{e}(\bsl {u}_0)\ \dx+\frac {\theta^2}{6}\int_\Omega\int_QK_1\boldsymbol{\chi}''\cdot \boldsymbol{\chi}''\ \dl\dx + 
								%\\ + 
								\frac {1}{2}\int_\Omega\int_QA_0{\bsl{e}}_{\bsl{y}}(\bsl {U}): {\bsl{e}}_{\bsl{y}}(\bsl {U})\,\dy\dx.
							\end {equation*}
						%%%as required.
					%\end {enumerate}
					\end {proof}

				%This concludes the proof of the main homogenisation result of this work. The next section investigates the spectral theory of the limit operator described in Section \ref {sec231}.

		\section {Convergence of spectra}\label {sec24}
			%The goal of the following section is to 
			Here we establish the convergence of the spectra of the operators associated with (\ref {prob1}) to the spectrum given by the limit problem (\ref {hom21}). 
			%This wis done in several stages; firstly a description of the spectrum of the limit 						problem will be given. This follows closely the analogous scalar calculation carried out by Zhikov \cite {bib31} with notable differences for the vector case presented here. Consideration will then be given to the spectrum of the $\ep$-dependent problem. 
			%%The strong two-scale resolvent convergence (Definition \ref {stworcon}) and the relevant extension theorem (Lemma \ref {extension32}) are key to this.
			%showing that there is convergence of the spectra in the sense of Hausdorff. 
			%The proof of the extension theorem for the current problem is not so simple since the intricate boundary conditions make choosing a suitable extension function more challenging. 
			%%In conclusion of  this section 
			%%%We then calculate the spectrum on a model network and compare it to the spectrum of the analogous problem without high-contrast.

			\subsection {Spectrum of the limit operator}\label {sec241}
					Consider the bilinear forms ({\it cf.} (\ref{hom21}))
					\begin {equation}\label {bilinmac}
						{\mathfrak b}_{\textnormal{macro}}(\bsl {u}_0,\bphi_0)=\frac{1}{2}\int_\Omega A^\home_\lambda\bsl{e}(\bsl {u}_0):\bsl{e}(\boldsymbol{\varphi}_0)\,\dx,\ \ \ \bsl{u}_0, \bphi_0\in\bigl[H_0^1(\Omega)\bigr]^2,
						%\bphi_0(\bsl {x}),
						\ \ \ \ \ \ \ \ \ \ \ \ \ \ \ 
					\end {equation}
					\begin {equation}\label {bmic1}
						{\mathfrak b}_{\textnormal {micro}}(\bsl {U},\bsl {\Phi})= \frac {\theta^2}{6}\int_QK_1\boldsymbol{\chi}''\cdot \bsl {\Phi}''\,\dl + \frac {1}{2}\int_QA_0{\bsl{e}}_{\bsl{y}}(\bsl {U}): {\bsl{e}}_{\bsl{y}}(\bsl {\Phi})\,\dy,\ \ \ 
						\bsl{U}, \bsl {\Phi}\in\widetilde{V},
						%\widetilde{V}:\{\bsl{U}: \bsl{u}_0+\bsl{U}\in V\ {\rm for\ some}\ \bsl{u}_0\}.
						%=\bsl {\Phi}(\bsl {y}).
					\end {equation}
					where the space $\widetilde{V}$ is defined in Section \ref{sec231}, see the paragraph preceding (\ref{other1}). 
				%Remark that both operators are symmetric and that the corresponding differential operator for the macroscopic bilinear form is self-adjoint.
The spectral problem associated with (\ref{hom21}) can be written in the form
					\begin {align}
					%\begin{align}
					%\begin {cases}
						&{\mathfrak b}_{\textnormal {macro}}(\bsl {u}_0,\bphi_0)=s\bigl( \bsl {u}_0+\langle \bsl {U}\rangle,\bphi_0\bigr)_{[L^2(\Omega)]^2}\ \ \ \ \forall\,\bphi_0\in\bigl[H_0^1(\Omega)\bigr]^2,\label {bilinear1a}
						\\[0.35em]
						%\\
						&{\mathfrak b}_{\textnormal {micro}}(\bsl{U},\bsl {\Phi})=s\bigl( \bsl {u}_0+\bsl {U},\bsl {\Phi}\bigr)_{[\lopen{Q}{\dmu}]^2}\ \ \ \ \ \ \ \ \ \ \forall\,\bsl {\Phi}
						%=\bsl {\Phi}
						\in\widetilde{V}.\nonumber
						%\label {bilinear1b}
					%\end {cases}
					%\end{align}
					\end {align}
	%			where $\bphi_0=\bphi_0(\bsl {x})$, and $\Bphi=\Bphi(\bsl {y})$. 
				%It is noted that there is no need to split the microscopic bilinear form into two bilinear forms; one on the interior of the unit cell $Q$ and the other on the limiting network $F_1$. In fact, if $\{\bvphi_n\}_{n\in\nbb}$ is a sequence of eigenfunctions for the bilinear form on the interior of $Q$ with eigenvalues $\{\omega_n\}_{n\in\nbb}$ then $\{\phi^{(\nu)}_n\}_{n\in\nbb}$ are the corresponding eigenfunctions for the bilinear form on $F_1$ with the same eigenvalues. This is as should be expected.

				%Consider firstly the spectral problem for the microscopic bilinear form in (\ref {bilinear1}). Since the induced operator is $\bsl {y}$-periodic on a bounded domain, it follows that it's spectrum will consist of a discrete set of eigenvalues (see Eastham \cite {bib103}). 
				Let $\{\bvphi_n\}_{n\in\nbb}\subset\widetilde{V}$ be an orthonormal set of eigenvectors with non-zero average for the bilinear form ${\mathfrak b}_{\textnormal {micro}}$ with corresponding set of eigenvalues $\{\omega_n\}_{n\in\nbb}:$
				%, i.e.,
					\begin {equation}\label {specmic}
						{\mathfrak b}_{\textnormal {micro}}(\bvphi_n,\bsl {\Phi})=\omega_n\left( \bvphi_n,\bsl {\Phi}\right)_{[\lopen{Q}{\dmu}]^2}\ \ \ \forall \bsl {\Phi}\in\widetilde{V}.
					\end {equation}
				%Note that $B_\mic$ may also have a set of eigenvalues $\{ \omega_n'\}_{n\in\nbb}$ corresponding to those eigenfunctions $\{\bsl {\phi}_n'\}_{n\in\nbb}$ with zero average. 
				%Using the above observations, t
				Assuming that the value $s$ is outside the spectrum ${\rm Sp}({\mathfrak b}_{\rm micro})$ of the form ${\mathfrak b}_{\rm micro},$ the  function $\bsl {U}\xy$ is written as a series in terms of eigenfunctions $\{\bvphi_n\}_{n\in\nbb}:$ 
					\begin{equation}
					\bsl {U}\xy
					%sb(\bsl {y},s)\bsl {u}_0(\bsl {x})
					%+\sum_{n=1}^\infty z_n(\bsl {x})\bvphi_n'(\bsl {y})
					%$$
					%$$\ \ \ \ \ \ \ \ \ \ \ \ \ \ \ \ \ \ \ \ \ \ \ \ \ \ 
					=s\sum_{n=1}^\infty\frac{\langle\boldsymbol{\phi}_n\rangle\cdot\bsl {u}_0(\bsl {x})}{\omega_n-s}\bvphi_n(\bsl {y}).
					\label{U_form}
					\end{equation}
					%+\sum_{n=1}^\infty z_n(\bsl {x})\bvphi_n'(\bsl {y}).$$
				%Note that no concern is given to the coefficients $z_n(\bsl {x})$. 
				Substituting this expansion for $\bsl {U}$ into (\ref {bilinear1a}), we obtain
				%. Doing so, after simplification, it follows that
					\begin {equation}\label {spec6}
						\mathfrak{b}_{\textnormal {macro}}(\bsl {u}_0,\bphi_0)=\bigl( \boldsymbol{\beta}(s)\bsl {u}_0,\bphi_0
						\bigr)_{[L^2(\Omega)]^2}\ \ \forall\,\bphi_0\in\bigl[H_0^1(\Omega)\bigr]^2,\ \ \ \ \ \ \ \ 
				%	\end {equation}
				%where $\boldsymbol{\beta}(s)$ is the matrix function defined by the relation 
					%\begin {equation}\label {spec7}
						\boldsymbol{\beta}(s):= s\left (I+s\sum_{n=1}^\infty \frac {\langle\boldsymbol{\phi}_n\rangle\otimes
						\langle\boldsymbol{\phi}_n\rangle}{\omega_n-s}\right ).
					\end {equation}
				Versions of the function $\boldsymbol{\beta}$ appear in the study of scalar \cite{bib31} and vector (\cite{Smyshlyaev_degeneracies}, \cite{bib66}, \cite{Zh_Past_Doklady}) homogenisation problems. The following statement is a straightforward modification of a result in \cite{Zh_Past_Doklady}. 
								%%In \cite {bib31}, a scalar function analogous to $\boldsymbol{\beta}(s)$ arose in the study of the spectrum of the scalar elasticity equation in the context of porous media in high-contrast. The matrix $\boldsymbol{\beta}(s)$ occurs in the literature also (see %Zhikov \& Pastukhova 
			%%	\cite{bib66} and 
				%by Smyshlyaev 
				%%\cite {Smyshlyaev_degeneracies}) in studies related to elasticity where there is high-contrast between the material inclusions. From these studies, the following result has been found:
					\begin {Prop}
					\label{limit_spectrum}
					Consider the operator $\mathfrak{A}$ whose domain consists of all solution pairs $(\bsl{u}_0, \bsl{U)}$
						for the identity
						\begin{equation}
						\mathfrak{b}_{\rm macro}(\bsl{u}_0,\boldsymbol{\varphi}_0)+\mathfrak{b}_{\rm micro}(\bsl{U},\boldsymbol{\Phi})
						=(\bsl{f}, \boldsymbol{\varphi}_0+\boldsymbol{\Phi})_{[L^2(\Omega\times Q, \dx\times\dmu)]^2}\ \ \ \ \ \ \forall\boldsymbol{\varphi}_0+\boldsymbol{\Phi}\in V,
						\label{Adef_id}
						\end{equation}
						as the right-hand side $\bsl{f}$ runs over all elements of $\mathfrak{H}$ and defined by $\bsl{f}=\mathfrak{A}(\bsl{u}_0+\bsl{U})$ if and only if (\ref{Adef_id}) holds. Then 
						%$s\in{\mathbb C}$ belongs to 
						the resolvent set $\rho(\mathfrak{A})$ of the operator $\mathfrak{A}$ is given by
%						if and only if $s\notin{\rm Sp}(\mathfrak{b}_{\textnormal {micro}})$ and the matrix $\boldsymbol{\beta}(s)$ is negative definite:
						%. This may be written as
							\begin{equation}
							\rho (\mathfrak{A})=\rho(\mathfrak{b}_{\textnormal {micro}})\cap\bigl\{s\,|\,{\rm all\ eigenvalues\ of\ }\boldsymbol{\beta}(s)\ {\rm belong\ to\ }\rho(\mathfrak{b}_{\rm macro})\bigr\},
							\label{spectrum_A}
							\end{equation}
							where $\rho(\mathfrak{b}_{\textnormal {micro}})$ is the resolvent set of the operator  generated by 
							the form $\mathfrak{b}_{\textnormal {micro}}$  in the closure\footnote{Note that the domain of this operator is dense in this closure.} of $\widetilde{V}$ in $\bigl[L^2(Q)\bigr]^2,$ and $\rho(\mathfrak{b}_{\textnormal{macro}})$ is the resolvent set of the operator generated by the form 
							$\mathfrak{b}_{\textnormal {macro}}.$
					\end {Prop}
					\begin {proof}
						Suppose that $s$ belongs to the right-hand side of (\ref{spectrum_A}). We argue that the problem
							\begin {equation}
							\begin {cases}
								\mathfrak{b}_{\textnormal {macro}}(\bsl {u}_0,\bphi_0)-s\bigl( \bsl {u}_0+\langle \bsl {U}\rangle,\bphi_0\bigr)_{[L^2(\Omega)]^2} =(\bsl {f},\bphi_0)_{[L^2(\Omega)]^2}\ \ \ \ \forall\boldsymbol{\varphi}_0,\label {specish1}
								\\[0.6em]
								\mathfrak{b}_{\textnormal {micro}}(\bsl {U},\bsl {\Phi})-s\bigl( \bsl {u}_0+\bsl {U},\bsl {\Phi}\bigr)_{[\lopen{Q}{\dmu}]^2} =(\bsl {f},\bsl {\Phi})_{[\lopen{Q}{\dmu}]^2}\ \ \ \ \forall\bsl{\Phi}.
							\end {cases}
							\end {equation}
						%We argue that the above problem 
						has a solution for every $\bsl {f}\in\mathfrak{H},
						%\bigl[L^2(\Omega)\bigr]^2\oplus\bigl[L^2(\Omega\times Q)\bigr]^2
						$
						%\bigl[\lopen{\Omega}{L^2(Q)}\bigr]^2$ 
						given that $s$ satisfies the required assumptions of the lemma. Since $s\notin{\rm Sp}(\mathfrak{b}_{\textnormal {micro}})$, it follows that $\bsl{U}$ can be written in the form (\ref{U_form}) with $\bsl{u}_0$ replaced by							
						%$\bsl {U}=b\left(
						$s\bsl {u}_0+\bsl {f}.$
						%\right ),$
						%where the matrix $b$ is given by (\ref {spec4}). 
						Substituting this into the first equation of (\ref {specish1}) yields
							\begin {equation}\label {specish2}
								\mathfrak{b}_{\textnormal {macro}}(\bsl {u}_0,\bphi_0)-\bigl( \boldsymbol{\beta}(s)\bsl {u}_0,\bphi_0\bigr)_{[L^2(\Omega)]^2}=\bigl(
								%(I+s\langle b\rangle)
								s^{-1}\boldsymbol{\beta}(s)\bsl {f},\bphi_0\bigr)_{[L^2(\Omega)]^2}\ \ \ \ \forall\boldsymbol{\varphi}_0.
							\end {equation}
						Since all eigenvalues of $\boldsymbol{\beta}(s)$ are in $\rho({\mathfrak{b}_{\rm macro}}),$ the operator induced by the bilinear form on the left-hand side of (\ref{specish2}) is invertible and thus the identity (\ref {specish2}) has a unique solution. 
						%Hence the result.

						Conversely, one has $\rho(\mathfrak{A})\subset\rho(\mathfrak{b}_{\textnormal {micro}})$ and if 
						$s\in\rho(\mathfrak{A})$ then $\boldsymbol{\beta}(s)$ has no eigenvalues in ${\rm Sp}({\mathfrak{b}_{\rm macro}}),$ 
						%is positive definite 
						for otherwise the problem (\ref {specish1}) would not be uniquely solvable for any $\bsl {f}\in\mathfrak{H}.$
						%%Therefore (\ref{specish2}) is solvable 
%							$\bsl {U}=b\left (s\bsl {u}_0+\bsl {f}\right),$
						%%and
							%\begin {equation}\label {specish3}
							%	B_{\textnormal {macro}}(\bsl {u}_0,\bphi_0)-\bigl\langle \boldsymbol{\beta}(s)\bsl {u}_0,\bphi_0\bigr\rangle_{\lopen{\Omega}{\dx}}=\bigl\langle
								%\bigl(I+s\langle b\rangle\bigr)
							%	s^{-1}\boldsymbol{\beta}(s)\bsl {f},\bphi_0\bigr\rangle_{\lopen{\Omega}{\dx}}.
							%\end {equation}
					%	Noting that 
						%$I+s\langle b\rangle=
						%$s^{-1}\boldsymbol{\beta}(s)>0$ implies that equation (\ref {specish3}) has a solution for every 
						%%all right-hand sides, which is a contradiction. 
						% However, this equation cannot have a solution for every right-hand side since $\boldsymbol{\beta}(s)>0$. This contradiction proves the result.
					\end {proof}
				In the case of the model framework, the matrix $\boldsymbol{\beta}$ is proportional to the identity matrix $I.$ Indeed, if the set $F_1\cap Q$ is invariant with respect to a rotation $\bsl{R},$ {\it i.e.} one has
				$
				\bigl\{\bsl{R}\bsl{y}: \bsl{y}\in F_1\cap Q\bigr\}=F_1\cap Q,
				$
				then for an eigenfunction $\boldsymbol{\phi}$ of the bilinear form $\mathfrak{b}_{\rm micro},$ the vector $\bsl{R}\boldsymbol{\phi}$ is an eigenvector with the same eigenvalue, hence one has $\bsl{R}\boldsymbol{\beta}(s)\bsl{R}^{-1}=\boldsymbol{\beta}(s),$ in view of the definition of $\boldsymbol{\beta},$ see (\ref{spec6}). Taking 
				$\bsl{R}$ to be the rotation through $\pi/2$ yields the required claim, namely $\boldsymbol{\beta}(s)=b(s)I$ for a scalar function $b.$
			Let $\{\gamma_n\}_{n\in\nbb}$ denote the increasing sequence of zeros of the function $b$ and let $\{\delta_n\}_{n\in\nbb}$ be the increasing sequence of all eigenvalues in the set $\{\omega_n\}_{n\in\nbb},$ counting multiple eigenvalues only once.
				%values for which $\beta_{11}(s)=0$. Let $\{\gamma_n\}_{n\in P}$ denote the increasing sequence of values for which $\det\boldsymbol{\beta}(s)=0$ where $P\subset \nbb_0$. Note that $\nu_0=\gamma_0=0$. 						Hence, provided $\gamma_n<\nu_n$ for all $n\in P$, 
				The spectrum of the limit operator $\mathfrak{A}$ has the ``band'' form:
					$${\rm Sp}(\mathfrak{A})=\biggl(\bigcup_{n\in\nbb}\bigl\{s\in(\gamma_n,\delta_{n}): 
					b(s)\in{\rm Sp}(\mathfrak{b}_{\rm macro})\bigr\}\biggr)\cup\{\delta_n\}_{n\in\nbb}\cup\{\alpha_n\}_{n\in\nbb},$$
					%\exists{\rm\ eigenvalue\ of\ }\boldsymbol{\beta}
					%\cup\biggl(\bigcup_{n\in\nbb_0\backslash P}[\nu_n,\omega_{n+1}]\biggr)
					 where $\alpha_n$ are the eigenvalues of $\mathfrak{b}_{\rm micro}$ such that all of the corresponding eigenfunctions have zero average over $Q.$ The intervals $(\delta_n,\gamma_{n+1}),$ $n\in{\mathbb N},$ are ``gaps'' in the spectrum, which do not have common points with 
					 ${\rm Sp}(\mathfrak{A}),$ except, possibly, for elements of the set 
					 $\{\alpha_n\}_{n\in\nbb}.$ 
					
				%Therefore, the intervals $(\gamma_n,\nu_n)$, $n\in P$ and the intervals $(\omega_n,\nu_n)$ $n\in\nbb_0\backslash P$ are gaps in the spectrum of the operator $\mathfrak{A}$ provided it does not contain a point from the set $\{\omega_1',\omega_2',\dots\}$ and 						$\gamma_n<\nu_n$. 
				%Hence, since there is an infinite number of spectral points $\{\omega_n\}_{n\in\nbb}$, there will be an infinite number of gaps in the spectrum.

				%Given the spectrum of the limit operator, the goal of the following section will be to establish the so called Hausdorff convergence of the spectra. In order to do so, a key result regarding strong two-scale resolvent convergence of the operator 
			%	$\mathcal {A}_\ep$ and a key result regarding the compactness of eigenfunctions of the operator $\mathcal {A}_\ep$ will be proven.

			\subsection {Proof of spectral convergence}\label {sec242}
				%The work in this section is devoted to the study of the spectral problem for the original problem (\ref {prob1}).
				 Here we show that the spectra of the original problems converge 
				 %in the sense of Hausdorff (see Definition \ref {haus1}) 
				 to the spectrum of the limit problem (\ref {hom21}). 
				 %Define the bilinear form $B_\ep(\cdot,\cdot)$ by the relation
				%	$$B_\ep(\ueph,\bphi)=\int_\Omega A^\ep \bsl{e}(\ueph )\cdot \bsl{e}(\boldsymbol{\varphi} )\ \dmu_\ep^h=\int_{\whep_1}A_1\bsl{e}(\ueph )\cdot \bsl{e}(\boldsymbol{\varphi} )\ \dmu_\ep^h +\ep^2\int_{\whep_0}A_0\bsl{e}(\ueph )\cdot \bsl{e}(\boldsymbol{\varphi} )\ \dmu_\ep^h.$$
				%The definition of Hausdorff convergence is now introduced.
					\begin {Def}
					%[Hausdorff Convergence]
					\label {haus1} We say that a sequence of sets ${\mathcal X}_\varepsilon\subset{\mathbb R},$ $\varepsilon>0,$ {\sl converges in the sense of Hausdorff} to 
					${\mathcal X}\subset{\mathbb R}$  
						%%%Let $\mathcal {A}_\ep$ be the operator and let $\mathcal {A}$ denote the limit operator of $\mathcal {A}_\ep$. Then the spectrum ${\rm Sp}(\mathcal {A}_\ep)$ is said to {\sl converge in the sense of Hausdorff} to the spectrum ${\rm Sp}(\mathcal {A})$ of a operator ${\mathcal A}$ 
						if the following two statements hold:
							%\begin {enumerate}[(H1)]
								%\item 
								
								(H1) For each $\omega\in{\mathcal X}$, there exists a sequence $\omega_\ep\in{\mathcal X}_\ep$ such that $\omega_\ep\tends \omega;$
								%%%$\omega\in{\rm Sp}(\mathcal {A})$, $\exists\,\omega_\ep\in{\rm Sp}(\mathcal {A}_\ep)$ such that $\omega_\ep\tends \omega;$
								%\item 
								
								(H2) For all sequences $\omega_\ep\in{\mathcal X}_\ep$ such that $\omega_\ep\tends \omega\in{\mathbb R}$, it follows that 
								$\omega\in{\mathcal X}.$
								%%%$\omega_\ep\in{\rm Sp}(\mathcal {A}_\ep)$ such that $\omega_\ep\tends \omega$, it follows that $\omega\in {\rm Sp}(\mathcal {A}).$
							%\end {enumerate}
					\end {Def}
				%It can be shown that the first property (H1)  holds for the operators under consideration if there is strong two-scale resolvent convergence $\mathcal {A}_\ep\stwor\mathcal {A}$. 
				%The definition of strong two-scale resolvent convergence is now given.
					\begin {Def}
					%[Strong Two-scale Resolvent Convergence]
					\label {stworcon}
						%Consider the following two problems:
					%		\begin {equation*}
								%\mathcal {A}_\ep\ueph=\bsl {f}^h_\ep,\ \ \ \mathcal {A}\bsl {u}=\bsl {f}.
					%		\end {equation*}
						 We say that a family of operators $\mathcal {A}_\ep$  in $\bigl[L^2(\Omega,\dmu_\ep^h)\bigr]^2$ {\sl strongly two-scale resolvent converges} as $\varepsilon\to0$ to an operator $\mathcal {A}$ in $\bigl[L^2(\Omega\times Q, \dx\times\dmu)\bigr]^2,$ and write $\mathcal {A}_\ep\stwor\mathcal {A},$ if  for all $\bsl{f}$ in the range 
						 $R({\mathcal A})$ 
						 of the operator ${\mathcal A}$ and for all sequences $\bsl {f}^h_\ep\in\bigl[L^2(\Omega,\dmu_\ep^h)\bigr]^2$ 
						 %the implication
						%\[
						such that $\bsl {f}^h_\ep\stwo\bsl {f},$ the two-scale convergence 
						%\ \ \ \implies\ \ \ 
						$(\mathcal {A}_\ep+I)^{-1}\,\bsl {f}^h_\ep\stwo (\mathcal {A}+I)^{-1}\,\bsl {f}$ holds.
						%\]
						%holds.
					\end {Def}
				%What follows next is an argument which illustrates the fact that given the strong two-scale resolvent convergence property then property (H1) of the Hausdorff convergence also holds. 
					\begin {Prop}
					\label{lower_semi}
						%Given strong two-scale resolvent convergence, 
						If $\mathcal {A}_\ep\stwor\mathcal {A},$ then
						the property 
						%of the Hausdorff convergence 
						(H1) holds with ${\mathcal X}_\ep={\rm Sp}({\mathcal A}_\ep),$ ${\mathcal X}={\rm Sp}({\mathcal A}).$
					\end {Prop}
					\begin{proof}
					Let $T_\ep:=(\mathcal {A}_\ep+I)^{-1}$ and $T:=(\mathcal {A}+I)^{-1}.$ If $s\in{\rm Sp}(\mathcal {A})$ then $t=(1+s)^{-1}\in{\rm Sp}(T)$. Therefore, for any $\delta>0$, there exists a vector $\bsl{f}\in R({\mathcal A})$
					%\in H\subset \lqw$ 
					such that 
					$$
					\|\bsl {f}\|_{[L^2(\Omega\times Q, \dx\times\dmu)]^2}=1,\ \ \ \ \ \ 
					%%%$ $
					\bigl\|(T-t)\bsl {f}\bigr\|_{[L^2(\Omega\times Q, \dx\times\dmu)]^2}\leq\delta/4.
					$$
				Consider a sequence $\bsl {f}^h_\ep\in\bigl[L^2(\Omega,\dmu_\ep^h)\bigr]^2$ such that $\bsl {f}^h_\ep\stwo \bsl {f}.$ 
				%and moreover such that $\lime \|\bsl {f}^h_\ep\|_{L^2(\Omega,\dmu_\ep^h)}=\|\bsl {f}\|_H$. 
				Using the definition of strong two-scale resolvent convergence, one has
					$$
					\lime\bigl\| (T_\ep-t)\bsl {f}^h_\ep\bigr\|_{[L^2(\Omega,\dmu_\ep^h)]^2}=\bigl\| (T-t)\bsl {f}\bigr\|_{[L^2(\Omega\times Q, \dx\times\dmu)]^2}\leq\delta/4.
					$$
				Hence, $\| (T_\ep-t)\bsl {f}^h_\ep\|_{[L^2(\Omega,\dmu_\ep^h)]^2}\leq\delta/2$ and $\|\bsl {f}^h_\ep\|_{[L^2(\Omega,\dmu_\ep^h)]^2}\geq1/2$ for sufficiently small $\ep.$ Therefore, the interval $(t-\delta,t+\delta)$ contains a point of the 					spectrum of the operator $T_\ep$. Moreover, every interval centered at $s$ contains a point of the spectrum of the operator $\mathcal {A}_\ep$ for small enough $\ep,$ which completes the proof. 
				%%%Hence we have proved the following result.
					\end{proof}
				%%It is well known in the literature on homogenisation that there is strong two-scale resolvent convergence of the elasticity operator given in equation (\ref {prob1}). The analogous scalar case is considered 
				%by Zhikov
				%%in \cite {bib31} and the current vector case in 
				%Zhikov \& Pastukhova 
				%%\cite {bib66} which in turn refers to 
				%the work by Zhikov 
				%%\cite{bib105}.
				
				\begin{Corol}
				For the operators $\mathfrak{A}_\varepsilon^h$ defined by the identity
				$$
				\mathfrak{B}^h_\ep(\bsl {u},\bsl {v})=\mathfrak{L}^h_\ep(\bsl{v}),
				$$
				where the forms $\mathfrak{B}^h_\ep,$ $\mathfrak{L}^h_\ep$ are defined by (\ref{forms}), $\bsl{f}=\mathfrak{A}_\varepsilon^h\bsl{u},$ and the operator 
				$\mathfrak{A}$ is defined in proposition \ref{limit_spectrum}, the property (H1) holds with 
				${\mathcal X}_\ep={\rm Sp}(\mathfrak{A}_\ep^h),$ ${\mathcal X}={\rm Sp}(\mathfrak{A}),$ $h=h(\varepsilon).$
				\end{Corol}
			
				The property (H2) of the Hausdorff convergence does not hold for spectra ${\rm Sp}(\mathfrak{A}^h_\varepsilon)$ in general, due to the fact that the soft component may have a non-empty intersection with the boundary of $\Omega.$ 
				%%In addition, sequences of eigenfunctions of ${\rm Sp}(\mathfrak{A}_\varepsilon)$ may converge to the eigenfunctions of the ``Bloch spectrum'' associated with the expression (\ref{bmic1})  
				%form $\mathfrak{b}_{\rm micro}$ on 
				%%$\varkappa$-quasiperiodic functions, $\varkappa\in[0,2\pi)^2.$   
				However, a suitable version of (H2) does hold for a modified operator family, where the corresponding elements of the soft component are replaced by the stiff material.  More precisely, for each $\varepsilon,$ $h,$ denote by $\widehat{\mathfrak{A}}_\varepsilon^h$ the operator defined similarly to $\mathfrak{A}_\ep^h,$ with $\Omega_0^{\varepsilon,h}$ and $\Omega_1^{\varepsilon,h}$ in (\ref{prob1}) replaced by $\widehat{\Omega}_0^{\varepsilon,h}$ and 
				$\Omega\setminus\widehat{\Omega}_0^{\varepsilon,h}.$
				%$\widehat{\Omega}_1^{\varesilon,h}.$ 
				Here, the set $\widehat{\Omega}_0^{h, \varepsilon}$ is the union of the sets $\varepsilon (F_0\cap Q^h+\bsl{n})$ over all $\bsl{n}\in{\mathbb Z}^2$ such that 
				$\varepsilon (Q+\bsl{n})\subset\Omega.$
%				$$ for the problem describ is established next.
				% i.e., assuming that $\omega_\ep\in{\rm Sp}(\mathcal {A}_\ep)$ and that $\omega_\ep\tends \omega$, the goal is to show that 									$\omega\in{\rm Sp}(\mathcal {A})$. Since ${\rm Sp}(B_{\textnormal {micro}})\subset{\rm Sp}(\mathcal {A})$, it will be assumed that $\omega\notin{\rm Sp}(B_{\textnormal {micro}})$. In order to proceed, a compactness result akin to that seen in Zhikov \cite {bib31} will be 						needed for the current problem. Hence the following lemma.
					\begin {Thrm}\label {extension32}
					%%%For all $\varkappa\in[0,2\pi),$ denote by $\widetilde{V}^\varkappa$ the space of functions $\bsl{U}(\bsl{y})={\rm e}^{{\rm i}\boldsymbol{\varkappa}\cdot\bsl{y}}\bsl{U}_\#(\bsl{y}),$ $\bsl{y}\in Q,$ 
					%%%such that
					%%%\[
					%%%\bsl{U}_\#\in\bigl[H^1_\per(Q)\bigr]^2,\\ \ \ \ \ \ \ \bsl{U}(\bsl{y})=\boldsymbol{\chi}(\bsl{y})\ \ \lambda{\text -}{\rm a.e.}\ \bsl{y}\in F_1,\ \ \ \ \ \ \boldsymbol{\chi}\in\widehat{\rigid}^0_\varkappa,
					%%%\]
					%%%where $\widehat{\rigid}^0_\varkappa$ is the set of $\varkappa$-quasiperiodic rigid displacements, defined analogously to $\widehat{\rigid}_0,$
					%%%see Definition \ref{rigidstuff}. Consider the bilinear form
					%%%\begin {equation*}
					%\label {bmic1}
						%%%{\mathfrak b}^\varkappa_{\textnormal {micro}}(\bsl {U},\bsl {\Phi})= \frac {\theta^2}{6}\int_QK_1\boldsymbol{\chi}''\cdot \bsl {\Phi}''\,\dl + \frac {1}{2}\int_QA_0{\bsl{e}}_{\bsl{y}}(\bsl {U})\cdot {\bsl{e}}_{\bsl{y}}(\bsl {\Phi})\,\dy,\ \ \ \ \ 
						%%%\bsl{U}, \bsl {\Phi}\in\widetilde{V}^\varkappa,
						%\widetilde{V}:\{\bsl{U}: \bsl{u}_0+\bsl{U}\in V\ {\rm for\ some}\ \bsl{u}_0\}.
						%=\bsl {\Phi}(\bsl {y}).
					%%%\end {equation*}
					
											Suppose that of all $\varepsilon,$ $h,$ the function $\ueph\in\bigl[H_0^1(\Omega)\bigr]^2$ is the 
						$L^2$-normalised eigenfunction of 
						%the operator 
						$\widehat{\mathfrak{A}}_\ep:$
							\begin{equation}
				\widehat{\mathfrak{A}}_\ep\ueph=\omega_\ep\ueph,\ \ \ \ \ \ \ \|\ueph\|_{[\lopen{\Omega}{\dmu_\ep^h}]^2}=1.
							\label{spectral_problem}
							\end{equation}
						If $\omega_\ep\tends \omega\notin
						%\bigcup_\varkappa
						{\rm Sp}(\mathfrak{b}_{\textnormal {micro}}),$ then the eigenfunction sequence $\{\ueph\}$ is compact with respect to strong two-scale convergence in $\bigl[L^2(\Omega, \dmu_\varepsilon^h)\bigr]^2.$
					\end {Thrm}
					\begin {proof}
						%Let $\ueph\in\bigl[H^1_0(\Omega)\bigr]^2$ and consider the following spectral problem: 
						%for all $\bphi$: 
							%\begin {multline*}
						%	\begin{equation*}
						%		\int_{\whep_1}A_1\bsl{e}(\ueph )\cdot \bsl{e}(\boldsymbol{\varphi} )\ \dmu_\ep^h +\ep^2\int_{\whep_0}A_0\bsl{e}(\ueph )\cdot \bsl{e}(\boldsymbol{\varphi} )\ \dmu_\ep^h =\omega_\ep\int_\Omega \ueph \cdot \boldsymbol{\varphi}\ \dmu_\ep^h\ \ \ \ \forall\varphi\in\bigl[H^1_0(\Omega)\bigr]^2.
						%		\end{equation*}
%							\end {multline*} 
						The eigenvalue problem (\ref{spectral_problem}) is understood in the sense of the identity
													\begin{equation}
													\label{proof10}
								\int_{\Omega\setminus\widehat{\Omega}_0^{\varepsilon,h}}A_1\bsl{e}(\ueph ):\bsl{e}(\bphi)\,\dmu_\ep^h+\ep^2
								\int_{\widehat{\Omega}^{\ep,h}_0}A_0\bsl{e}(\ueph ):\bsl{e}(\bphi)\,\dmu_\ep^h 
								%=\omega_\ep\int_\Omega |\ueph|^2\ \dmu_\ep^h
								=\omega_\ep\int_\Omega\ueph\cdot\bphi\,\dmu_\ep^h\ \ \ \ \ \ \ \forall\bphi\in\bigl[H_0^1(\Omega)\bigr]^2,
								\end{equation}
								which implies, in particular, that
							%\begin {multline*}
							\begin{equation*}
								\int_{\Omega\setminus\widehat{\Omega}_0^{\varepsilon,h}}A_1\bsl{e}(\ueph ):\bsl{e}(\ueph)\,\dmu_\ep^h +\ep^2\int_{\widehat{\Omega}^{\ep,h}_0}A_0\bsl{e}(\ueph ):\bsl{e}(\ueph)\,\dmu_\ep^h
								%=\omega_\ep\int_\Omega |\ueph|^2\ \dmu_\ep^h
								=\omega_\ep,
								\end{equation*}
								and hence $\bigl\|\bsl{e}(\bsl {u}_\ep^h)\bigr\|_{[L^2(\widehat{\Omega}^{\varepsilon,h}_1,\dmu_\ep^h)]^2}$ are uniformly bounded.
							%\end {multline*} 
						%By the work described in Jikov, Kozlov \& Olenik \cite {bib6}, Chapter III or the work described by Stein \cite {bib104}, Chapter VI, i
					%%%By a suitable version of an extension result from \cite[Lemma 4.1]{OShY}) taking into account the Korn inequality 
						%%%$$
					%%%	\Vert\bsl{u}\Vert_{[L^2(F_1\cap Q^h)]^2}\le\frac{\widetilde{c}}{h}\Vert\bsl{e}(\bsl{u})\Vert_{[L^2(F_1\cap Q^h)]^2},
					%%%	$$
						%{\it e.g.} \cite[Lemma 3.1]{bib6}, \cite[Chapter IV]{bib104}) 
					Denote by $\widehat{\Omega}_1^{h, \varepsilon}$ the union of $\varepsilon (F_1\cap Q^h+\bsl{n})$ over all $\bsl{n}\in{\mathbb Z}^2$ such that 
				$\varepsilon (Q+\bsl{n})\subset\Omega.$ We claim that for all $\varepsilon,$ $h,$ there exists 
						%an extension 
						$\widetilde {\bsl {u}}_\ep^h$ 
						%of $\ueph|_{\whep_1}$ 
						such that 
						%Denote this function 								$\widetilde {\bsl {u}}_\ep^h$. This extension has the following properties:
							\begin {equation}
								\bsl{e}(\ueph)
								%_{\whep_1}
								=\bsl{e}(\widetilde{\bsl {u}}_\ep^h)\ {\rm on}\ \widehat{\Omega}_1^{\ep,h},\ \ \ \ \ \ \ \widetilde {\bsl {u}}_\ep^h\in \bigl[H_0^1(\Omega)\bigr]^2,\ \ \ \ \ \ \ \ \ \ 
								\bigl\|\bsl{e}(\widetilde {\bsl {u}}_\ep^h)\bigr\|_{[L^2(\widehat{\Omega}^{\varepsilon,h}_0, \dmu_\ep^h)]^2}\leq 
								C\bigl\|\bsl{e}(\bsl {u}_\ep^h)\bigr\|_{[L^2(\widehat{\Omega}^{\varepsilon,h}_1,\dmu_\ep^h)]^2},
								%\le 2C\omega,
						        \label {bounded77}
						        \end{equation}
							%\end {equation}
						%Moreover, this extension can be chosen so that
						%	$$
							\begin{equation}	
							\label{proof2}
							\int_{\widehat{\Omega}_0^{\ep,h}}A_0\bsl{e}(\widetilde {\bsl {u}}_\ep^h):\bsl{e}(\bsl {\bphi} )\,\dmu_\ep^h=0\ \ \ \ \ \ \ \ \ \ 
							\forall\,\bphi\in[H_0^1(\Omega)\bigr]^2\ {\rm such\ that}\ \bsl{e}({\bphi})=0\ {\rm  in}\ \widehat{\Omega}_1^{\varepsilon,h},     
							\end{equation}
							where the constant $C>0$ is independent of $\varepsilon,$ $h.$
			Indeed, 
					%%%for each $h$ and $\bsl{n}\in{\mathbb Z}^2$ such that 
			%%%$\varepsilon (Q+\bsl{n})\subset\Omega,$  
			we can consider $\widetilde {\bsl {u}}_\ep^h$ 
			%%%such that
			%%%			the function $\bsl{W}_\ep^h(y;\bsl{n}):=\bsl{u}^h_\ep\bigl(\varepsilon(y+m)\bigr),$ 
			%%%$\bsl{y}\in Q$ and the solution $\bsl{v}=\widetilde{\bsl{W}}_\ep^h(y;\bsl{n})$
			such that $\zeph:=\ueph-\uepht$ solves the minimisation problem for the functional
		\begin{equation}
		\frac{1}{2}\int_{\widehat{\Omega}_0^{\varepsilon,h}}A_0\bsl{e}(\bsl{v}):\bsl{e}(\bsl{v})\,\dmu^h_\ep-
		\int_{\widehat{\Omega}_0^{h, \varepsilon}}A_0\bsl{e}(\bsl{u}_\ep^h):\bsl{e}(\bsl{v})\,\dmu_\ep^h
		%\ \mapsto\ \min,
		\label{min_prob}
		\end{equation} 			
		over all functions 
		%%%$\bsl{v}\in\bigl[H^1_{\rm per}(Q,\dmu^h)\bigr]^2$
		$\bsl{v}\in\bigl[H^1_0(\Omega)\bigr]^2$ 
		%%%whose restriction to $F_1^h\cap Q$
		whose restriction to $\widehat{\Omega}_1^{h, \varepsilon}$ 
		is a rigid-body motion with respect to the Lebesgue measure, {\it i.e.} 
		one has $\bsl{e}(\bsl{v})=0$ in $\widehat{\Omega}_1^{h, \varepsilon}.$
		%%%$F_1^h\cap Q.$ 
									%Now, consider $\zeph=\ueph-\uepht,$ where $\uepht$ is the function constructed above. 
						Clearly, one has $\bsl{e}(\zeph)=0$ in 
						$\widehat{\Omega}_1^{\varepsilon,h}$ and 
						%Hence, $\zeph$ 
						%which satisfies the identity
							\begin {equation*}
								\int_{\Omega\setminus(\widehat{\Omega}_0^{\varepsilon,h}\cup\widehat{\Omega}_1^{\varepsilon,h})}A_1\bsl{e}(\ueph ):\bsl{e}(\bphi)\,\dmu_\ep^h+\int_{\widehat{\Omega}_1^{\ep,h}}A_1\bsl{e}(\zeph ):\bsl{e}(\bphi )\,\dmu_\ep^h+
								\ep^2\int_{\widehat{\Omega}_0^{\ep,h}}A_0\bsl{e}(\zeph ): \bsl{e}(\bphi )\,\dmu_\ep^h\ \ \ \ \ \ \ \ \ \ \ \ \ \ \ \ \ \ \ \ \ \ \ 
								\end{equation*}
								\begin{equation}
								\label{compact.5}
								\ \ \ \ \ \ \ \ \ \ \ \ \ \ \ \ \ \ \ \ \ \ \ \ \ \ \ \ -\omega_\ep\int_{\Omega} \zeph\cdot\bphi\,\dmu_\ep^h=\omega_\ep\int_{\Omega} \uepht\cdot\bphi\,\dmu_\ep^h\ \ \ \ \ \ \ \forall\bphi\in[H_0^1(\Omega)\bigr]^2,\ \ \ \bsl{e}({\bphi})=0\ {\rm  in}\ \widehat{\Omega}_1^{\varepsilon,h}, 
							\end {equation}
							by combining (\ref{proof10}), (\ref{proof2}) and the Euler-Lagrange equation for (\ref{min_prob}).
						It follows from the bound (\ref{bounded77}) that $\uepht$ is compact 
						%in $\lopen {\Omega}{\dmu_\ep^h}$ 
						with respect to strong convergence in $\bigl[\lopen {\Omega}{\dmu_\ep^h}\bigr]^2,$ {\it i.e.} there exists 
						$\widetilde{\bsl{u}}=\widetilde {\bsl {u}}(\bsl{x})$ such that, up to selecting a subsequence, one has 
							%$$
							%\begin {equation}\label {compact.75}
							$\uepht\rightarrow\widetilde{\bsl {u}}$ in $\bigl[\lopen {\Omega}{\dmu_\ep^h}\bigr]^2.$
							%%%\ \ \ \ \ \ \ \ \ \ 
							%\ \ \ \text {in}\ \ ,\ \ \ 
							%.$$
						%It is also not too difficult to deduce that there is also weak two-scale convergence of the sequence $\chi_0^{\ep,h}\uepht$ where $\chi_0^{\ep,h}$ is the characteristic function on the domain $\whep_0$. Indeed, it follows that
			%%%\chi_0^{\ep,h}\uepht\stwo\chi_0(\bsl {y})\widetilde {\bsl {u}}(\bsl{x}),
			%=:\widetilde {\bsl {U}}\xy.
							%\end {equation}
							%%%where we the second convergence follows from Proposition 
							%\ref{strong_comp} and 
							%%%\ref{a_factor}.
					\begin{Lem} 
					\label{extra_prop}
					%Consider the space		 
						Suppose that for each $\varepsilon, h$ the function
						%$\bsl {f}_\ep^h\in$ and 
						%$\bigl[L^2(\Omega)\bigr]^2\ni
						$\bsl {f}_\ep^h$ belongs to the closure in $\bigl[L^2(\Omega)\bigr]^2$ of the set of smooth functions whose restrictions to 
						$\widehat{\Omega}_1^{\ep,h}$ are rigid-body motions with respect to the Lebesgue measure.
						Suppose also that $\bsl {f}_\ep^h\wtwo \bsl {f}\in\mathfrak{H},$ where the space $\mathfrak{H}$ is given in Definition \ref{V_definition}.
						% in Proposition \ref{limit_spectrum}. 
						
						%\bigl[L^2(\Omega\times Q, \dx\times\dmu)\bigr]^2.$
						%%%V,$ where the space $V$ is given by Definition \ref{V_definition}.
						%\subset L^2(\Omega)\oplus L^2(\Omega\times Q)$
						For all $\varepsilon, h,$ consider the function $\veph\in\bigl[H^1_0(\Omega)\bigr]^2$ such that $\bsl{e}(\veph)=0$ in 
						$\widehat{\Omega}_1^{\varepsilon,h}$ and the following resolvent identity 
						holds ({\it cf.} (\ref{compact.5})):
							\begin{equation*}
							\int_{\Omega\setminus(\widehat{\Omega}_0^{\varepsilon,h}\cup\widehat{\Omega}_1^{\varepsilon,h})}A_1\bsl{e}(\ueph )\cdot \bsl{e}(\bphi)\,\dmu_\ep^h+\int_{\widehat{\Omega}_1^{\ep,h}}A_1\bsl{e}(\bsl{v}^h_\ep)\cdot \bsl{e}(\bphi )\,\dmu_\ep^h+\ep^2\int_{\widehat{\Omega}_0^{h, \varepsilon}}A_0\bsl{e}(\veph)\cdot \bsl{e}(\bphi )\,\dmu_\ep^h\ \ \ \ \ \ \ \ \ \ \ \ \ \ \ \ \ \ \ \ \ \ \ 
							\end{equation*}
							\begin{equation}
							\ \ \ \ \ \ \ \ \ \ \ \ \ \ \ \ \ \ \ \ \ \ \ \ \ \ \ \ \ \ \ \ -\omega_\ep\int_{\Omega} \veph\cdot\bphi\,\dmu_\ep^h=\int_{\Omega} \bsl {f}_\ep^h\cdot\bphi\,\dmu_\ep^h\ \ \ \ \ \ \ \forall\bphi\in
							\bigl[H^1_0(\Omega)\bigr]^2,\ \ \ \bsl{e}({\bphi})=0\ {\rm  in}\ \widehat{\Omega}_1^{\varepsilon,h}.
							\label{v_id}
							\end{equation} 
					%	The following fact will now be used; 
						%If $\bsl {f}_\ep^h\wtwo \bsl {f}\in H\subset L^2(\Omega)\oplus L^2(\Omega\times Q)$ then 
						Then $\bsl {v}_\ep^h\wtwo\bsl {v}=\bsl {v}\xy
						%,$ 
						%$\bsl {v}
						\in\bigl[L^2(\Omega, \widetilde{V})\bigr]^2,$ and
							\begin {equation*}
								\frac {\theta^2}{6}\int_\Omega\int_QK_1\boldsymbol{\chi}''\cdot \bsl {\Phi}''\,\dl(\bsl{y})\dx+\frac{1}{2}\int_\Omega\int_QA_0{\bsl{e}}_{\bsl{y}}(\bsl {v})\cdot {\bsl{e}}_{\bsl{y}}(\bphi)\,\dy\dx-\omega\int_\Omega\int_Q\bsl {v}\cdot\bphi\,\dmu(\bsl{y})\dx\ \ \ \ \ \ \ \ \ \ \ \ \ \ \ \ \ \ \ \ \ \ \ \ \ \ \ \ \ \ \ \ 
								\end{equation*}
								\begin{equation}
								\label {compact2}
								\ \ \ \ \ \ =\int_\Omega\int_Q\bsl {f}\cdot\bphi\,\dmu(\bsl{y})\dx\ \ \ \ \forall\bphi\in\bigl[L^2(\Omega, \widetilde{V})\bigr]^2, \ \ \ \ \bphi(\bsl{x}, \bsl{y})=\bsl{\Phi}(\bsl{x}, \bsl{y})\ \ {\rm a.e}\ \bsl{x}\in\Omega,\ \ \lambda{\text -}{\rm a.e}\ \bsl{y}\in F_1\cap Q,
							\end {equation}
							where $\boldsymbol{\chi}(\bsl{x}, \cdot)$ is the trace of the function $\bsl{v}(\bsl{x}, \cdot)$ on $F_1\cap Q$ for a.e. $\bsl{x}\in\Omega.$
						%The set $X$ denoted above is defined as the space of functions $\bsl {v}\xy\in\bigl[L^2(\Omega\times Q)\bigr]^2$ such that
						%	$$\bsl {v}(\cdot,\bsl {y})\in\bigl[H_0^1(\Omega)\bigr]^2,\ \ \ \bsl {v}(\bsl {x},\cdot)\in\bigl[H_\per^1(Q)\bigr]^2,\ \ \ \ \ \underset{\bsl {y}\in F_1}{\tr}\bsl {v}\xy=\boldsymbol{\chi}\xy
							%,\ \boldsymbol{\chi}\xy
						%	\in\lopen{\Omega}{\widehat{\mathcal R}^0}.$$
					\end{Lem}
					
					\begin{proof}
					We show first that the spectra of the operators $\widehat{\mathfrak{A}}_\varepsilon^0$ defined via the bilinear forms ({\it cf.} (\ref{v_id}))
					\[
					\widehat{\mathfrak{b}}_\varepsilon^0(\bsl{v}, \bphi)=\int_{\Omega\setminus(\widehat{\Omega}_0^{\varepsilon,h}\cup\widehat{\Omega}_1^{\varepsilon,h})}A_1\bsl{e}(\ueph):\bsl{e}(\bphi)\,\dmu_\ep^h+\int_{\widehat{\Omega}_1^{\ep,h}}A_1\bsl{e}(\zeph):\bsl{e}(\bphi )\,\dmu_\ep^h\ \ \ \ \ \ \ \ \ \ \ \ \ \ \ \ \ \ \ \ \ \ 
					\]
					\[
					\ \ \ \ \ \ \ \ \ \ \ \ \ \ \ \ \ \ \ +\ep^2\int_{\widehat{\Omega}^{h, \varepsilon}}A_0\bsl{e}(\bsl{v}):\bsl{e}(\bphi )\,\dmu_\ep^h
					%-\omega_\ep\int_{\Omega}\bsl{v}\cdot\bphi\,\dmu_\ep^h=\int_{\Omega} \bsl {f}_\ep^h\cdot\bphi\,\dmu_\ep^h\ \ \ \ \ \ \ \forall
					\ \ \ \ \ \ \bsl{v}, \bphi\in\bigl[H^1_0(\Omega)\bigr]^2,\ \ \ \bsl{e}(\bsl{v}), \bsl{e}({\bphi})=0\ {\rm  in}\ \widehat{\Omega}_1^{\varepsilon,h},
							%\label{v_id} 
					\]
					converge, in the sense of Hausdorff as $\varepsilon\to0,$ to 
					%a subset of 
					${\rm Sp}(\mathfrak{b}_{\rm micro}).$ Indeed,the convergence 
					$\widehat{\mathfrak{A}}_\varepsilon^0\stwor\widehat{\mathfrak{A}}^0$ holds, where 
					the operator $\widehat{\mathfrak{A}}^0$ is associated with the bilinear form
					\[
					\widehat{\mathfrak{b}}^0(\bsl{v}, \bphi)=\frac {\theta^2}{6}\int_QK_1\boldsymbol{\chi}''\cdot \bsl {\Phi}''\,\dl+\frac{1}{2}\int_{F_0\cap Q}A_0\bsl{e}(\bsl{v}):\bsl{e}(\bphi )\,\dmu
					%-\omega_\ep\int_{\Omega}\bsl{v}\cdot\bphi\,\dmu_\ep^h=\int_{\Omega} \bsl {f}_\ep^h\cdot\bphi\,\dmu_\ep^h\ \ \ \ \ \ \ \forall
					\ \ \ \ \ \ \bsl{v},\bphi\in\bigl[L^2(\Omega,\widetilde{V})\bigr]^2,
					%%%\bigl[H^1_0(\Omega)\bigr]^2,\ \ \ \bsl{e}(\bsl{v}), \bsl{e}({\bphi})=0\ {\rm  in}\ \widehat{\Omega}_1^{\varepsilon,h}.
							%\label{v_id} 
					\]
					\[
					\bsl{v}(\bsl{y})=\boldsymbol{\chi}(\bsl{y}),\ \ \ \ \bphi(\bsl{y})=\bsl{\Phi}(\bsl{y}),\ \ \ \ \lambda{\text -}{\rm a.e}\ \bsl{y}\in F_1\cap  Q,
					\]
					and hence, ${\rm Sp}(\widehat{\mathfrak{A}}^0)\subset\lim_{\varepsilon}{\rm Sp}(\widehat{\mathfrak{A}}_\varepsilon^0)$ by 
					Proposition \ref{lower_semi}. On the other hand any sequence of $L^2$-normalised eigenfunctions of 
					$\widehat{\mathfrak{A}}_\varepsilon^0$ whose eigenvalues 
					$\omega^0_\varepsilon$ converge to $\omega^0\in{\mathbb R}$ is compact in the sense of two-scale convergence, thanks to 
					\cite[Theorem 12.2]{bib30}, and therefore $\omega\in{\rm Sp}(\widehat{\mathfrak{A}}^0).$ Finally, notice that 
					${\rm Sp}(\widehat{\mathfrak{A}}^0)=
					%{\rm Sp}(\widehat{\mathfrak{b}}^0)\subset
					{\rm Sp}(\mathfrak{b}_{\rm micro}).$

					It follows that whenever $\omega_\ep$ in (\ref{v_id}) converge to a point outside ${\rm Sp}(\mathfrak{b}_{\textnormal {micro}}),$
					%the set $\bigcup_\varkappa{\rm Sp}(\mathfrak{b}^\varkappa_{\textnormal {micro}}),$ 
					the identity (\ref{v_id}) does not have non-zero solutions $\bsl{v}_\ep^h$ for $\bsl{f}_\ep^h=0$ and 
					$\omega_\varepsilon$ is replaced by any value in some finite neighbourhood of the set $\{\omega_\varepsilon\}_{\varepsilon<\ep_0}$ for some $\ep_0>0.$ Hence, for an $L^2$-bounded sequence of the right-hand sides $\bsl{f}_\ep^h,$ the functions $\bsl{v}_\varepsilon^h$ that satisfy (\ref{v_id}) are uniformly bounded in $\bigl[L^2(\Omega, \dmu_\ep^h)\bigr]^2$ for $\ep<\ep_0.$ 
					
					Further, setting $\boldsymbol{\varphi}=\bsl{v}_\ep^h$ in (\ref{v_id}) and using the fact that $A_0$ is 
					%uniformly 
					positive definite yield the uniform estimate 
					\[
					\varepsilon\bigl\Vert\chi_0^{\ep,h}\bsl{e}(\bsl{v}_\varepsilon^h)\bigr\Vert_{[L^2(\Omega_0^{\ep,h}, \dmu_\ep^h)]^3}\le C,
					\]
					for some positive constant $C.$				
					%converges to the set  and since $\omega_\varepsilon$ converges 
					Proceeding as in Section \ref{sec22}, and using the fact that $\widehat{\Omega}_0^{h, \varepsilon}\cup\widehat{\Omega}_1^{h, \varepsilon}\to\Omega$ as $\ep\to0,$ we extract a subsequence of $\bsl{v}_\ep^h$ that weakly two-scale converges to a function $\bsl{v}\in \bigl[L^2(\Omega, \widetilde{V})\bigr]^2$ and such that 
					$\chi_0^{\ep,h}\bsl{e}(\bsl{v}_\ep^h)\wtwo\bsl{e}_{\bsl{y}}(\bsl{v})$ in $\bigl[L^2(\Omega,\dmu_\ep^h)\bigr]^3.$  
					
					Finally, passing to the limit as $\varepsilon\to0$ in (\ref{v_id}) yields the identity (\ref{compact2}). By the uniqueness of solution to 
					(\ref{v_id}), the whole sequence $\bsl{v}_\ep^h$ weakly two-scale converges to $\bsl{v}.$
					\end{proof}
					
				Lemma \ref{extra_prop} implies that the sequence $\zeph$ is compact with respect to weak two-scale convergence, its two-scale limit $\bsl {z}=\bsl {z}\xy$ is a rigid-body motion on $F_1$ and satisfies the weak problem
							\begin {equation*}
								\frac {\theta^2}{6}\int_\Omega\int_QK_1\boldsymbol{\upsilon}''\cdot \bsl {\Phi}''\,\dl(\bsl{y})\dx+\frac{1}{2}\int_\Omega\int_QA_0{\bsl{e}}_{\bsl{y}}(\bsl {z}):{\bsl{e}}_{\bsl{y}}(\bphi)\,\dy\dx-\omega\int_\Omega\int_Q\bsl {z}\cdot\bphi\,\dy\dx=\omega\int_\Omega\int_Q\widetilde {\bsl {u}}\cdot\bphi\,\dmu(\bsl{y})\dx
								\end{equation*}
								\begin{equation}
								\label {compact3}
								 \ \ \ \ \ \ \  \forall\bphi\in\bigl[L^2(\Omega, \widetilde{V})\bigr]^2,\ \ \ \ \ \bsl{z}(\bsl{x}, \bsl{y})=\boldsymbol{\upsilon}(\bsl{x}, \bsl{y}),\ \ \ \ \bphi(\bsl{x}, \bsl{y})=\bsl{\Phi}(\bsl{x}, \bsl{y}),\ \ \ \ {\rm a.e}\ \bsl{x}\in\Omega,\ \ \lambda{\text -}{\rm a.e}\ \bsl{y}\in F_1\cap Q,
							\end {equation}
						Setting $\bphi=\veph$ in the identiy (\ref {compact.5}) and %setting 
						$\bphi=\zeph$ 
						in (\ref {v_id}) yields
						% the following equality:
							\begin{equation}
							\int_\Omega\zeph\cdot\bsl {f}_\ep^h\,\dmu_\ep^h=\omega_\ep\int_\Omega\veph\cdot
							%\bigl(\chi_0^{\ep,h}
							\uepht
							%\bigr)
							\,\dmu_\ep^h\ \ \ \ \ \forall\varepsilon, h.
							\label{h_eps_id}
							\end{equation}
						Taking the limit of both sides (\ref{h_eps_id}) as $\ep\tends 0,$ $h=h(\ep),$ and using the convergence properties of 
						$\bsl{v}_\ep^h,$ $\bsl{u}_\ep^h,$ we obtain
							\begin{equation}
							\label{lim}
							%\int_\Omega\int_{Q}\bsl {z}\xy\cdot\bsl {f}\xy\,\dmu(\bsl{y})\dx
							\lim_{\varepsilon\to0}\int_\Omega\zeph\cdot\bsl {f}_\ep^h\,\dmu_\ep^h=\omega\int_\Omega\int_{Q}\bsl {v}\xy\cdot\widetilde {\bsl {u}}(\bsl{x})\,\dmu(\bsl{y})\dx.
							\end{equation}
						Further, using (\ref{compact2}) with 
						%$\bsl{f}=\bsl{z},$ 
						$\boldsymbol{\varphi}=\bsl{z},$ and (\ref{compact3}) with $\boldsymbol{\varphi}=\bsl{v},$ we obtain
						\begin{equation}
						\omega\int_\Omega\int_{Q}\bsl {v}\xy\cdot\widetilde {\bsl {u}}(\bsl{x})\,\dmu(\bsl{y})\dx=
						\int_\Omega\int_{Q}\bsl {f}\xy\cdot\bsl{z}\xy\,\dmu(\bsl{y})\dx
						\label{id}
						\end{equation}
						Finally, setting $\bsl {f}_\ep^h=\zeph$ in (\ref{lim}) and using (\ref{id}), we infer that
						%the convergence 
						$
						\Vert\zeph\Vert_{[L^2(\Omega,\dmu_\ep^h)]^2}\to\Vert\bsl{z}\Vert_{[L^2(\Omega\times Q, \dx\times\dmu)]^2}.
						$
							%$$\lime\int_\Omega|\zeph|^2\dmu_\ep^h=\int_\Omega\int_{Q}|\bsl {z}\xy|^2\dy\dx,$$
							 Therefore, the sequence 
						%there is strong two-scale convergence 
						$\zeph$ strongly two-scale converges to
						%\stwo\bsl 
						$\bsl{z},$ see Proposition \ref{strong_conv}. 
						%and hence the result follows. %Remark the property used to prove this lemma follows by the work already discussed in this chapter.
					\end {proof}
%				With this result, it immediately follows that the second property (H2) of the Hausdorff convergence follows and hence there is spectral convergence.

				%The next section will see the spectrum for the microscopic operator calculated in the case when the model geometry is considered.

%\begin {subappendices}

%		\section*{Appendic}
		%\addcontentsline{toc}{chapter}{Appendices}		

%\section*{Appendix A: Derivation of the equations (\ref{other1})--(\ref{wentzel2}).}	

\subsection*{Acknowledgements} This work was carried out under the financial support of
%the Leverhulme Trust (Grant RPG--167 ``Dissipative and non-self-adjoint problems'') and
the Engineering and Physical Sciences Research Council (Grant EP/I018662/1 ``The mathematical analysis and applications of a new class of high-contrast phononic band-gap composite media''; Grant EP/L018802/2 ``Mathematical foundations of metamaterials: homogenisation, dissipation and operator theory'').  We are grateful to  Dr Mikhail Cherdantsev for a thorough revision of the manuscript and a number of valuable comments, and to Professor Svetlana Pastukhova for her advice at the early stage of the work.

\end{document}